\documentclass[11pt]{article}
\usepackage{amsmath, amssymb, amsthm, mathrsfs, bm, geometry, mathtools, hyperref, booktabs, graphicx}
\usepackage{natbib}
\usepackage{algorithm}
\usepackage{algpseudocode}
\usepackage{fancyhdr}
\usepackage{tikz}
\usetikzlibrary{positioning, arrows.meta, shapes.geometric, decorations.pathmorphing}

\newtheorem{theorem}{Theorem}[section]

\newtheorem{definition}[theorem]{Definition}

\newcommand{\Var}{\mathrm{Var}}

\newcommand{\Prob}{\mathbb{P}}

\newcommand{\KL}[2]{D_{\text{KL}}\left(#1 \parallel #2\right)}
\newcommand{\dd}{\mathrm{d}}

\newcommand{\btheta}{\boldsymbol{\theta}}

\newcommand{\bA}{\boldsymbol{A}}

\title{\textbf{The Bayesian Reflex: Online Learning as the Autonomic Nervous System of Modern and Future AI}}
\author{Durba Bhattacharya$^{1}$, Sucharita Roy$^{2}$, and Sourabh Bhattacharya$^{3}$\\
$^{1}$Department of Statistics, St. Xavier's College (Autonomous), Kolkata\\
$^{2}$Department of Mathematics, St. Xavier's College (Autonomous), Kolkata\\
$^{3}$Interdisciplinary Statistical Research Unit, Indian Statistical Institute, Kolkata}

\date{}

\begin{document}

\maketitle

\begin{abstract}
This chapter introduces the concept of the \textit{Bayesian reflex} as a unifying framework for understanding online learning in artificial intelligence. Drawing an analogy with the autonomic nervous system that unconsciously regulates vital physiological functions, we conceptualize Bayesian online algorithms as the continuous, automatic learning mechanism that enables AI systems to maintain equilibrium in dynamic environments. The Bayesian reflex operates through three fundamental mechanisms: belief maintenance via probabilistic representations, sequential updating through Bayes' theorem, and uncertainty-driven action that balances exploration and exploitation.

We provide a comprehensive mathematical treatment of sequential Bayesian inference and survey the methodological landscape from Kalman filters to particle filters and variational methods. A central contribution is the recognition that the computational foundation of the Bayesian reflex lies in two complementary principles: the \textit{look-up table principle} for sequential inference in function space, and the \textit{ellipsoidal decomposition framework} for perfect iid sampling from arbitrary posterior distributions. We trace how these principles have been generalised across an increasingly broad range of problems—from dynamic computer model emulation to nonparametric state-space models, circular time series, inverse regression for climate model evaluation, and deep architectures through Recursive Gaussian Processes.

The decision-making aspects of the Bayesian reflex are examined through bandit theory, with detailed treatments of Thompson sampling and restless bandits. We also explore remarkable extensions to fundamental mathematical problems: assessing convergence of infinite series (with applications to climate dynamics and the Riemann Hypothesis), modelling the distribution of prime numbers (leading to the discovery of 184 strong Mersenne prime candidates), detecting stationarity in time series, characterising point processes, and determining frequencies in oscillatory signals. This remarkable versatility across seemingly disparate domains underscores the unifying power of the Bayesian reflex as the essential infrastructure for building truly adaptive artificial intelligence—systems that do not merely learn once, but continuously evolve in harmony with the complex, dynamic world they inhabit.
\end{abstract}

\vspace{0.5cm}

\newpage

\tableofcontents

\newpage

\section{Introduction: The Imperative for Online Adaptation}

Contemporary artificial intelligence has achieved remarkable successes through large-scale batch training. Models such as large language models (LLMs), computer vision systems, and recommendation engines are typically trained on massive static datasets, optimised through millions of gradient updates, and then deployed in fixed configurations. This paradigm, while powerful, embodies a fundamental limitation: the world is not static. The assumption that the environment remains unchanged after deployment is increasingly untenable in real-world applications where data streams continuously, user preferences evolve, and system dynamics shift \citep{hazan2016introduction, lattimore2020bandit, orabona2019modern}.

Consider the following scenarios that illustrate this limitation. A recommendation system must adapt to shifting user preferences, seasonal trends, and emerging content, requiring continuous learning rather than periodic retraining \citep{dasgupta2025bayesian}. A maternal health program must learn optimal call times for millions of beneficiaries, where individual preferences vary due to work schedules, phone sharing patterns, and network reliability, and these preferences themselves change over time \citep{liang2025context}. A sepsis prediction system in an intensive care unit must continuously update its risk assessments as new patient data arrives every minute, where delays in adaptation could have life-or-death consequences \citep{zhou2025sepsyn}. An autonomous vehicle navigating unfamiliar terrain must make real-time decisions under profound uncertainty, integrating sensor data streams and adapting to novel situations without human intervention \citep{kalman1960new}. Even in function optimisation—a classical problem in numerical analysis—sequential strategies that adaptively select new points to evaluate based on previous observations embody the same principles of online learning \citep{roy2021function, snoek2012practical}. Remarkably, the same principles extend to evaluating climate model projections, where inverse regression formulations ask whether observed past temperatures are plausible given assumed future scenarios \citep{chatterjee2022ominous}. Recursive Bayesian methods have also been developed to assess the convergence of infinite series—a fundamental problem in mathematical analysis—with applications ranging from climate change to the Riemann Hypothesis \citep{roy2020bayesian, roy2020bayesian2}. Most astonishingly, these same recursive Bayesian principles have been applied to model the distribution of prime numbers, leading to insights that challenge the Riemann Hypothesis and to the discovery of 259 new primes exceeding 140 million, including 184 strong Mersenne prime candidates with potential digit lengths reaching 242 million \citep{roy2025prime}. In each of these scenarios, the system faces a common challenge: it must learn and adapt while operating, integrating new information sequentially without the luxury of re-training on entire accumulated datasets. This is the domain of online learning.

The recursive Bayesian framework has proven to be remarkably versatile, extending far beyond these applications to encompass a broad range of statistical problems. As demonstrated by \citet{roy2021stationarity}, the same core methodology—partitioning the index set into locally stationary regions, constructing binary indicators based on empirical distribution functions, and applying recursive Beta-Binomial updates—can be applied to detect stationarity in time series models, assess convergence of Markov chain Monte Carlo algorithms, test for spatial and spatio-temporal stationarity, characterise point processes, and determine frequencies in oscillatory signals. This unification of seemingly disparate problems under a common methodological framework is a testament to the power of the Bayesian reflex.

However, a fundamental challenge underlies all these applications: how can we obtain (practically) exact samples from the complex, high-dimensional posterior distributions that arise in sequential Bayesian updating? Traditional approaches rely on Markov chain Monte Carlo (MCMC) methods, which suffer from well-known convergence assessment difficulties \citep{robert2004monte}. Even the most experienced practitioners often resort to deterministic or ad-hoc approximations, considering the challenges of diagnosing convergence insurmountable \citep{bhattacharya2025iid}. This is where the concept of \textit{perfect simulation}—generating (practically) exact independent and identically distributed (iid) samples from the target distribution—becomes essential.

\begin{figure}[htbp]
\centering
\begin{tikzpicture}[
    block/.style={
        rectangle, draw, fill=blue!20, text width=3.2cm, text centered,
        rounded corners, minimum height=1.2cm, font=\small
    },
    env/.style={
        rectangle, draw, fill=gray!20, text width=3.0cm, text centered,
        rounded corners, minimum height=1.0cm, font=\small
    },
    arrow/.style={thick, ->, >=stealth},
    reflex/.style={draw=red!70, dashed, thick, inner sep=0.3cm},
    lab/.style={font=\footnotesize, fill=white, inner sep=1pt}
]

% Analogy box at top
\node (ans) [draw, fill=green!20, text width=6cm, text centered, dashed,
             rounded corners, font=\small] at (0, 5.8)
    {Analogy: Autonomic Nervous System \\
     Unconscious regulation of heartbeat, breathing, homeostasis};

% Dashed ellipse for Bayesian Reflex
\node (reflex) [reflex, ellipse, minimum width=11cm, minimum height=7cm,
                anchor=center] at (0, 0) {};

% Reflex title below the ellipse
\node (title) [font=\small\bfseries, align=center] at (0, -3.7)
    {Bayesian Reflex\\ \normalfont\small Automatic, Continuous Learning};

% Core mechanism nodes inside the ellipse
\node (seq)  [block] at (0, 2.8)     {Sequential Updating \\ Bayes' Theorem};
\node (bel)  [block] at (-4.5, -0.2) {Belief Maintenance \\ Probabilistic Representation};
\node (act)  [block] at ( 4.5, -0.2) {Uncertainty-Driven Action \\ Explore vs.\ Exploit};

% Environment node outside, below the ellipse
\node (env)  [env]   at (0, -4.8)    {Environment \\ (Outcome / Reward)};

% ---- Arrows (now with prior belief flow) ----

% 1. New observation: env -> seq
\draw[arrow] (env.north) -- (seq.south)
    node[midway, right, lab] {new observation};

% 2. Prior belief: bel -> seq (the missing piece)
\draw[arrow] (bel.north east) to[out=90, in=-150] (seq.south west)
    node[midway, above, sloped, lab] {prior belief};

% 3. Updated belief: seq -> bel
\draw[arrow] (seq.west) -- (bel.north)
    node[pos=0.45, above, lab] {updated belief};

% 4. Posterior for action selection: bel -> act
\draw[arrow] (bel.east) -- (act.west)
    node[midway, below, lab] {posterior};

% 5. Choose action: act -> env
\draw[arrow] (act.south) |- (env.east)
    node[pos=0.2, right, lab] {choose action};

% 6. Analogy -> Reflex
\draw[arrow] (ans.south) -- (0, 3.8);

% Loop label
\node at (0, -6.0) [font=\small\itshape] {Continuous Learning Loop};

\end{tikzpicture}
\caption{Bayesian Reflex diagram with complete information flow. New observations from the environment and the prior belief (the current posterior) are combined via Bayes' theorem; the resulting posterior becomes the new belief; actions are chosen based on that posterior; outcomes feed back into the environment, closing the loop.}
\label{fig:bayesian_reflex}
\end{figure}
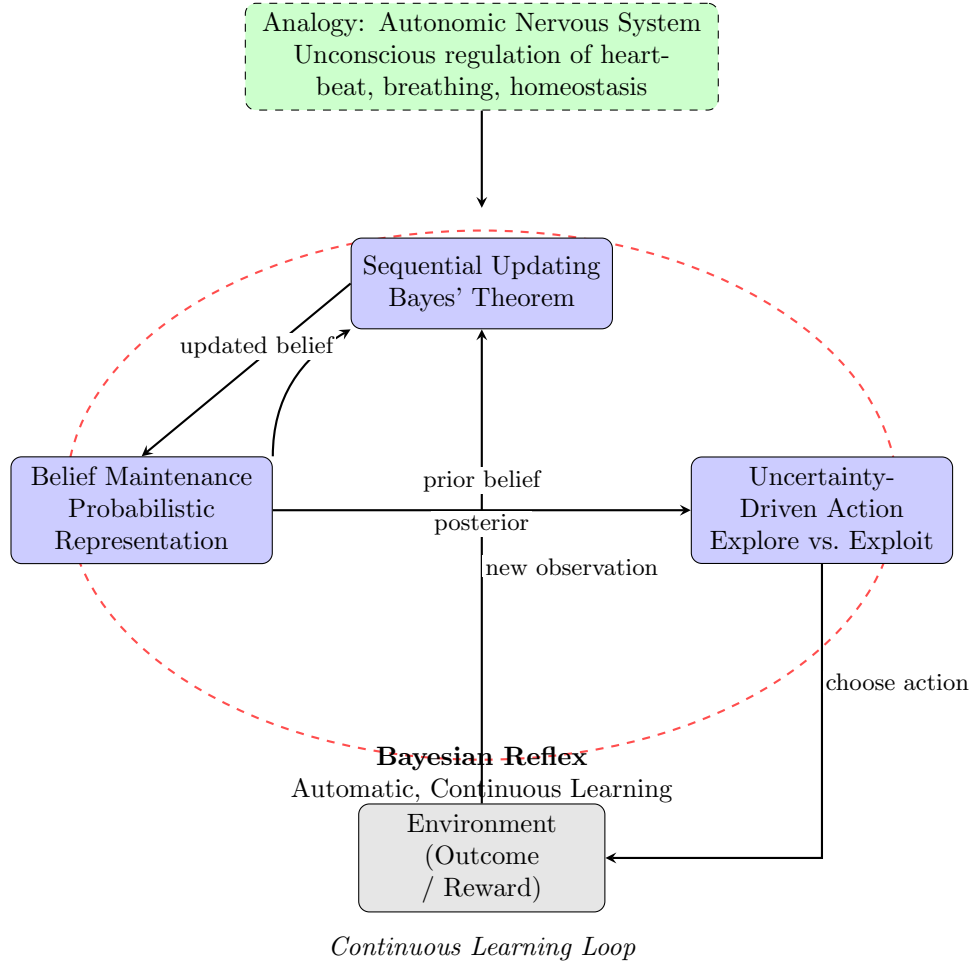

We introduce the concept of the \textit{Bayesian reflex} as an organising principle for understanding such online adaptive systems. Just as the autonomic nervous system unconsciously and continuously regulates heartbeat, breathing, and homeostasis in biological organisms (Figure \ref{fig:bayesian_reflex}), Bayesian online algorithms provide the continuous, low-level intelligence that enables artificial systems to maintain equilibrium in dynamic environments. The term reflex is chosen deliberately to emphasise the automatic, immediate, and unconscious nature of these updates—they happen without explicit intervention, forming the bedrock upon which higher-level cognitive functions can be built.

The Bayesian reflex operates through three fundamental mechanisms that together enable intelligent adaptation. First, \textit{belief maintenance} means the system continuously maintains a probabilistic representation of its knowledge about relevant quantities of interest, taking the form of a posterior distribution \citep{berger2013statistical, bernardo2009bayesian}. This representation captures not only point estimates but also the uncertainty associated with those estimates, providing a rich characterisation of what is known and what remains uncertain. Second, \textit{sequential updating} dictates that upon receiving new observations, the system updates its beliefs according to Bayes' theorem, optimally integrating prior knowledge with new evidence \citep{gelman2013bayesian, murphy2023probabilistic}. This update is mathematically principled and ensures that information is accumulated efficiently over time. Third, \textit{uncertainty-driven action} ensures that decisions are guided not only by point estimates but by full uncertainty quantification, naturally balancing exploration and exploitation \citep{thompson1933likelihood, sutton2018reinforcement}. When uncertainty is high, the system explores to gather information; when uncertainty is low, it exploits its knowledge to maximise performance.

This chapter provides a treatment of the Bayesian reflex, from its mathematical foundations to its state-of-the-art implementations in modern AI systems. A unifying theme throughout is the interplay between two foundational computational principles: the \textit{look-up table principle}—a powerful strategy for sequential inference in function space, and the \textit{ellipsoidal decomposition principle}—a complementary approach for practically exact iid sampling in parameter space. Together, these principles provide the computational infrastructure that makes the Bayesian reflex practically realisable.

We argue that Bayesian online algorithms are not merely a specialised toolkit for niche applications but rather constitute essential infrastructure for building truly adaptive artificial intelligence. The principles we develop here have implications across the entire spectrum of AI applications, from healthcare to robotics to language models, and point toward a future where AI systems continuously learn and evolve throughout their operational lifetime.

\section{Mathematical Foundations of Sequential Bayesian Inference}

\subsection{Bayes' Theorem and the Sequential Update}

At its core, Bayesian inference provides a mathematical framework for reasoning under uncertainty that is both elegant and powerful. The fundamental relation known as Bayes' theorem expresses how prior beliefs should be updated in light of observed data, serving as the mathematical foundation for the Bayesian reflex. Formally, if we denote by \(\theta\) the unknown quantities of interest and by \(\mathcal{D}\) the observed data, Bayes' theorem states that 
\begin{equation}
p(\theta \mid \mathcal{D}) = \frac{p(\mathcal{D} \mid \theta) p(\theta)}{p(\mathcal{D})}.
\label{eq:bayes}
\end{equation}
Here, \(p(\theta)\) is the prior distribution encoding beliefs before observation, representing our initial state of knowledge or ignorance about \(\theta\) \citep{robert2007bayesian}. The term \(p(\mathcal{D} \mid \theta)\) is the likelihood, which describes the probabilistic mechanism by which data are generated given the parameters, capturing our modelling assumptions about the world. The denominator \(p(\mathcal{D}) = \int p(\mathcal{D} \mid \theta) p(\theta) \, d\theta\) is the marginal likelihood or evidence, which serves as a normalising constant ensuring that the posterior integrates to one. The result \(p(\theta \mid \mathcal{D})\) is the posterior distribution representing updated beliefs after observing the data, synthesising prior knowledge with empirical evidence.

In the online setting, data arrives sequentially over time: \(\mathcal{D}_t = \{x_1, x_2, \ldots, x_t\}\). This sequential nature fundamentally changes the inference problem, as we cannot simply wait for all data to arrive before computing the posterior. Instead, we must update our beliefs incrementally as each new observation becomes available. The Bayesian update then becomes a sequential process that mirrors the operation of the Bayesian reflex:
\begin{equation*}
p(\theta \mid \mathcal{D}_t) \propto p(x_t \mid \theta) \cdot p(\theta \mid \mathcal{D}_{t-1}).
\end{equation*}
This recursive structure is the mathematical embodiment of the Bayesian reflex: each new observation triggers an automatic update of beliefs, with the posterior at time \(t-1\) serving as the prior for time \(t\). The proportionality constant ensures proper normalisation, but the core insight is that yesterday's posterior becomes today's prior—a natural and elegant way to accumulate knowledge over time.

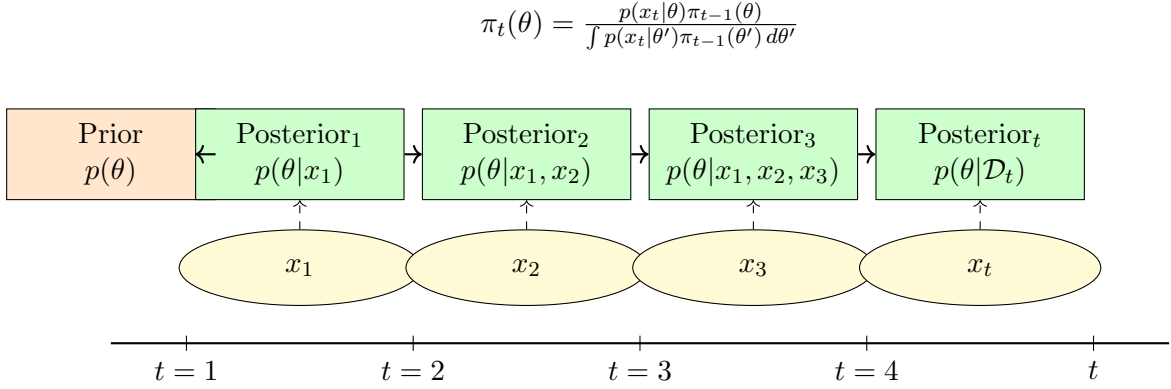
\begin{figure}[htbp]
\centering
\begin{tikzpicture}[node distance=2.2cm]
    % Define styles
    \tikzstyle{prior} = [rectangle, draw, fill=orange!20, text width=2.5cm, text centered, minimum height=1.2cm]
    \tikzstyle{posterior} = [rectangle, draw, fill=green!20, text width=2.5cm, text centered, minimum height=1.2cm]
    \tikzstyle{data} = [ellipse, draw, fill=yellow!20, text width=2cm, text centered, minimum height=1cm]
    
    % Time line
    \draw[thick] (0,0) -- (14,0);
    
    % Time points
    \foreach \x/\t in {1/$t=1$, 4/$t=2$, 7/$t=3$, 10/$t=4$, 13/$t$}
        \draw (\x,0.1) -- (\x,-0.1) node[below] {\t};
    
    % Priors and posteriors
    \node (prior0) [prior] at (0,2.5) {Prior \\ $p(\theta)$};
    \node (post1) [posterior] at (2.5,2.5) {Posterior$_1$ \\ $p(\theta|x_1)$};
    \node (post2) [posterior] at (5.5,2.5) {Posterior$_2$ \\ $p(\theta|x_1,x_2)$};
    \node (post3) [posterior] at (8.5,2.5) {Posterior$_3$ \\ $p(\theta|x_1,x_2,x_3)$};
    \node (post4) [posterior] at (11.5,2.5) {Posterior$_t$ \\ $p(\theta|\mathcal{D}_t)$};
    
    % Data points
    \node (data1) [data] at (2.5,1) {$x_1$};
    \node (data2) [data] at (5.5,1) {$x_2$};
    \node (data3) [data] at (8.5,1) {$x_3$};
    \node (data4) [data] at (11.5,1) {$x_t$};
    
    % Update arrows
    \draw[->, thick] (prior0) -- (post1);
    \draw[->, thick] (post1) -- (post2);
    \draw[->, thick] (post2) -- (post3);
    \draw[->, thick] (post3) -- (post4);
    
    % Data influence
    \draw[->, dashed] (data1) -- (post1);
    \draw[->, dashed] (data2) -- (post2);
    \draw[->, dashed] (data3) -- (post3);
    \draw[->, dashed] (data4) -- (post4);
    
    % Recursion formula
    \node at (7,4.2) {$\pi_t(\theta) = \frac{p(x_t \mid \theta) \pi_{t-1}(\theta)}{\int p(x_t \mid \theta') \pi_{t-1}(\theta') \, d\theta'}$};
\end{tikzpicture}
\caption{Sequential Bayesian updating: Yesterday's posterior becomes today's prior.}
\label{fig:sequential_update}
\end{figure}

More formally, let \(\pi_t(\theta) = p(\theta \mid \mathcal{D}_t)\) denote the posterior distribution after \(t\) observations. As illustrated in Figure \ref{fig:sequential_update}, the sequential update satisfies the fundamental recursion
\begin{equation}
\pi_t(\theta) = \frac{p(x_t \mid \theta) \pi_{t-1}(\theta)}{\int p(x_t \mid \theta') \pi_{t-1}(\theta') \, d\theta'}.
\label{eq:sequential_update}
\end{equation}
This recursion possesses several remarkable properties that make it theoretically appealing \citep{wainwright2008graphical, ghahramani2015probabilistic}. 

\begin{theorem}[Optimality of Sequential Bayesian Updating]
For any sequence of observations \(x_1, \ldots, x_t\) generated independently conditioned on \(\theta\) according to \(p(x_i \mid \theta)\), the sequential update in \eqref{eq:sequential_update} yields the same posterior distribution as the batch update that processes all data simultaneously:
\begin{equation*}
\pi_t(\theta) = \frac{\prod_{i=1}^t p(x_i \mid \theta) p(\theta)}{\int \prod_{i=1}^t p(x_i \mid \theta') p(\theta') \, d\theta'}.
\end{equation*}
\end{theorem}

\begin{proof}
The proof proceeds by induction. For \(t=1\), the sequential update is identical to the batch update. Assume the equality holds for \(t-1\). Then:
\begin{align*}
\pi_t(\theta) &= \frac{p(x_t \mid \theta) \pi_{t-1}(\theta)}{\int p(x_t \mid \theta') \pi_{t-1}(\theta') \, d\theta'} \\
&= \frac{p(x_t \mid \theta) \frac{\prod_{i=1}^{t-1} p(x_i \mid \theta) p(\theta)}{\int \prod_{i=1}^{t-1} p(x_i \mid \theta') p(\theta') \, d\theta'}}{\int p(x_t \mid \theta') \frac{\prod_{i=1}^{t-1} p(x_i \mid \theta') p(\theta')}{\int \prod_{i=1}^{t-1} p(x_i \mid \theta'') p(\theta'') \, d\theta''} \, d\theta'} \\
&= \frac{\prod_{i=1}^t p(x_i \mid \theta) p(\theta)}{\int \prod_{i=1}^t p(x_i \mid \theta') p(\theta') \, d\theta'}.
\end{align*}
The denominator simplifies because the integrals cancel, establishing the result.
\end{proof}

This theorem establishes that sequential processing loses no information compared to batch processing, a crucial property for online learning systems.

\subsection{Conjugate Families and Exponential Families}

For general likelihoods, however, the sequential update may lead to posterior distributions of increasing complexity. Even if the prior has a simple form, the posterior after just a few updates may become intractable, requiring approximation methods. This motivates the use of conjugate families, which are pairs of likelihood and prior distributions such that the posterior belongs to the same family as the prior. Conjugacy ensures that the mathematical form of the distribution is preserved under updating, making sequential inference computationally tractable and elegant \citep{bernardo2009bayesian, gelman2013bayesian}.

\begin{definition}[Conjugate Family]
A family of prior distributions \(\mathcal{F}\) is conjugate to a likelihood function \(p(x \mid \theta)\) if for any prior \(\pi \in \mathcal{F}\) and any observation \(x\), the posterior distribution \(\pi(\theta \mid x) \propto p(x \mid \theta) \pi(\theta)\) also belongs to \(\mathcal{F}\).
\end{definition}

The exponential family provides a unifying framework for understanding conjugacy in Bayesian inference \citep{wainwright2008graphical}. 

\begin{definition}[Exponential Family]
A probability distribution belongs to the exponential family if it can be written in the canonical form
\begin{equation}
p(x \mid \theta) = h(x) \exp\left(\eta(\theta)^\top T(x) - A(\theta)\right),
\label{eq:exp_family}
\end{equation}
where \(T(x)\) are sufficient statistics, \(\eta(\theta)\) are natural parameters, \(A(\theta)\) is the log-partition function ensuring normalisation, and \(h(x)\) is a base measure.
\end{definition}

Many common distributions belong to the exponential family. For example, the Bernoulli distribution can be written as \(p(x \mid \theta) = \theta^x (1-\theta)^{1-x} = \exp\left(x \log\frac{\theta}{1-\theta} + \log(1-\theta)\right)\). The Poisson distribution takes the form \(p(x \mid \theta) = \frac{e^{-\theta}\theta^x}{x!} = \exp\left(x \log\theta - \theta - \log x!\right)\). The Gaussian distribution is expressed as \(p(x \mid \mu, \sigma^2) = \frac{1}{\sqrt{2\pi\sigma^2}}\exp\left(-\frac{(x-\mu)^2}{2\sigma^2}\right) = \exp\left(\frac{\mu}{\sigma^2}x - \frac{1}{2\sigma^2}x^2 - \frac{\mu^2}{2\sigma^2} - \frac{1}{2}\log(2\pi\sigma^2)\right)\).

For exponential family distributions, the conjugate prior takes a natural form:

\begin{theorem}[Conjugate Prior for Exponential Families]
For a likelihood in exponential family form \eqref{eq:exp_family}, the conjugate prior is given by
\begin{equation}
p(\theta \mid \chi, \nu) \propto \exp\left(\eta(\theta)^\top \chi - \nu A(\theta)\right),
\label{eq:conjugate_prior}
\end{equation}
where \(\chi\) and \(\nu\) are hyperparameters. The posterior after observing \(n\) data points \(\{x_1,\ldots,x_n\}\) updates the hyperparameters as
\begin{align*}
\chi_n &= \chi_0 + \sum_{i=1}^n T(x_i), \\
\nu_n &= \nu_0 + n.
\end{align*}
\end{theorem}

\begin{proof}
The joint distribution of data and parameters is:
\begin{align*}
p(\{x_i\}, \theta) &\propto \left[\prod_{i=1}^n h(x_i) \exp\left(\eta(\theta)^\top T(x_i) - A(\theta)\right)\right] \exp\left(\eta(\theta)^\top \chi_0 - \nu_0 A(\theta)\right) \\
&= \left[\prod_{i=1}^n h(x_i)\right] \exp\left(\eta(\theta)^\top \left(\chi_0 + \sum_{i=1}^n T(x_i)\right) - (\nu_0 + n) A(\theta)\right).
\end{align*}
Treating this as a function of \(\theta\), we recognise it as having the same form as the prior, with updated hyperparameters \(\chi_n = \chi_0 + \sum_{i=1}^n T(x_i)\) and \(\nu_n = \nu_0 + n\).
\end{proof}

This update is remarkably simple: we merely add the sufficient statistics of the new data to the prior hyperparameters. This property makes conjugate families particularly attractive for online learning, where updates must be performed quickly and efficiently \citep{murphy2023probabilistic}.

However, in real-world applications, the true posterior distribution rarely belongs to a conjugate family. This fundamental limitation necessitates the development of general-purpose methods for sampling from arbitrary posterior distributions—methods that are practically exact, efficient, and free from convergence concerns.

\subsection{The Challenge of Exact Posterior Sampling}

The sequential update in \eqref{eq:sequential_update} is mathematically elegant, but its practical implementation requires the ability to sample from the posterior distribution \(\pi_t(\theta)\) at each time step. Traditional approaches to this problem fall into two categories. MCMC methods construct a Markov chain whose stationary distribution is the target posterior. While immensely popular and widely applicable \citep{robert2004monte, brooks2011handbook}, MCMC suffers from a fundamental limitation: it is impossible to determine with certainty when the chain has converged to its stationary distribution. As noted by \citet{bhattacharya2025iid}, ``even the most experienced MCMC theorists and practitioners often resort to deterministic or ad-hoc approximations to posterior distributions, considering the difficulties of diagnosing convergence insurmountable." The second category comprises Sequential Monte Carlo (SMC) methods, such as particle filters, which approximate the posterior using weighted samples that evolve over time \citep{doucet2001sequential, doucet2009tutorial}. While these methods are naturally suited to online settings, they suffer from particle degeneracy—the tendency for most particles to have negligible weight after a few updates—and provide only approximate samples from the posterior.

The concept of \textit{perfect simulation}, introduced by \citet{propp1996exact}, initially appeared promising as a solution to the convergence problem. However, as \citet{bhattacharya2025iid} observe, ``the apparent infeasibility of direct practical implementation in general Bayesian problems led to scepticism and, in many cases, outright disbelief." The challenge of developing practical perfect simulation methods for arbitrary posterior distributions remained largely unsolved for decades.

This is where the \textit{ellipsoidal decomposition framework} of \citep{bhattacharya2025iid} aims to provide a transformative breakthrough, to enable practically exact iid sampling from any target distribution on \(\mathbb{R}^d\) without the convergence concerns that plague MCMC methods.

\section{The Ellipsoidal Decomposition Framework: Perfect IID Sampling as a Computational Foundation}
\label{sec:iid}

The ellipsoidal decomposition framework of \citep{bhattacharya2025iid} provides a general methodology for producing iid realisations from any distribution on \(\mathbb{R}^d\), for arbitrary \(d \geq 1\). This framework is not merely a theoretical curiosity—it has been successfully applied to generate 10,000 iid realisations from standard distributions (normal, Student's \(t\), Cauchy) for dimensions up to 100, from 50-dimensional normal mixtures, and from challenging posterior distributions including a 160-dimensional spatial model for radionuclide count data on Rongelap Island. In all cases, implementation times are reasonable, often less than a minute in parallel computing environments.

\begin{figure}[htbp]
\centering
\begin{tikzpicture}[scale=1.2]
    % Define styles
    \tikzstyle{region} = [draw, thick]
    
    % Concentric ellipsoids
    \foreach \r/\c in {0.6/blue!20, 1.2/blue!40, 1.8/blue!60, 2.4/blue!80, 3.0/blue!100} {
        \draw[fill=\c, opacity=0.3] (0,0) ellipse (\r cm and \r*0.5 cm);
    }
    
    % Labels for regions
    \node at (0,0) {$\mathbf{A}_1$};
    \node at (0.9,0.3) {$\mathbf{A}_2$};
    \node at (1.5,0.6) {$\mathbf{A}_3$};
    \node at (2.1,0.9) {$\mathbf{A}_4$};
    \node at (2.7,1.2) {$\mathbf{A}_5$};
    
    % Center point
    \filldraw[red] (0,0) circle (3pt) node[above right] {$\boldsymbol{\mu}$};
    
    % Radius annotation
    \draw[<->] (0,0) -- (1.8,0.9) node[midway, above, sloped] {$\sqrt{c_i}$};
    
    % Formula - placed below with enough space
    \node at (0,-2.8) [text width=10cm, align=center] {$\mathbf{A}_i = \{\boldsymbol{\theta}: c_{i-1} \leq (\boldsymbol{\theta} - \boldsymbol{\mu})^\top \boldsymbol{\Sigma}^{-1} (\boldsymbol{\theta} - \boldsymbol{\mu}) \leq c_i\}$};
    
    % Perfect sampling algorithm - placed below formula
    \node at (0,-4.2) [rectangle, draw, fill=yellow!20, text width=10cm, text centered, minimum height=1.5cm] {
        \textbf{Perfect Sampling:} Draw $T_i \sim \text{Geometric}(\hat{p}_i)$, \\
        then simulate from $t=-T_i$ to $0$ using split-chain representation
    };
\end{tikzpicture}
\caption{Ellipsoidal decomposition: Target distribution decomposed into concentric ellipsoidal regions $\mathbf{A}_i$, enabling perfect iid sampling from each component.}
\label{fig:ellipsoidal}
\end{figure}

\subsection{The Core Idea: Decomposition into Ellipsoidal Regions}

Let \(\pi(\boldsymbol{\theta})\) be the target distribution supported on \(\mathbb{R}^d\), where \(d \geq 1\). As illustrated in Figure \ref{fig:ellipsoidal}, the key insight is to represent the distribution as an infinite mixture of component distributions defined on a nested sequence of compact sets:

\begin{equation}
\pi(\boldsymbol{\theta}) = \sum_{i=1}^{\infty} \pi(\mathbf{A}_i) \pi_i(\boldsymbol{\theta}), \label{eq:mixture_rep}
\end{equation}

where \(\mathbf{A}_i\) are disjoint compact subsets of \(\mathbb{R}^d\) such that \(\cup_{i=1}^{\infty} \mathbf{A}_i = \mathbb{R}^d\), and

\begin{equation*}
\pi_i(\boldsymbol{\theta}) = \frac{\pi(\boldsymbol{\theta})}{\pi(\mathbf{A}_i)} I_{\mathbf{A}_i}(\boldsymbol{\theta})
\end{equation*}

is the distribution of \(\boldsymbol{\theta}\) restricted to \(\mathbf{A}_i\). The mixing probabilities \(\pi(\mathbf{A}_i) = \int_{\mathbf{A}_i} \pi(d\boldsymbol{\theta}) \geq 0\) satisfy \(\sum_{i=1}^{\infty} \pi(\mathbf{A}_i) = 1\).

The elegance of this representation lies in its interpretation: the central ellipsoid \(\mathbf{A}_1\) captures the modal region of the distribution, while the annuli \(\mathbf{A}_i\) for \(i \geq 2\) represent increasingly distant tail regions. By constructing these sets appropriately, we transform the problem of sampling from a complicated distribution on an unbounded domain into a collection of simpler problems on compact domains.

\subsection{Constructing the Ellipsoidal Sequence}

For an appropriate \(d\)-dimensional vector \(\boldsymbol{\mu}\) and \(d \times d\) positive definite scale matrix \(\boldsymbol{\Sigma}\), we define:

\begin{equation*}
\mathbf{A}_i = \{\boldsymbol{\theta}: c_{i-1} \leq (\boldsymbol{\theta} - \boldsymbol{\mu})^\top \boldsymbol{\Sigma}^{-1} (\boldsymbol{\theta} - \boldsymbol{\mu}) \leq c_i\}, \quad i = 1,2,\ldots,
\end{equation*}

with \(0 = c_0 < c_1 < c_2 < \cdots\). Thus, \(\mathbf{A}_1\) is a closed ellipsoid centred at \(\boldsymbol{\mu}\) with shape determined by \(\boldsymbol{\Sigma}\), while for \(i \geq 2\), \(\mathbf{A}_i\) are ellipsoidal annuli—the regions between two successive concentric ellipsoids.

The parameters \(\boldsymbol{\mu}\) and \(\boldsymbol{\Sigma}\) can be estimated from a preliminary sample (e.g., from a short MCMC run or domain knowledge). Theoretically, any choice is valid, but judicious choices improve efficiency: \(\boldsymbol{\mu}\) should be near the centre of mass (or coordinate-wise median) of \(\pi\), and \(\boldsymbol{\Sigma}\) should approximate its covariance structure.

\subsection{Estimating the Mixing Probabilities}

The mixing probabilities \(\pi(\mathbf{A}_i)\) are not known a priori. However, they can be estimated using Monte Carlo sampling from the uniform distribution on each \(\mathbf{A}_i\). Writing \(\pi(\boldsymbol{\theta}) = C \tilde{\pi}(\boldsymbol{\theta})\) where \(C > 0\) is the unknown normalising constant, we have:

\begin{equation}
\pi(\mathbf{A}_i) = C \mathcal{L}(\mathbf{A}_i) \mathbb{E}_{\mathbf{A}_i}[\tilde{\pi}(\boldsymbol{\theta})], \label{eq:mixing_estimate}
\end{equation}

where \(\mathcal{L}(\mathbf{A}_i)\) is the Lebesgue measure of \(\mathbf{A}_i\) and the expectation is with respect to the uniform distribution on \(\mathbf{A}_i\). The Lebesgue measures are available in closed form:

\begin{align*}
\mathcal{L}(\mathbf{A}_1) &= |\mathbf{B}| \times \frac{\pi^{d/2}}{\Gamma(d/2 + 1)} c_1^{d/2}, \\
\mathcal{L}(\mathbf{A}_i) &= |\mathbf{B}| \times \frac{\pi^{d/2}}{\Gamma(d/2 + 1)} \left(c_i^{d/2} - c_{i-1}^{d/2}\right), \quad i \geq 2,
\end{align*}

where \(\boldsymbol{\Sigma} = \mathbf{B}\mathbf{B}^\top\) and \(|\mathbf{B}|\) denotes the determinant of the Cholesky factor \(\mathbf{B}\).

Uniform sampling from \(\mathbf{A}_1\) is straightforward: generate \(\mathbf{X} = (X_1,\ldots,X_d)^\top\) with \(X_k \stackrel{iid}{\sim} N(0,1)\), draw \(U \sim U(0,1)\), set \(\mathbf{Y} = \sqrt{c_1} U^{1/d} \mathbf{X}/\|\mathbf{X}\|\), and finally \(\boldsymbol{\theta} = \boldsymbol{\mu} + \mathbf{B}\mathbf{Y}\). For \(\mathbf{A}_i\) with \(i \geq 2\), when the annuli are narrow, a more efficient method samples near the surface of the ellipsoid by using a large ``pseudo-dimension'' \(\tilde{d}\) in the transformation, effectively concentrating mass near the boundary.

Let \(\hat{\pi}(\mathbf{A}_i)\) denote the estimate of \(\pi(\mathbf{A}_i)\) obtained from \(N_i\) Monte Carlo samples. For sufficiently large \(N_i\), we can ensure that for any \(\epsilon > 0\),

\begin{equation}
1 - \epsilon \leq \frac{{\pi}(\mathbf{A}_i)}{\hat\pi(\mathbf{A}_i)} \leq 1 + \epsilon. \label{eq:estimate_bound}
\end{equation}
Since $\sum_{i=1}^{\infty}\pi(\bA_i)=1$, it follows from (\ref{eq:estimate_bound}), that for any $\epsilon\in (0,1)$, 
\begin{equation*}
	\frac{1}{1+\epsilon}\leq \sum_{i=1}^{\infty}\hat\pi(\bA_i)\leq \frac{1}{1-\epsilon},%\label{eq:sum_bound}
\end{equation*}
%
%Furthermore, for sufficiently large \(M\),
%
%\begin{equation}
%1 - \epsilon \leq C \sum_{i=1}^M \hat{\pi}(\mathbf{A}_i) \leq 1 + \epsilon. \label{eq:normalisation_bound}
%\end{equation}
%
which allows us to construct an approximate, but well-defined density
\begin{equation*}
	\hat{\pi}(\boldsymbol{\theta}) =  \sum_{i=1}^{\infty} \hat{\pi}(\mathbf{A}_i) \pi_i(\boldsymbol{\theta}).
\end{equation*}
%
%which is well-defined due to \eqref{eq:normalisation_bound}. 
The above constructions and observations yield the following result: 

\begin{theorem}[Total variation bound]
	\label{theorem:tv_bound}
%Suppose that for all \(\boldsymbol{\theta} \in \mathbb{R}^d\) the density estimates satisfy
%\[
%1-\epsilon \;\le\; \frac{\pi(\boldsymbol{\theta})}{\hat\pi(\boldsymbol{\theta})} \;\le\; 1+\epsilon ,
%\]
%as implied by the component‑wise \(\epsilon\)‑relative error of the mixing weights (equation~(4) of \citet{bhattacharya2025iid}). Then
The total variation norm between $\pi$ and $\hat\pi$ is
\[
%\|\pi - \hat\pi\|_{\mathrm{TV}} \;\le\; \frac{\epsilon}{2(1+\epsilon)} < \frac{\epsilon}{2} = \mathcal{O}(\epsilon).
\|\pi - \hat\pi\|_{\mathrm{TV}} \; \leq \frac{\epsilon}{2}. %= \mathcal{O}(\epsilon).
\]
\end{theorem}

\begin{proof}
Due to the upper bound of (\ref{eq:estimate_bound}), for all $\btheta\in\mathbb R^d$, 
\begin{equation*}
	\pi(\btheta)=\sum_{i=1}^{\infty}{\hat\pi(\bA_i)}\times\frac{\pi(\bA_i)}{\hat\pi(\bA_i)}\times\pi_i(\btheta)
	\leq (1+\epsilon)\sum_{i=1}^{\infty}{\hat\pi(\bA_i)}\pi_i(\btheta)=(1+\epsilon)\hat\pi(\btheta).
\end{equation*}
That is, for all $\btheta\in\mathbb R^d$,
\begin{equation}
	%\frac{\pi(\btheta)}{\sum_{i=1}^{\infty}\widehat{\pi(\bA_i)}\pi_i(\btheta)}\leq 1+\epsilon.
	\frac{\pi(\btheta)}{\hat\pi(\btheta)}\leq 1+\epsilon,
	\label{eq:p5}
\end{equation}
	Similarly, the lower bound of (\ref{eq:estimate_bound}) yields
\begin{equation}
	%\frac{\pi(\btheta)}{\sum_{i=1}^{\infty}\widehat{\pi(\bA_i)}\pi_i(\btheta)}\leq 1+\epsilon.
	\frac{\pi(\btheta)}{\hat\pi(\btheta)}\geq 1-\epsilon,
	\label{eq:p6}
\end{equation}
	Combining (\ref{eq:p5}) and (\ref{eq:p6}) we obtain \(|\pi(\boldsymbol{\theta}) - \hat\pi(\boldsymbol{\theta})| \le \epsilon\,\hat\pi(\boldsymbol{\theta})\) for every \(\boldsymbol{\theta}\). Because both \(\pi\) and \(\hat\pi\) are probability densities,
\[
\|\pi - \hat\pi\|_{\mathrm{TV}}
 = \frac12 \int_{\mathbb{R}^d} |\pi(\boldsymbol{\theta}) - \hat\pi(\boldsymbol{\theta})|\,\dd\boldsymbol{\theta}
 \le \frac12 \int_{\mathbb{R}^d} \epsilon\,\hat\pi(\boldsymbol{\theta})\,\dd\boldsymbol{\theta}
 = \frac{\epsilon}{2}.
\]
\end{proof}

Crucially, for sufficiently small \(\epsilon > 0\), sampling from \(\hat{\pi}(\boldsymbol{\theta})\) 
is essentially equivalent to sampling from the true target \(\pi(\boldsymbol{\theta})\) via a rejection sampling argument: sample $\boldsymbol{\theta}$ from 
$\hat\pi(\boldsymbol{\theta})$ and $U$ from $U(0,1)$ until 
the following is satisfied: $$U<\frac{\pi(\boldsymbol{\theta})}{(1+\epsilon)\hat\pi(\boldsymbol{\theta})}.$$
Since
\begin{equation}
\frac{\pi(\boldsymbol{\theta})}{(1+\epsilon)\hat{\pi}(\boldsymbol{\theta})} \geq \frac{1-\epsilon}{1+\epsilon} \approx 1, \quad \forall \boldsymbol{\theta} \in \mathbb{R}^d, \label{eq:rejection_bound}
\end{equation}
and $U<1$ almost surely,
it follows that, with \(\epsilon\) sufficiently small, any sample from \(\hat{\pi}(\boldsymbol{\theta})\) is accepted with probability nearly one, making the approximation practically exact.

\subsection{Perfect Sampling from the Component Distributions}

The key to the ellipsoidal decomposition framework is the ability to draw perfect samples from each component distribution \(\pi_i(\boldsymbol{\theta})\) on the compact set \(\mathbf{A}_i\). This is achieved through a novel perfect sampling scheme based on a minorisation inequality for the Metropolis-Hastings algorithm with a uniform proposal density on \(\mathbf{A}_i\).

For \(\pi_i\), consider the Metropolis-Hastings algorithm in which all coordinates are updated simultaneously using an independence uniform proposal:

\begin{equation*}
q_i(\boldsymbol{\theta}') = \frac{1}{\mathcal{L}(\mathbf{A}_i)} I_{\mathbf{A}_i}(\boldsymbol{\theta}').
\end{equation*}

This choice is deliberate: since \(\mathbf{A}_i\) are constructed so that \(\pi_i\) is relatively flat on these sets (the central ellipsoid captures the mode, the annuli capture relatively flat tail regions), the uniform proposal provides a good approximation to the target, leading to a large minorisation constant.

Let \(s_i = \inf_{\boldsymbol{\theta} \in \mathbf{A}_i} \tilde{\pi}(\boldsymbol{\theta})\) and \(S_i = \sup_{\boldsymbol{\theta} \in \mathbf{A}_i} \tilde{\pi}(\boldsymbol{\theta})\). Then for the Metropolis-Hastings transition kernel \(P_i(\boldsymbol{\theta}, \cdot)\),

\begin{align*}
P_i(\boldsymbol{\theta}, \mathbb{B} \cap \mathbf{A}_i) &\geq \int_{\mathbb{B} \cap \mathbf{A}_i} \min\left\{1, \frac{\tilde{\pi}(\boldsymbol{\theta}')}{\tilde{\pi}(\boldsymbol{\theta})}\right\} q_i(\boldsymbol{\theta}') d\boldsymbol{\theta}' \\
&\geq \left(\frac{s_i}{S_i}\right) \times \frac{\mathcal{L}(\mathbb{B} \cap \mathbf{A}_i)}{\mathcal{L}(\mathbf{A}_i)} \\
&= p_i Q_i(\mathbb{B} \cap \mathbf{A}_i),
\end{align*}

where \(p_i = s_i/S_i\) and \(Q_i\) is the uniform probability measure on \(\mathbf{A}_i\). In practice, we use Monte Carlo estimates \(\hat{s}_i\) and \(\hat{S}_i\) of the infimum and supremum, obtaining \(\hat{p}_i = \hat{s}_i/\hat{S}_i - \eta_i\) with \(\eta_i\) very small (e.g., \(10^{-5}\)) to ensure the inequality holds.

This minorisation leads to a split-chain representation:

\begin{equation}
P_i(\boldsymbol{\theta}, \cdot) = \hat{p}_i Q_i(\cdot) + (1 - \hat{p}_i) R_i(\boldsymbol{\theta}, \cdot), \label{eq:split_chain}
\end{equation}

where \(R_i\) is the residual distribution. This representation has a remarkable consequence: the time to coalescence for two independent chains started from arbitrary initial states is geometrically distributed with parameter \(\hat{p}_i\). More precisely, if we run two independent chains using the same sequence of uniform random numbers, the time \(T_i\) until they first meet satisfies

\begin{equation}
\Prob(T_i = k) = (1 - \hat{p}_i)^{k-1} \hat{p}_i, \quad k = 1,2,\ldots \label{eq:coalescence_time}
\end{equation}

This geometric coalescence time enables a simple perfect sampling algorithm:

\begin{algorithm}[H]
\caption{Perfect Sampling from \(\pi_i\) (Adapted from \citet{bhattacharya2025iid})}
\label{alg:perfect_sampling}
\begin{algorithmic}[1]
\State Draw \(T_i \sim \text{Geometric}(\hat{p}_i)\) from \eqref{eq:coalescence_time}
\State Draw \(\boldsymbol{\theta}^{(-T_i)} \sim Q_i(\cdot)\) (uniformly from \(\mathbf{A}_i\))
\For{\(t = -T_i, -T_i+1, \ldots, -1\)}
    \State Simulate \(\boldsymbol{\theta}^{(t+1)} \sim R_i(\boldsymbol{\theta}^{(t)}, \cdot)\) using rejection sampling based on the Metropolis-Hastings kernel
\EndFor
\State Return \(\boldsymbol{\theta}^{(0)}\) as a perfect sample from \(\pi_i\)
\end{algorithmic}
\end{algorithm}

The rejection sampling step for the residual distribution \(R_i\) is implemented by noting that

\begin{equation}
\frac{\tilde{R}_i(\boldsymbol{\theta}, \boldsymbol{\theta}')}{\tilde{P}_i(\boldsymbol{\theta}, \boldsymbol{\theta}')} \leq \frac{1}{1 - \hat{p}_i}, \quad \forall \boldsymbol{\theta}, \boldsymbol{\theta}' \in \mathbf{A}_i, \label{eq:rejection_ratio}
\end{equation}

where \(\tilde{R}_i\) and \(\tilde{P}_i\) are the densities of \(R_i\) and \(P_i\) respectively. Thus, given \(\boldsymbol{\theta}^{(t)}\), we can simulate \(\boldsymbol{\theta}' \sim \tilde{P}_i(\boldsymbol{\theta}^{(t)}, \cdot)\) using the Metropolis-Hastings kernel and accept it as a draw from \(R_i\) with probability \((1 - \hat{p}_i)\tilde{R}_i(\boldsymbol{\theta}^{(t)}, \boldsymbol{\theta}')/\tilde{P}_i(\boldsymbol{\theta}^{(t)}, \boldsymbol{\theta}')\).

This perfect sampling scheme avoids the computationally expensive coupling-from-the-past (CFTP) algorithm of \citet{propp1996exact}, making it practical for high-dimensional problems. The key innovation is using the split-chain representation to transform the problem of finding coalescence times into a simple geometric draw.

\subsection{The Complete IID Sampling Algorithm}

Combining the component selection and perfect sampling steps yields the complete algorithm for generating iid samples from any target distribution \(\pi\):

\begin{algorithm}[H]
\caption{IID Sampling from Target Distribution \(\pi\) \citep{bhattacharya2025iid}}
\label{alg:iid_sampling}
\begin{algorithmic}[1]
\State Obtain estimates \(\boldsymbol{\mu}\) and \(\boldsymbol{\Sigma}\) for the ellipsoidal sets (e.g., from a short MCMC run)
\State Choose initial \(M\) (number of ellipsoidal regions) and radii \(\sqrt{c_i}\) appropriately
\State \textbf{Parallel Step:} For \(i = 1,\ldots,M\), in parallel processors:
    \State \quad Generate \(N_i\) uniform samples from \(\mathbf{A}_i\)
    \State \quad Compute Monte Carlo estimates \(\hat{\pi}(\mathbf{A}_i)\) using \eqref{eq:mixing_estimate}
    \State \quad Compute \(\hat{s}_i = \min \tilde{\pi}(\boldsymbol{\theta})\) and \(\hat{S}_i = \max \tilde{\pi}(\boldsymbol{\theta})\) over the samples
    \State \quad Set \(\hat{p}_i = \hat{s}_i/\hat{S}_i - \eta_i\) with \(\eta_i\) small (e.g., \(10^{-5}\))
\State Each processor broadcasts its estimates to all others
\State \textbf{Parallel IID Generation:} For each desired iid sample (in parallel processors):
    \State \quad Select component \(i\) with probability proportional to \(\hat{\pi}(\mathbf{A}_i)\) using \eqref{eq:component_selection}
    \State \quad If \(i = M\), increase \(M\) to \(2M\) and return to Step 3 for the new indices
    \State \quad Generate perfect sample \(\boldsymbol{\theta}^{(0)}\) from \(\pi_i\) using Algorithm \ref{alg:perfect_sampling}
    \State \quad Return \(\boldsymbol{\theta}^{(0)}\) as a perfect iid realisation from \(\pi\)
\State Aggregate all iid samples for downstream analysis
\end{algorithmic}
\end{algorithm}

The component selection in Line 10 of the algorithm uses a uniform random number \(U \sim U(0,1)\) and selects the smallest \(i\) such that

\begin{equation}
\frac{\sum_{j=1}^{i-1} \hat{\pi}(\mathbf{A}_j)}{\sum_{j=1}^M \hat{\pi}(\mathbf{A}_j)} < U \leq \frac{\sum_{j=1}^i \hat{\pi}(\mathbf{A}_j)}{\sum_{j=1}^M \hat{\pi}(\mathbf{A}_j)}. \label{eq:component_selection}
\end{equation}

For sufficiently small \(\epsilon\) in \eqref{eq:estimate_bound}, %-\eqref{eq:normalisation_bound}, 
this selection is effectively equivalent to sampling from the true mixing probabilities.

The schematic diagram of the iid sampling algorithm is provided in Figure \ref{fig:iid_algorithm}.

\subsection{Theoretical Guarantees and Practical Performance}

The ellipsoidal decomposition framework comes with strong theoretical guarantees. The rejection sampling argument \eqref{eq:rejection_bound} ensures that for sufficiently small \(\epsilon\), sampling from \(\hat{\pi}(\boldsymbol{\theta})\) is practically equivalent to sampling from the true target \(\pi(\boldsymbol{\theta})\); the precise formalism is
provided by Theorem \ref{theorem:tv_bound}. The perfect sampling algorithm for each component \(\pi_i\) produces practically exact draws from the component distribution, with coalescence time geometrically distributed and thus almost surely finite.

Empirical validation across a wide range of distributions demonstrates the practical power of this approach. \citet{bhattacharya2025iid} report generating 10,000 iid realisations from normal distributions in dimensions \(d = 1,5,10,50,100\) with computation times ranging from 6 minutes to 1 hour 45 minutes; the iid samples accurately capture both marginal densities and correlation structures. For Student's \(t\) distributions with 5 degrees of freedom in dimensions up to 100, computation times were as low as less than a minute for \(d=100\); the thicker tails lead to larger minorisation constants \(\hat{p}_i\) and thus faster coalescence. Cauchy distributions in dimensions up to 100 were completed in under 2 minutes, with extremely thick tails yielding the largest minorisation constants and fastest sampling. A 50-dimensional mixture of normals (two components with means \(\boldsymbol{\nu}\) and \(2\boldsymbol{\nu}\), mixing proportions \(2/3\) and \(1/3\)) was completed in 7 minutes with highly accurate results.

Perhaps most impressively, the framework has been successfully applied to generate 10,000 iid realisations from the 160-dimensional posterior distribution for radionuclide count data on Rongelap Island—a problem widely regarded as extremely difficult for MCMC convergence. Using a diffeomorphic transformation to flatten the posterior and improve minorisation constants, the iid samples were generated in just 30 minutes, with results closely matching those from a carefully tuned run of Transformation-based Markov Chain Monte Carlo
(TMCMC) \citep{dutta2014markov} that took 1 hour 40 minutes.

\begin{figure}[htbp]
\centering
\begin{tikzpicture}[node distance=2.8cm, scale=0.9, transform shape]
    % Initialisation
    \node (start) [rectangle, draw, fill=gray!20, minimum width=6cm, text centered] {Target $\pi(\boldsymbol{\theta})$};
    \node (estimate) [rectangle, draw, fill=blue!20, below of=start, minimum width=6cm, text centered] {Estimate $\boldsymbol{\mu}$, $\boldsymbol{\Sigma}$ from preliminary sample};
    
    % Parallel estimation phase
    \node (parallel1) [rectangle, draw, fill=green!20, below of=estimate, minimum width=8cm, text width=8cm, align=center, yshift=-0.8cm] {
        \textbf{Parallel Step:} For $i=1,\ldots,M$ (in parallel) \\
        \quad - Generate $N_i$ uniform samples from $\mathbf{A}_i$ \\
        \quad - Compute $\hat{\pi}(\mathbf{A}_i) = C\mathcal{L}(\mathbf{A}_i)\mathbb{E}_{\mathbf{A}_i}[\tilde{\pi}(\boldsymbol{\theta})]$ \\
        \quad - Compute $\hat{s}_i = \min \tilde{\pi}$, $\hat{S}_i = \max \tilde{\pi}$, $\hat{p}_i = \hat{s}_i/\hat{S}_i - \eta_i$
    };
    
    % Component selection
    \node (select) [rectangle, draw, fill=yellow!20, below of=parallel1, minimum width=8cm, text width=8cm, align=center, yshift=-1.0cm] {
        \textbf{Component Selection:} Draw $U \sim U(0,1)$, select smallest $i$ such that \\
        $\frac{\sum_{j=1}^{i-1} \hat{\pi}(\mathbf{A}_j)}{\sum_{j=1}^M \hat{\pi}(\mathbf{A}_j)} < U \leq \frac{\sum_{j=1}^i \hat{\pi}(\mathbf{A}_j)}{\sum_{j=1}^M \hat{\pi}(\mathbf{A}_j)}$
    };
    
    % Perfect sampling
    \node (perfect) [rectangle, draw, fill=purple!20, below of=select, minimum width=8cm, text width=8cm, align=center, yshift=-1.2cm] {
        \textbf{Perfect Sampling from $\pi_i$:} \\
        \quad 1. Draw $T_i \sim \text{Geometric}(\hat{p}_i)$ \\
        \quad 2. Draw $\boldsymbol{\theta}^{(-T_i)} \sim Q_i$ (uniform on $\mathbf{A}_i$) \\
        \quad 3. For $t = -T_i, -T_i+1, \ldots, -1$: simulate $\boldsymbol{\theta}^{(t+1)} \sim R_i(\boldsymbol{\theta}^{(t)}, \cdot)$ \\
        \quad 4. Return $\boldsymbol{\theta}^{(0)}$ as perfect sample from $\pi_i$
    };
    
    % Output
    \node (output) [rectangle, draw, fill=red!20, below of=perfect, minimum width=6cm, text centered, yshift=-1.0cm] {IID Sample from $\pi$};
    
    % Connections
    \draw[->] (start) -- (estimate);
    \draw[->] (estimate) -- (parallel1);
    \draw[->] (parallel1) -- (select);
    \draw[->] (select) -- (perfect);
    \draw[->] (perfect) -- (output);
    
    % Rejection bound - placed much further below with clear separation
    \node at (0,-18.5) [text width=10cm, align=center, yshift=-1.5cm] {$\frac{\pi(\boldsymbol{\theta})}{(1+\epsilon)\hat{\pi}(\boldsymbol{\theta})} \geq \frac{1-\epsilon}{1+\epsilon} \approx 1$ for small $\epsilon$};
\end{tikzpicture}
\caption{Practically perfect IID sampling algorithm using ellipsoidal decomposition framework.}
\label{fig:iid_algorithm}
\end{figure}

\subsection{Extensions and Generalisations}

The foundational iid sampling methodology has spawned a rich sequence of extensions that further enhance its applicability to online learning contexts. \citet{bhattacharya2021multimodal} extended the framework to multimodal and variable-dimensional distributions by constructing separate ellipsoidal sequences for each mode, with appropriate mixing probabilities across modes; variable-dimensional problems such as mixture models with unknown number of components are also handled within the perfect sampling framework. For doubly intractable distributions, where the normalising constant depends on parameters and is analytically intractable, \citet{bhattacharya2021doubly} combined Monte Carlo and importance sampling approximations with Gaussian process interpolations to achieve highly accurate iid sampling; applications include Ising models, Strauss processes, and autologistic models in dimensions up to 100. \citet{bhattacharya2022dirichlet} enabled practically exact iid sampling from posterior Dirichlet process mixtures—distributions that are notoriously difficult to sample using MCMC due to their infinite-dimensional and discrete nature; the method is validated on benchmark datasets (enzyme, acidity, galaxy) and generates 10,000 iid realisations in minutes using parallel computing. Most recently, \citet{bhattacharya2025flows} extended the iid sampling procedure to Bayesian nonparametrics based on random normalising flows—distributions on non-Euclidean spaces induced by compositions of monotone Gaussian processes; this extension highlights the convergence challenges of MCMC in high dimensions and the necessity of iid sampling---indeed, iid sampling has been applied to a simulated data, associated with $4776$ parameters pf the normalizing flows model. These extensions demonstrate that the ellipsoidal decomposition framework is not merely a standalone method but a versatile foundation upon which a comprehensive approach to practically exact Bayesian computation can be built—an essential infrastructure for the Bayesian reflex.

\section{Methodological Foundations: From Kalman to Particles}

\subsection{The Kalman Filter: Exact Online Inference for Linear Gaussian Systems}

The Kalman filter stands as the archetypal Bayesian online learning algorithm, providing optimal recursive state estimation for linear Gaussian systems since its development in the 1960s \citep{kalman1960new}. Its elegance, efficiency, and theoretical optimality have made it indispensable in countless applications from aerospace navigation to financial forecasting \citep{kailath2000linear, anderson1979optimal}.

Consider a latent state \(\theta_t \in \mathbb{R}^d\) evolving according to linear dynamics with Gaussian noise: \(\theta_t = F_t \theta_{t-1} + w_t\) with \(w_t \sim \mathcal{N}(0, Q_t)\), and observations given by \(x_t = H_t \theta_t + v_t\) with \(v_t \sim \mathcal{N}(0, R_t)\), where \(F_t\) is the state transition matrix, \(H_t\) is the observation matrix, and \(w_t, v_t\) are Gaussian noise processes with covariances \(Q_t\) and \(R_t\) respectively.

\begin{theorem}[Kalman Filter Recursions \citep{kalman1960new}]
For the linear Gaussian state-space model, the posterior distribution \(p(\theta_t \mid x_{1:t})\) is Gaussian with mean \(\hat{\theta}_{t|t}\) and covariance \(P_{t|t}\) given by the following recursions:

\textbf{Prediction step:}
\begin{align*}
\hat{\theta}_{t|t-1} &= F_t \hat{\theta}_{t-1|t-1}, \\
P_{t|t-1} &= F_t P_{t-1|t-1} F_t^\top + Q_t.
\end{align*}

\textbf{Innovation:}
\begin{align*}
\tilde{y}_t &= x_t - H_t \hat{\theta}_{t|t-1}, \\
S_t &= H_t P_{t|t-1} H_t^\top + R_t.
\end{align*}

\textbf{Kalman gain:}
\begin{equation*}
K_t = P_{t|t-1} H_t^\top S_t^{-1}.
\end{equation*}

\textbf{Update step:}
\begin{align*}
\hat{\theta}_{t|t} &= \hat{\theta}_{t|t-1} + K_t \tilde{y}_t, \\
P_{t|t} &= (I - K_t H_t) P_{t|t-1}.
\end{align*}
\end{theorem}

For nonlinear systems, the Extended Kalman Filter (EKF) linearises the dynamics and observation models \citep{julier2004unscented}. The Unscented Kalman Filter (UKF) propagates sigma points through the nonlinear transformations and reconstructs the mean and covariance \citep{julier2004unscented, wan2000unscented}. These methods have been widely used in online learning contexts, though they remain approximations for non-Gaussian settings.

\subsection{Particle Filters: Sequential Monte Carlo Methods}

When dynamics are highly nonlinear or non-Gaussian, particle filters provide a flexible and powerful approximation \citep{doucet2001sequential, doucet2009tutorial, gordon1993novel}. Rather than assuming a specific parametric form for the posterior, particle filters represent the distribution using a set of weighted samples that evolve over time.

A particle approximation to a distribution \(\pi(\theta)\) is a set of weighted samples \(\{\theta^{(i)}, w^{(i)}\}_{i=1}^N\) such that \(\pi(\theta) \approx \sum_{i=1}^N w^{(i)} \delta_{\theta^{(i)}}(\theta)\), where \(\delta\) denotes the Dirac delta function and \(\sum_i w^{(i)} = 1\) \citep{doucet2001sequential}.

The particle filter proceeds through sequential importance sampling with resampling. The bootstrap filter \citep{gordon1993novel} is a simple and popular variant; see
Algorithm \ref{alg:particle_filter}
\begin{algorithm}
\caption{Particle Filter (Bootstrap Filter) \citep{gordon1993novel}}
\label{alg:particle_filter}
\begin{algorithmic}[1]
\State \textbf{Initialisation:} For \(i=1,\ldots,N\), sample \(\theta_0^{(i)} \sim p(\theta_0)\) and set \(w_0^{(i)} = 1/N\).
\For{\(t = 1\) to \(T\)}
    \State \textbf{Propagation:} For \(i=1,\ldots,N\), sample \(\theta_t^{(i)} \sim p(\theta_t \mid \theta_{t-1}^{(i)})\).
    \State \textbf{Weighting:} For \(i=1,\ldots,N\), compute \(w_t^{(i)} = w_{t-1}^{(i)} p(x_t \mid \theta_t^{(i)})\).
    \State \textbf{Normalisation:} \(w_t^{(i)} = w_t^{(i)} / \sum_{j=1}^N w_t^{(j)}\).
    \State \textbf{Resampling:} Compute \(N_{\text{eff}} = 1 / \sum_{i=1}^N (w_t^{(i)})^2\).
    \If{\(N_{\text{eff}} < N_{\text{thresh}}\)}
        \State Resample \(N\) particles with probabilities proportional to \(w_t^{(i)}\), set weights to \(1/N\).
    \EndIf
\EndFor
\end{algorithmic}
\end{algorithm}

The effective sample size \(N_{\text{eff}} = 1 / \sum_{i=1}^N (w_t^{(i)})^2\) measures the quality of the particle approximation. More sophisticated proposal distributions can improve efficiency \citep{doucet2000sequential}. Despite their flexibility, particle filters provide only approximate samples and can suffer from degeneracy in high dimensions \citep{bengtsson2008curse}. The ellipsoidal decomposition framework offers a potential path toward exact sequential Monte Carlo methods that maintain iid samples over time without degeneracy.

\subsection{Variational Online Learning}

For high-dimensional problems where exact inference is intractable and particle methods become inefficient, variational inference provides a scalable alternative rooted in optimisation rather than sampling \citep{blei2017variational, jordan1999introduction}.

Variational inference approximates a target distribution \(\pi(\theta)\) by finding the closest distribution \(q^*(\theta)\) within a tractable family \(\mathcal{Q}\): \(q^* = \arg\min_{q \in \mathcal{Q}} \KL{q(\theta)}{\pi(\theta)}\) \citep{blei2017variational}. Minimising the KL divergence is equivalent to maximizing the evidence lower bound (ELBO) \citep{wainwright2008graphical}: \(\log p(\mathcal{D}) = \KL{q(\theta)}{\pi(\theta \mid \mathcal{D})} + \mathcal{L}(q)\) where \(\mathcal{L}(q) = \mathbb{E}_q[\log p(\mathcal{D}, \theta) - \log q(\theta)]\).

For exponential family models, the optimal variational distribution often takes a convenient form. In the online setting, we maintain a variational posterior \(q_t(\theta)\) that is updated recursively \citep{ghahramani2015probabilistic, nguyen2018variational}: \(q_t(\theta) \approx \frac{p(x_t \mid \theta) q_{t-1}(\theta)}{\int p(x_t \mid \theta') q_{t-1}(\theta') \, d\theta'}\). This leads to online variational inference algorithms. For exponential family models, the natural parameter \(\eta_t\) updates as \(\eta_t = \eta_{t-1} + \nabla_\eta \mathbb{E}_{q_{t-1}}[\log p(x_t \mid \theta)]\).

For mean-field variational families where \(q(\theta) = \prod_j q_j(\theta_j)\), the coordinate ascent updates take the form \(\log q_j(\theta_j) = \mathbb{E}_{q_{-j}}[\log p(\mathcal{D}, \theta)] + \text{constant}\).

Variational methods provide computational efficiency at the cost of approximation error. The ellipsoidal decomposition framework offers a way to validate these approximations by providing practically exact iid samples for comparison, and in some cases, may even replace variational approximations with essentially exact inference for moderate-dimensional problems.

\subsection{Connections to Online Learning Theory}

Bayesian online learning connects deeply with classical online learning theory through the concept of regret \citep{lattimore2020bandit, cesa2006prediction, hazan2016introduction}. The regret of an online algorithm that makes predictions \(\hat{\theta}_t\) and suffers losses \(\ell_t(\hat{\theta}_t)\) is \(R_T = \sum_{t=1}^T \ell_t(\hat{\theta}_t) - \min_{\theta} \sum_{t=1}^T \ell_t(\theta)\) \citep{lattimore2020bandit}.

For Bayesian algorithms based on posterior sampling, we have the following bound \citep{russo2018tutorial}: under suitable conditions on the reward distributions and the prior, Thompson sampling (to be discussed in Section \ref{subsec:thompson_sampling}) achieves expected regret \(\mathbb{E}[R_T] = \tilde{\mathcal{O}}(\sqrt{d_E T})\), where \(d_E\) is the Eluder dimension of the reward function class \citep{russo2014learning}. The Eluder dimension measures how many observations are needed to distinguish between different hypotheses \citep{russo2014learning}. Other popular algorithms include Upper Confidence Bound (UCB) \citep{auer2002finite} and EXP3 \citep{auer2002nonstochastic} for adversarial settings. The availability of practically exact iid samples from the posterior, enabled by the ellipsoidal decomposition framework, allows for sharper theoretical analysis of these regret bounds and potentially tighter constants.

\begin{figure}[htbp]
\centering
\begin{tikzpicture}[scale=0.7, transform shape]
    % Timeline arrow
    \draw[thick, ->] (2, -1.8) -- (22, -1.8);
    \foreach \x/\y in {3/1960, 6/1980, 9/2000, 12/2010, 15/2020} {
        \draw (\x,-1.7) -- (\x,-1.9) node[below] {\y};
    }

    % Top row (y=3)
    \node (kalman) [rectangle, draw, fill=red!20, text width=2.2cm, align=center, minimum height=1cm] at (3,3) {\small Kalman Filter};
    \node (ekf)    [rectangle, draw, fill=orange!20, text width=2.2cm, align=center, minimum height=1cm] at (6,3) {\small EKF/UKF};
    \node (particle)[rectangle, draw, fill=yellow!20, text width=2.2cm, align=center, minimum height=1cm] at (9,3) {\small Particle Filters};
    \node (rgp)    [rectangle, draw, fill=blue!40, text width=2.2cm, align=center, minimum height=1cm] at (16,3) {\small Recursive GP};

    % Bottom row (y=0)
    \node (lut)    [rectangle, draw, fill=green!20, text width=2.2cm, align=center, minimum height=1cm] at (6,0) {\small Look-up Table};
    \node (ssm)    [rectangle, draw, fill=green!40, text width=2.2cm, align=center, minimum height=1cm] at (9,0) {\small Nonparametric SSM};
    \node (circular)[rectangle, draw, fill=green!60, text width=2.2cm, align=center, minimum height=1cm] at (12,0) {\small Circular Latent};
    \node (ellip)  [rectangle, draw, fill=purple!40, text width=2.2cm, align=center, minimum height=1cm] at (16,0) {\small Ellipsoidal};

    % Bayesian Reflex node
    \node (reflex) at (21,1.5) [rectangle, draw, fill=red!20, text width=2.5cm, align=center, minimum height=1.2cm] {\small \textbf{Bayesian Reflex}};

    % Merge point
    \coordinate (merge) at (18.5,1.5);

    % --- Arrows ---
    % Top row connections
    \draw[->, thick] (kalman) -- (ekf);
    \draw[->, thick] (ekf) -- (particle);
    \draw[->, thick] (particle) -- (12.5,3) -- (12.5,2.5) -- (16,2.5) -- (rgp);

    % Bottom row connections
    \draw[->, thick] (lut) -- (ssm);
    \draw[->, thick] (ssm) -- (circular);
    \draw[->, thick] (circular) -- (ellip);

    % Dashed arrow: Particle → Look-up Table (influence) - approaches from left
    \draw[->, thick, dashed] (particle) -- (9,2) -- (5,2) -- (5,0.5) -- (lut);

    % Solid arrow: Look-up Table → merge - exits from right
    \draw[->, thick] (lut.east) -- (8,0.5) -- (12,1) -- (16,1) -- (merge);

    % Connections from RGP and Ellip to merge
    \draw[->, thick] (rgp) -- (17,3) -- (17,2) -- (merge);
    \draw[->, thick] (ellip) -- (17,0) -- (17,0.8) -- (merge);

    % Merge to Reflex
    \draw[->, thick] (merge) -- (reflex);

    % Citations
    \node at (6,-2.6)  [font=\tiny] {\cite{bhattacharya2007simulation}};
    \node at (9,-2.6)  [font=\tiny] {\cite{ghosh2014bayesian}};
    \node at (12,-2.6) [font=\tiny] {\cite{mazumder2016bayesian}};
    \node at (15,-2.6) [font=\tiny] {\cite{bhattacharya2025iid}};
    \node at (17,-2.6) [font=\tiny] {\cite{bhattacharya2025bayesian}};
    \node at (21,-2.6) [font=\tiny] {Bayesian Reflex (this work)};

    % Legend
    \node at (4, -3.5) [text width=8cm, align=left, font=\tiny] {
        Solid arrows: Direct evolution \\
        Dashed arrow: Influence on perfect sampling
    };
\end{tikzpicture}
\caption{Methodological evolution: From Kalman filters through look-up table principle to Recursive Gaussian Processes and ellipsoidal decomposition.}
\label{fig:method_evolution}
\end{figure}
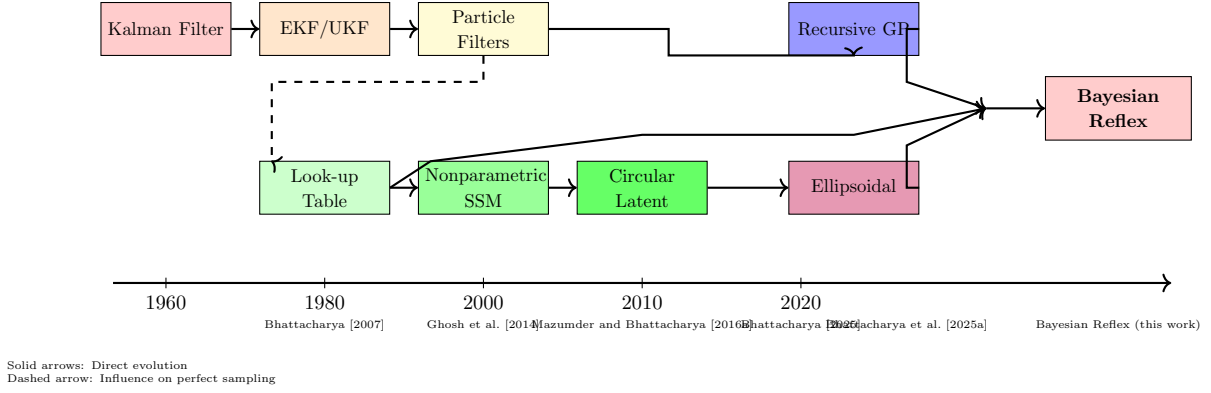

Figure \ref{fig:method_evolution} illustrates the methodological evolution from classical Kalman filters to modern Recursive Gaussian Processes and the ellipsoidal decomposition framework.

\section{The Look-Up Table Principle: A Foundational Contribution to Online Bayesian Learning}

Before delving into modern Recursive Gaussian Processes and their applications, we highlight a contribution that is related to the core computational principle underlying many online Bayesian methods for dynamic systems. The work of \citet{bhattacharya2007simulation} introduced the concept of using auxiliary variables and ``look-up tables'' to enable efficient sequential inference in dynamic computer models. This idea has since been a cornerstone of Gaussian process induced online Bayesian works of these authors, providing the mathematical foundation for the recursive architectures and inverse regression methods we discuss in subsequent sections. Somewhat parallel ideas appear in the literature on sparse Gaussian processes \citep{quinonero2005sparse, snelson2006sparse} and inducing point methods, where a set of pseudo-inputs summarises the function.

\subsection{Dynamic Computer Models and the Need for Sequential Emulation}

Consider a dynamic computer model that evolves over time according to the relationship:
\begin{equation*}
\mathbf{y}_t = \pmb{\eta}(\mathbf{z}_t, \mathbf{y}_{t-1}) = \pmb{\eta}(\mathbf{v}_t)
\end{equation*}
where \(\mathbf{y}_t\) is a \(p\)-dimensional output at time \(t\), \(\mathbf{z}_t\) represents forcing inputs (external factors recorded over time), and \(\pmb{\eta}(\cdot)\) is an unknown function treated as a black box. Such models arise in numerous scientific and engineering applications, from climate modelling to vegetation dynamics.

The goal is to predict the entire dynamic sequence \(\{\mathbf{y}_1, \mathbf{y}_2, \ldots, \mathbf{y}_T\}\) given only a training dataset \(\mathbf{D} = \{\pmb{\eta}(\mathbf{g}_1), \ldots, \pmb{\eta}(\mathbf{g}_n)\}\) obtained by running the expensive computer code at a carefully chosen set of input points \(\mathbf{G} = (\mathbf{g}_1, \ldots, \mathbf{g}_n)'\). This is fundamentally an online learning problem: we must simulate the sequence step by step, with each prediction depending on the previous one, while quantifying uncertainty throughout.

\subsection{The Gaussian Process Emulator}

Following standard practice in computer experiments \citep{sacks1989design, santner2018design}, the unknown function \(\pmb{\eta}(\cdot)\) is modelled as a Gaussian process. For the multivariate case, this takes the form of a matrix-variate Gaussian process:
\begin{equation*}
[\mathbf{D} \mid \mathbf{B}, \boldsymbol{\Sigma}, \mathbf{R}] \sim \mathcal{N}_{n,p}(\mathbf{H}_D \mathbf{B}, \mathbf{A}_D, \boldsymbol{\Sigma})
\end{equation*}
where \(\mathbf{H}_D\) contains basis functions evaluated at the training inputs, \(\mathbf{A}_D\) is the correlation matrix with entries \(c(\mathbf{g}_r, \mathbf{g}_\ell) = \exp\{-(\mathbf{g}_r - \mathbf{g}_\ell)'\mathbf{R}(\mathbf{g}_r - \mathbf{g}_\ell)\}\), and \(\boldsymbol{\Sigma}\) captures output covariances.

Given the training data, the posterior predictive distribution for \(\pmb{\eta}(\cdot)\) at any new input \(\mathbf{v}\) is a multivariate Student's \(t\) distribution:
\begin{equation*}
[\pmb{\eta}(\mathbf{v}) \mid \mathbf{R}, \mathbf{D}] \sim \mathcal{T}_p\left(\boldsymbol{\mu}_2(\mathbf{v}), c_2(\mathbf{v},\mathbf{v})\hat{\boldsymbol{\Sigma}}_{GLS}; n-m\right)
\end{equation*}
with mean \(\boldsymbol{\mu}_2(\mathbf{v})\) and scale parameter derived from the Gaussian process. This emulator provides both predictions and uncertainty quantification—the first principle of the Bayesian reflex.

\subsection{The Look-Up Table Innovation}

The key challenge in simulating a dynamic sequence is the recursive dependence: to predict \(\mathbf{y}_2\), we need \(\mathbf{y}_1\) as input; to predict \(\mathbf{y}_3\), we need \(\mathbf{y}_2\), and so on. A naive approach would attempt to directly simulate from the joint posterior \([\mathbf{y}_1, \mathbf{y}_2, \ldots, \mathbf{y}_T \mid \mathbf{D}]\), but this quickly becomes intractable due to the growing dependencies.

The key insight of \citet{bhattacharya2007simulation} was to introduce a set of latent auxiliary variables \(\mathbf{D}^* = \{\pmb{\eta}(\mathbf{g}_1^*), \ldots, \pmb{\eta}(\mathbf{g}_N^*)\}\) defined on a fixed grid \(\mathbf{G}^* = (\mathbf{g}_1^*, \ldots, \mathbf{g}_N^*)'\) spanning the expected range of the dynamic sequence. This ``look-up table'' serves as a finite-dimensional proxy for the entire Gaussian process.

The crucial property is that, given \(\mathbf{D}^*\), the dynamic sequence exhibits a Markov structure:
\begin{equation*}
[\pmb{\eta}(\mathbf{v}_t) \mid \mathbf{D}^*, \mathbf{D}, \mathbf{v}_t, \mathbf{v}_{t-1}, \ldots] \approx [\pmb{\eta}(\mathbf{v}_t) \mid \mathbf{D}^*, \mathbf{D}, \mathbf{v}_t].
\end{equation*}

This conditional independence is not assumed but arises from the properties of Gaussian process regression: given the function values on the grid \(\mathbf{G}^*\), the distribution at any new point depends only on \(\mathbf{D}^*\) and the training data, not on the previous points in the sequence.
Figure \ref{fig:lut} illustrates this schematically.

\subsection{The Algorithm: Sequential Simulation with Latent Variables}

The complete algorithm for simulating from the posterior dynamic sequence proceeds as follows:

\textbf{Step 1:} Simulate the first value \(\mathbf{y}_1 = \pmb{\eta}(\mathbf{v}_1)\) from the marginal posterior \([\pmb{\eta}(\mathbf{v}_1) \mid \mathbf{R}, \mathbf{D}]\).

\textbf{Step 2:} Simulate the look-up table \(\mathbf{D}^*\) from the conditional distribution \([\mathbf{D}^* \mid \mathbf{R}, \mathbf{D}, \pmb{\eta}(\mathbf{v}_1)]\). This is facilitated by simulating sequentially for \(j=1,\ldots,N\):
\begin{align*}
[\pmb{\eta}(\mathbf{g}_j^*) \mid \mathbf{R}, \mathbf{D}, \pmb{\eta}(\mathbf{v}_1), \pmb{\eta}(\mathbf{g}_1^*), \ldots, \pmb{\eta}(\mathbf{g}_{j-1}^*)] \sim 
\mathcal{T}_p(\boldsymbol{\mu}_{2,j}(\cdot), c_{2,j}(\cdot,\cdot)\hat{\boldsymbol{\Sigma}}_{GLS,j}; n+j-m),
\end{align*}
the $p$-variate Student's $t$ distribution with degrees of freedom $n+j-m$. The exact forms of $\boldsymbol{\mu}_{2,j}(\cdot)$ and $c_{2,j}(\cdot,\cdot)\hat{\boldsymbol{\Sigma}}_{GLS,j};$ are provided in \citet{bhattacharya2007simulation}.

\textbf{Step 3:} For \(t = 2, \ldots, T\), simulate \(\mathbf{y}_t = \pmb{\eta}(\mathbf{v}_t)\) from \([\pmb{\eta}(\mathbf{v}_t) \mid \mathbf{R}, \mathbf{D}, \mathbf{D}^*]\), which is a multivariate Student's \(t\) distribution.

This procedure yields an approximate realisation from the joint posterior distribution of the entire dynamic sequence. The approximation becomes exact as the grid \(\mathbf{G}^*\) becomes infinitely fine.

\begin{figure}[htbp]
\centering
\begin{tikzpicture}[scale=0.9, transform shape]
    % Training data
    \node (train) [rectangle, draw, fill=blue!20, minimum width=2.5cm, text centered] at (0,2.5) {Training Data $\mathbf{D}$};
    \node (gp) [rectangle, draw, fill=blue!40, minimum width=2.5cm, text centered] at (0,1) {Gaussian Process $\pmb{\eta}(\cdot)$};
    \draw[->] (train) -- (gp);
    
    % Look-up table
    \node (grid) [rectangle, draw, fill=green!40, minimum width=2.5cm, text centered] at (4,2.5) {Fixed Grid $\mathbf{G}^*$};
    \node (lut) [rectangle, draw, fill=green!60, minimum width=2.5cm, text centered] at (4,1) {Look-up Table $\mathbf{D}^*$};
    \draw[->] (grid) -- (lut);
    \draw[->, dashed] (gp) -- (lut);
    
    % Dynamic sequence
    \node (y1) [circle, draw, fill=orange!40, minimum size=0.8cm] at (8,2) {$\mathbf{y}_1$};
    \node (y2) [circle, draw, fill=orange!60, minimum size=0.8cm] at (9.5,1.2) {$\mathbf{y}_2$};
    \node (y3) [circle, draw, fill=orange!80, minimum size=0.8cm] at (11,0.4) {$\mathbf{y}_3$};
    \node (y4) [circle, draw, fill=orange!100, minimum size=0.8cm] at (12.5,-0.4) {$\mathbf{y}_T$};
    
    \draw[->] (y1) -- (y2);
    \draw[->] (y2) -- (y3);
    \draw[->] (y3) -- (y4);
    
    % Conditional independence
    \draw[->, thick, red] (lut) to[out=0, in=180] (y1);
    \draw[->, thick, red] (lut) to[out=0, in=180] (y2);
    \draw[->, thick, red] (lut) to[out=0, in=180] (y3);
    \draw[->, thick, red] (lut) to[out=0, in=180] (y4);
    
    % Key insight
    \node at (6,-1.8) [text width=12cm, align=center, minimum height=1.5cm] {
        \textbf{Conditional Independence:} Given $\mathbf{D}^*$, $[\mathbf{y}_t \mid \mathbf{D}^*, \mathbf{D}, \mathbf{y}_{t-1}, \ldots] = [\mathbf{y}_t \mid \mathbf{D}^*, \mathbf{D}, \mathbf{y}_{t-1}]$ \\
        \small (Markov property emerges from look-up table)
    };
\end{tikzpicture}
\caption{Look-up table principle: Auxiliary variables on a fixed grid create conditional independence for sequential inference.}
\label{fig:lut}
\end{figure}
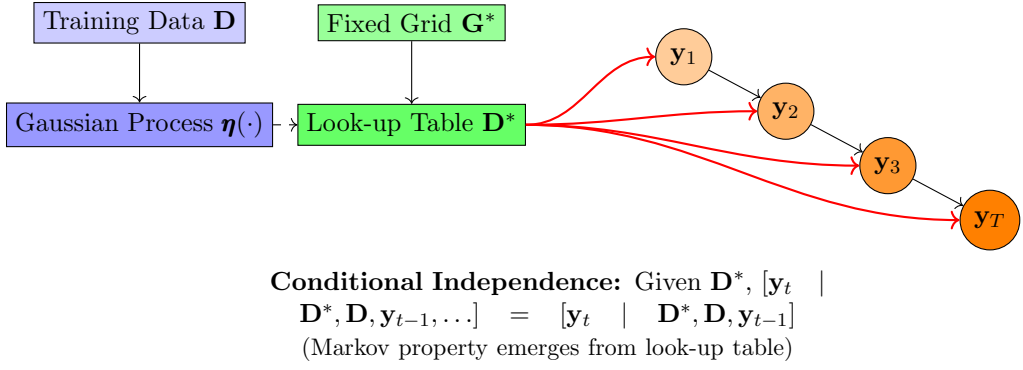

\subsection{Computational Efficiency and Stability}

A remarkable feature of this approach is its computational efficiency. The matrix \(\mathbf{A}_{(D^*,D)}\)—the covariance matrix of \((\mathbf{D}^*, \mathbf{D})\)—is fixed and can be inverted once before any simulations begin. All subsequent conditional simulations involve only this pre-computed inverse and simple vector operations.

In contrast, attempting to integrate out \(\mathbf{D}^*\) leads to an algorithm that requires inverting a new random matrix at each time step, inevitably leading to numerical instability and computational intractability after just a few steps. As \citet{bhattacharya2007simulation} demonstrates through extensive experiments, the look-up table approach remains stable even for long sequences, while the marginalised version breaks down completely.

\subsection{Parallelisability and Scalability}

Another crucial advantage is that the look-up table method is embarrassingly parallel. Each random sequence can be simulated independently, making the approach ideal for modern parallel computing architectures. As noted in the original work:

\begin{quote}
``The independence renders our dynamic emulator parallelisable, that is, the random sequences can be generated in independent parallel processors, although the original dynamic computer code may not be parallelisable.''
\end{quote}

This property anticipates the scalability requirements of modern AI systems, where parallel computation is essential for handling large-scale data streams. Notably, this parallelisability mirrors that of the ellipsoidal decomposition framework for iid sampling, where each processor independently handles a different ellipsoidal region. Together, these parallel architectures provide the computational muscle for the Bayesian reflex.

\subsection{Significance for the Bayesian Reflex}

The look-up table principle embodies all three mechanisms of the Bayesian reflex: belief maintenance, as the Gaussian process emulator maintains a full probabilistic representation of the unknown function with uncertainty quantified at every input point; sequential updating, as the algorithm proceeds step by step with each new prediction building upon previous ones through the recursive structure of the dynamic model; and uncertainty-driven action, as while the original application focuses on prediction rather than decision-making, the posterior distributions provide natural uncertainty quantification that could guide experimental design or adaptive sampling. This foundational work establishes that introducing latent auxiliary variables to create conditional independence is a powerful strategy for making sequential Bayesian inference tractable. As we shall see in the next sections, this same principle generalises naturally to deep architectures and inverse regression problems.

\section{Generalising the Look-Up Table: From Dynamic Emulators to Nonparametric State-Space Models}

The look-up table principle introduced by \citet{bhattacharya2007simulation} for dynamic computer model emulation proved to be useful in various aspects. A natural and powerful generalisation came with its application to Bayesian nonparametric state-space models \citep{ghosh2014bayesian}. This work recognised that the same fundamental challenge—sequential inference in a dynamic system with unknown functions—arises in state-space modelling, where both the observational and evolutionary equations may have unknown, time-varying forms. Similar approaches have been explored in the literature on nonparametric state-space models \citep{frigola2013bayesian, turner2010bayesian}.

\subsection{The State-Space Formulation}

Consider a general state-space model of the form:
\begin{align*}
y_t &= f_t(x_t) + \epsilon_t, \quad \epsilon_t \sim N(0, \sigma_\epsilon^2) \\
x_t &= g_t(x_{t-1}) + \eta_t, \quad \eta_t \sim N(0, \sigma_\eta^2)
\end{align*}
where \(y_t\) are the observed data, \(x_t\) are latent states, and both \(f_t\) and \(g_t\) are unknown functions that may vary with time. Traditional parametric approaches assume fixed functional forms (typically linear), but this assumption is often questionable, particularly because data on the latent states are unavailable for validation.

\citet{ghosh2014bayesian} proposed modelling both \(f_t\) and \(g_t\) as Gaussian processes, with the key insight that the look-up table methodology could be applied to handle the complex dependence structure induced by the unknown evolutionary function \(g_t\). This approach yields a realistic, non-Markovian dependence structure for the latent states—a significant advance over traditional state-space models that impose Markovian structure by assumption.

\subsection{The Look-Up Table for State-Space Models}

Following \citet{bhattacharya2007simulation}, the authors introduce a set of auxiliary variables \(\mathbf{D}_n^*\) representing function values on a fixed grid \(\mathbf{G}_n\). For the evolutionary function \(g\), this yields:

\begin{align*}
[x_t = g(x_{t-1}) + \eta_t \mid \mathbf{D}_n^*, x_{t-1}, \boldsymbol{\theta}_g, \sigma_\eta^2] \sim N(\mu_t, \sigma_t^2)
\end{align*}

with mean and variance given by GP regression formulas analogous to those in the original dynamic emulator work. Crucially, the correlation matrix \(\mathbf{A}_{g,D_n^*}\) of $\mathbf{D}_n^*$, is fixed and inverted once, enabling efficient sequential inference.

The authors prove a theorem establishing that the order of approximation of the true function by the conditional distribution given \(\mathbf{D}_n^*\) is \(O(n^{-1})\), where \(n\) is the grid size. This theoretical guarantee, combined with the computational efficiency of the look-up table approach, makes the methodology both sound and practical.

\subsection{The Non-Markovian Nature of the Latent States}

A crucial insight from this work is that, even though conditionally on \(\mathbf{D}_n^*\) the latent states exhibit Markov structure, the marginalised distribution is non-Markovian. As the authors note:

\begin{quote}
``As in Bhattacharya (2007), here also it is possible to marginalise out \(\mathbf{D}_n^*\) ... However, if \(\mathbf{D}_n^*\) is integrated out, it is clear that the conditional independence (Markov) property of \(x_t\)'s given \(\mathbf{D}_n^*\) -- will be lost. Thus, the marginalised conditional distribution of \(x_{t+1}\) depends upon \(\{x_k; k < t+1\}\), that is, the set of all the past state variables, unlike the non-marginalised case, where conditionally on \(\mathbf{D}_n^*\), \(x_{t+1}\) depends only upon \(x_t\). This makes it clear that even though for fixed (known) evolutionary functions \(g_t\) the corresponding equation satisfies the Markov property, such Markov property is lost when the function $g_t(x)=g(t,x)$ is modelled as a Gaussian process.''
\end{quote}

This observation has beneficial implications for the Bayesian reflex: it shows that introducing latent auxiliary variables can transform an intractable, fully dependent process into a conditionally independent structure that enables efficient sequential updating—while still preserving the realistic dependence structure upon marginalisation.

\subsection{Extension to Circular Latent States}

A further generalisation of the look-up table principle came with its application to models involving circular variables—quantities measured on a circle such as directions or times of day \citep{mazumder2016bayesian, mazumder2016nonparametric}. This extension is particularly relevant for real-world problems where latent processes are inherently circular, such as wind direction influencing wind speed, or ocean current direction affecting whale migration.

In the first of these works \citep{mazumder2016bayesian}, the authors consider a state-space model where the observed time series \(y_t\) is linear (e.g., wind speed), but the latent states \(x_t\) are circular (e.g., wind direction). The evolutionary equation becomes:
\begin{equation*}
x_t = \{g(t, x_{t-1}) + \eta_t\} [2\pi], \quad \eta_t \sim N(0, \sigma_\eta^2)
\end{equation*}
where \([2\pi]\) denotes the modulo \(2\pi\) operation that maps the linear sum back to the circle.

The look-up table methodology proves invaluable here. The authors construct a wrapped Gaussian process for the circular evolutionary function by first modelling its linear counterpart \(g^*\) using the standard GP framework, then applying the modulo operation. The auxiliary variables \(\mathbf{D}_z\) represent function values on a grid \(\mathbf{G}_z\) that spans both time and the circular domain.

A particularly elegant aspect of this work is the introduction of additional auxiliary variables \(K_t = \langle x_t^* / 2\pi \rangle\) (the ``wrapping numbers'') that track how many times the linearised process has wrapped around the circle. This transforms the intractable wrapped normal distribution into a conditionally Gaussian form:

\begin{align*}
[x_t \mid g^*(1,x_0), \beta_g, \sigma_\eta^2, \sigma_g^2, K_t] = \frac{\frac{1}{\sqrt{2\pi}\sigma_\eta}\exp\left(-\frac{1}{2\sigma_\eta^2}(x_t + 2\pi K_t - g^*(1,x_0))^2\right)I_{[0,2\pi]}(x_t)}{\Phi\left(\frac{2\pi(K_t+1) - g^*(1,x_0)}{\sigma_\eta}\right) - \Phi\left(\frac{2\pi K_t - g^*(1,x_0)}{\sigma_\eta}\right)}
\end{align*}

This hierarchical structure of auxiliary variables—first the look-up table \(\mathbf{D}_z\) representing the unknown function, then the wrapping numbers \(K_t\) handling the circular nature—exemplifies the recursive application of the Bayesian reflex. Each level of auxiliary variables creates conditional independence that makes sequential inference tractable, while preserving realistic dependence upon marginalisation.

\subsection{The Fully Circular Case}

The most challenging extension comes when both the observed and latent processes are circular \citep{mazumder2016nonparametric}. In this setting, the observational equation becomes:
\begin{equation*}
y_t = \{f(t, x_t) + \epsilon_t\} [2\pi], \quad \epsilon_t \sim N(0, \sigma_\epsilon^2)
\end{equation*}
adding another layer of complexity through the modulo operation on the observations.

This fully circular case requires introducing wrapping numbers \(N_t\) for the observations as well, creating a doubly intractable likelihood. Despite this formidable complexity, the look-up table methodology enables efficient MCMC inference. The fixed grid \(\mathbf{G}_z\) ensures that \(\mathbf{A}_{g,D_z}^{-1}\) can be pre-computed, and the conditional independence structures induced by the auxiliary variables make sequential updating feasible even in the most complex fully circular case.

\subsection{Empirical Validation}

All three papers provide extensive empirical validation of their methodologies. Simulation studies demonstrate that the posterior distributions successfully capture true parameter values, and the latent state posteriors assign high probability to the true states. Real data applications—including wind speed/direction analysis and ozone level modelling—show that the methods successfully predict future observations and recover latent circular processes even when the true functions are highly nonlinear and contain change points.

This ability to capture abrupt changes without explicit change-point modelling is a hallmark of the Bayesian reflex: by maintaining a flexible probabilistic representation and updating it sequentially, the system can adapt to unforeseen changes in the underlying dynamics.

\subsection{Significance for the Bayesian Reflex}

These extensions to circular variables demonstrate that the look-up table principle is not limited to linear, real-valued processes. The Bayesian reflex operates just as effectively when the quantities of interest live on a circle, with the modulo operation creating additional complexity that is tamed through the introduction of auxiliary wrapping numbers.

The hierarchical structure of auxiliary variables—look-up tables for unknown functions, wrapping numbers for circular transformations—illustrates a general principle: complex, intractable dependencies can be transformed into tractable sequential updating problems by introducing latent variables at multiple levels. Each level of auxiliary variables creates conditional independence that makes the next level of inference feasible, while the full joint distribution, upon marginalisation, retains the realistic dependence structure of the original process.

This recursive application of the Bayesian reflex—using auxiliary variables to create conditional independence, then marginalising to recover realistic dependence—provides a template for building complex adaptive systems that can learn online while maintaining theoretical guarantees and computational tractability.

\section{From Look-Up Tables to Recursive Gaussian Processes}

The look-up table principle introduced by \citet{bhattacharya2007simulation} for single-layer dynamic systems has been extended to deep architectures, yielding the Recursive Gaussian Process (RGP) framework \citep{bhattacharya2025bayesian}. RGPs represent a significant advance in Bayesian deep learning, offering a principled approach to hierarchical modelling with uncertainty quantification that naturally lends itself to online inference. Other deep Bayesian nonparametric models include deep Gaussian processes \citep{damianou2013deep} and neural processes \citep{garnelo2018conditional}.

\subsection{The RGP Architecture: A Hierarchical State-Space Model}

The RGP framework constructs a deep neural architecture where each hidden layer is governed by a Gaussian process prior, building upon the look-up table concept recursively. For a network with \(T\) hidden layers, let \(\mathbf{x} \in \mathbb{R}^p\) denote the input. The hidden activations are defined recursively as: for $i=1,\ldots,n$, $n$ being the number of individuals,

\begin{align*}
h_{ij}^{(0)} &= g(0, \mathbf{w}^{(0)'}\mathbf{x}_i + b_j^{(0)}), \quad j = 1,\ldots,K \\
h_{ij}^{(t)} &= g(t, \mathbf{w}^{(t)'}\mathbf{h}_i^{(t-1)} + b_j^{(t)}), \quad t = 1,\ldots,T; \quad j = 1,\ldots,K \\
y_{ij} &= f(\tilde{\mathbf{w}}'\mathbf{h}_i^{(T)} + \tilde{b}_j) + \epsilon_{ij}, \quad \epsilon_{ij} \sim \mathcal{N}(0, \sigma_\epsilon^2)
\end{align*}

Here, \(g(t, \cdot)\) (induced by a single Gaussian process $g(\cdot,\cdot)$) and \(f(\cdot)\) are Gaussian processes with appropriate mean and covariance functions, while the weights \(\mathbf{w}^{(t)}\) and biases \(b_j^{(t)}\) are assigned spike-and-slab priors that enable automatic variable selection.

\begin{figure}[htbp]
\centering
\begin{tikzpicture}[node distance=2cm, scale=0.9, transform shape]
    % Input layer
    \node (input) at (0,0) {$\mathbf{x}$};
    
    % Hidden layers with GP
    \node (h1) at (2.5,0) [rectangle, draw, fill=blue!20, minimum width=2.2cm, minimum height=1.2cm] {$\mathbf{h}^{(1)}$};
    \node (h2) at (5,0) [rectangle, draw, fill=blue!40, minimum width=2.2cm, minimum height=1.2cm] {$\mathbf{h}^{(2)}$};
    \node (h3) at (7.5,0) [rectangle, draw, fill=blue!60, minimum width=2.2cm, minimum height=1.2cm] {$\mathbf{h}^{(3)}$};
    \node (hT) at (10,0) [rectangle, draw, fill=blue!80, minimum width=2.2cm, minimum height=1.2cm] {$\mathbf{h}^{(T)}$};
    
    % GP functions
    \node (gp1) at (3.75,1.8) [rectangle, draw, fill=green!40, minimum width=2cm, minimum height=1cm] {$g(1, \cdot)$};
    \node (gp2) at (6.25,1.8) [rectangle, draw, fill=green!40, minimum width=2cm, minimum height=1cm] {$g(2, \cdot)$};
    \node (gp3) at (8.75,1.8) [rectangle, draw, fill=green!40, minimum width=2cm, minimum height=1cm] {$g(3, \cdot)$};
    \node (gpT) at (11.25,1.8) [rectangle, draw, fill=green!60, minimum width=2cm, minimum height=1cm] {$g(T, \cdot)$};
    
    \draw[->, dashed] (gp1) -- (h1);
    \draw[->, dashed] (gp2) -- (h2);
    \draw[->, dashed] (gp3) -- (h3);
    \draw[->, dashed] (gpT) -- (hT);
    
    % Output layer
    \node (output) at (12.5,0) [rectangle, draw, fill=red!40, minimum width=2.2cm, minimum height=1.2cm] {$\mathbf{y}$};
    \node (f) at (13.75,1.8) [rectangle, draw, fill=green!80, minimum width=2cm, minimum height=1cm] {$f(\cdot)$};
    \draw[->, dashed] (f) -- (output);
    
    % Connections between layers
    \draw[->] (input) -- (h1);
    \draw[->] (h1) -- (h2);
    \draw[->] (h2) -- (h3);
    \draw[->] (h3) -- (hT);
    \draw[->] (hT) -- (output);
    
    % Look-up table shared across layers
    \node (lut) at (6.25,-2.5) [rectangle, draw, fill=purple!40, minimum width=9cm, minimum height=1.2cm] {Shared Look-up Table $\mathbf{D}_m$};
    
    \draw[->, thick, red] (lut) -- (h1);
    \draw[->, thick, red] (lut) -- (h2);
    \draw[->, thick, red] (lut) -- (h3);
    \draw[->, thick, red] (lut) -- (hT);
    \draw[->, thick, red] (lut) -- (output);
    
    % State-space interpretation
    \node at (6.25,-4) [text width=12cm, text centered] {
        \textbf{State-Space Interpretation:} $\mathbf{h}^{(t)}$ are latent states, $g(t,\cdot)$ are unknown evolutionary functions
    };
\end{tikzpicture}
\caption{Recursive Gaussian Process architecture: Hierarchical structure with shared look-up table enabling online inference.}
\label{fig:rgp}
\end{figure}
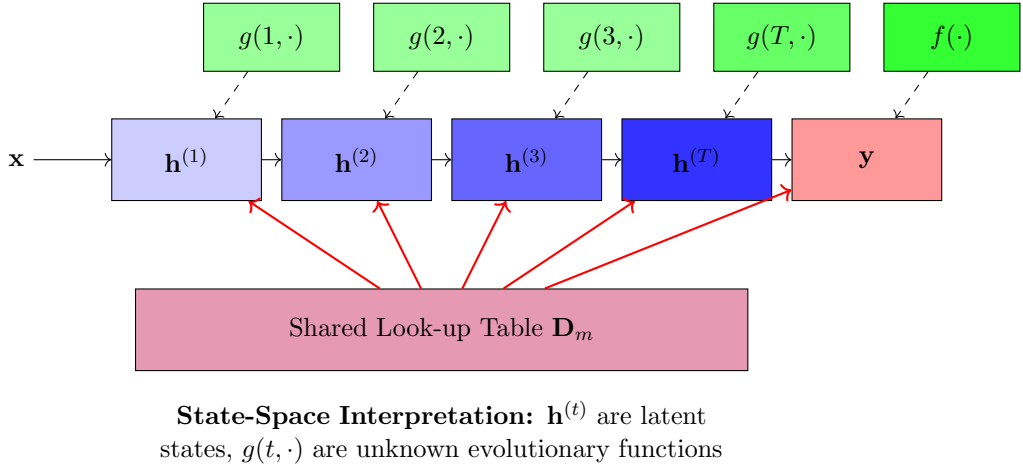

This hierarchical structure (Figure \ref{fig:rgp}) bears a striking resemblance to state-space models, which are fundamental to online learning. The hidden activations \(\mathbf{h}_i^{(t)}\) can be interpreted as latent states that evolve through the deterministic yet uncertain transformations provided by the GPs. This perspective immediately suggests how RGPs can power online learning systems: as new data arrives, we can perform forward filtering to update the latent states and then use Bayesian updating to refine our beliefs about the underlying functions.

\subsection{The Look-Up Table for RGPs}

One might anticipate that extending the look-up table concept to multiple layers requires introducing auxiliary variables at each level; however, as shown in \citep{bhattacharya2025bayesian}, a single grid \(\mathbf{G}_m = \{(t, z_1), \ldots, (t, z_m)\}\) and corresponding function values \(\mathbf{D}_m = \{g(t, z_1), \ldots, g(t, z_m)\}\) suffice. 

The key insight is that, conditional on the single look-up table induced by the single Gaussian process $g(\cdot,\cdot)$, the activations at each layer exhibit a Markov property: given \(\mathbf{D}_m\), the distribution of \(h_{ij}^{(t)}\) depends only on \(\mathbf{h}_i^{(t-1)}\) and the current layer's parameters, not on earlier layers.

This layered conditional independence structure is precisely what enables efficient sequential inference in deep architectures. It still exploits simply the original single look-up table principle in the hierarchical model setup, creating a foundation for online Bayesian deep learning.

\subsection{Toward Online Inference with RGPs}

While the original RGP framework was developed for batch-mode inference using Markov chain Monte Carlo, its architecture provides a natural foundation for online Bayesian learning. Several pathways exist for adapting RGPs to sequential settings.

\subsubsection{Particle filtering approaches:} The hierarchical structure of RGPs makes them amenable to sequential Monte Carlo methods. We can maintain a set of weighted particles representing the joint distribution of all latent variables—the hidden activations \(\mathbf{h}_i^{(t)}\), the look-up table \(\mathbf{D}_m\), and the parameters \(\{\mathbf{w}^{(t)}, b_j^{(t)}\}\). As each new observation arrives, particles are propagated through the layers and reweighted according to the likelihood, with resampling to combat degeneracy. The look-up table provides compact representations that can be updated incrementally.

\subsubsection{Sequential variational inference:} For scalability, variational approximations can be employed in an online fashion. Let \(q_t(\boldsymbol{\theta})\) denote the variational posterior after processing \(t\) observations, where \(\boldsymbol{\theta}\) encompasses all model parameters. When a new data point \((\mathbf{x}_{t+1}, y_{t+1})\) arrives, we update:

\begin{align*}
q_{t+1} = \arg\min_{q \in \mathcal{Q}} \KL{ q(\boldsymbol{\theta}) }{ \frac{p(y_{t+1} | \mathbf{x}_{t+1}, \boldsymbol{\theta}) q_t(\boldsymbol{\theta}) }{Z_{t+1}} }
\end{align*}

where \(Z_{t+1}\) is the normalising constant. For exponential family variational distributions, this reduces to simple moment updates. The spike-and-slab priors on weights further enable automatic pruning of irrelevant connections as the model learns online.

\subsubsection{Recursive Bayesian estimation:} The recursive structure of RGPs—where each layer's output depends only on the previous layer—suggests the possibility of closed-form or approximately closed-form updates for the posterior parameters. This mirrors the elegance of the Kalman filter but extended to the nonparametric, nonlinear setting. The look-up table representation transforms the infinite-dimensional GP inference problem into a finite-dimensional one, making recursive estimation tractable.

\subsection{Theoretical Guarantees and Model Misspecification}

A remarkable feature of the RGP framework is its robustness to model misspecification. The universal approximation theorem for deep neural networks guarantees that, for sufficiently large depth \(T\), there exist functions within the class of neural networks that approximate any continuous target function arbitrarily well \citep{cybenko1989approximation, hornik1989multilayer}. The Bayesian asymptotic theory developed for RGPs, exploiting the universal approximation theorem, shows that the posterior converges at a near-optimal rate in an appropriate sense, even when the true data-generating mechanism lies outside the prior support \citep{bhattacharya2025bayesian}.

For online learning, these guarantees translate into confidence that as data accumulates sequentially, the system's beliefs will converge to the true underlying dynamics under mild conditions. The Kullback-Leibler divergence rate \(\mathfrak{h}(\boldsymbol{\theta})\) provides a measure of model adequacy, and the posterior predictive distribution satisfies:

\begin{align*}
\lim_{n \to \infty} \rho_H^2(P^{nq}, F_\pi^{nq}) \leq \mathfrak{h}(\boldsymbol{\Theta})
\end{align*}

where \(\rho_H\) is the Hellinger distance, \(P^{nq}\) is the true predictive distribution, \(F_\pi^{nq}\) is the posterior predictive, and $\mathfrak{h}(\boldsymbol{\Theta})$
is the essential infimum of \(\mathfrak{h}(\boldsymbol{\theta})\) over the parameter space $\boldsymbol{\Theta}$. This bound can be made arbitrarily small by choosing sufficient depth \(T\), ensuring that the online learning system based on RGPs will make increasingly accurate predictions over time.

\subsection{Significance for the Bayesian Reflex}

The RGP framework exemplifies the Bayesian reflex at multiple levels, building directly on the foundational look-up table principle. At the architectural level, the recursive composition of functions mirrors the sequential nature of belief updating. At the inference level, the shared look-up table induces conditional independence structures that enable efficient online computation. At the theoretical level, the asymptotic guarantees ensure that the system's uncertainty-driven learning will converge to the truth. As online learning systems grapple with increasingly complex, high-dimensional data streams, architectures like RGPs that combine the flexibility of deep learning with principled Bayesian uncertainty quantification will become essential. The RGP framework provides a template for building AI systems that not only learn continuously but do so with calibrated uncertainty and theoretical guarantees—the very essence of the Bayesian reflex.

\section{Ellipsoidal Decomposition Meets Recursive Gaussian Processes: A Unified Computational Foundation}

The two foundational computational principles we have introduced—the look-up table principle for sequential inference in function space, and the ellipsoidal decomposition framework for nearly perfect iid sampling in parameter space—are not merely complementary; they can be integrated to create a unified foundation for the Bayesian reflex. This integration opens exciting possibilities for online learning systems that are both exact and scalable.

\subsection{Integrating Near-Perfect IID Sampling into RGP Inference}

The Recursive Gaussian Process framework, as originally developed, relies on MCMC for inference. However, the ellipsoidal decomposition framework offers a compelling alternative: instead of running MCMC chains with uncertain convergence, we can generate almost perfect iid samples from the RGP posterior using the methods of Section \ref{sec:iid}.

Consider the RGP posterior distribution over all latent variables \(\boldsymbol{\Theta} = \{\mathbf{h}_i^{(t)}, \mathbf{D}_m, \mathbf{w}^{(t)}, b_j^{(t)}\}\) given observed data \(\mathcal{D}\). This distribution lives on a high-dimensional space (potentially thousands of dimensions), but the ellipsoidal decomposition framework is specifically designed for such challenges. The key steps would be:

\begin{enumerate}
    \item Estimate the location \(\boldsymbol{\mu}\) and scale \(\boldsymbol{\Sigma}\) of the RGP posterior from a preliminary run or using variational approximations.
    \item Construct ellipsoidal regions \(\mathbf{A}_i\) in the parameter space of \(\boldsymbol{\Theta}\).
    \item For each region, compute Monte Carlo estimates \(\hat{\pi}(\mathbf{A}_i)\) using uniform samples transformed via the RGP likelihood.
    \item Generate near-perfect iid samples from the RGP posterior using Algorithm \ref{alg:iid_sampling}.
\end{enumerate}

The parallel nature of both frameworks makes this integration particularly attractive: each ellipsoidal region can be processed independently, and within each region, the look-up table structure enables efficient computation of the RGP likelihood.

\subsection{Online Learning with Near-Perfect IID Updates}

With the ability to generate perfect iid samples from the posterior at any time, online learning in the RGP framework becomes straightforward:

\begin{algorithm}[H]
\caption{Online RGP Learning with Near-Perfect IID Updates}
\label{alg:online_rgp_iid}
\begin{algorithmic}[1]
\State Initialise with prior distribution over RGP parameters \(\boldsymbol{\Theta}_0\)
\For{\(t = 1,2,\ldots\) (as new data arrives)}
    \State Observe new data point \((\mathbf{x}_t, y_t)\)
    \State Update posterior to \(\pi_t(\boldsymbol{\Theta}) \propto p(y_t | \mathbf{x}_t, \boldsymbol{\Theta}) \pi_{t-1}(\boldsymbol{\Theta})\)
    \State \textbf{Parallel Step:} Generate \(N\) near-perfect iid samples \(\{\boldsymbol{\Theta}_t^{(1)}, \ldots, \boldsymbol{\Theta}_t^{(N)}\}\) from \(\pi_t\) using Algorithm \ref{alg:iid_sampling}
    \State Use samples for prediction, decision-making, or uncertainty quantification
\EndFor
\end{algorithmic}
\end{algorithm}

This algorithm maintains exact posterior representations throughout the learning process, with no approximation error beyond the controlled \(\epsilon\) in the ellipsoidal decomposition. The computational cost at each step is amortised through parallel processing, making it feasible for real-time applications.

\subsection{Theoretical Synergies}

The integration of ellipsoidal decomposition with RGPs yields several theoretical benefits. First, exact uncertainty quantification: unlike variational approximations that underestimate uncertainty, near-perfect iid samples provide almost exact posterior summaries, including credible intervals and predictive distributions. Second, convergence-free inference: the ellipsoidal decomposition eliminates the need for MCMC convergence diagnostics, which are particularly problematic in high-dimensional, hierarchical models like RGPs. Third, principled model comparison: almost perfect iid samples enable practically exact computation of marginal likelihoods and Bayes factors through bridge sampling or importance sampling, facilitating rigorous model selection in online settings. Fourth, regret bound validation: with exact posterior samples, we can empirically validate theoretical regret bounds for Thompson sampling and other decision-making algorithms, potentially leading to tighter constants and improved understanding.

\subsection{A Unified Computational Architecture}

Figure \ref{fig:unified} illustrates how the two principles can be combined to form a unified computational architecture for the Bayesian reflex. The look-up table principle handles the sequential dynamics: it decomposes the temporal evolution into conditionally independent steps, enabling efficient forward filtering. The ellipsoidal decomposition principle handles the parameter inference: it decomposes the parameter space into manageable regions, enabling near-perfect iid sampling from the posterior at each time step. Parallel processing unifies both: the look-up table enables parallel simulation of multiple sequences, while the ellipsoidal decomposition enables parallel sampling across parameter regions.

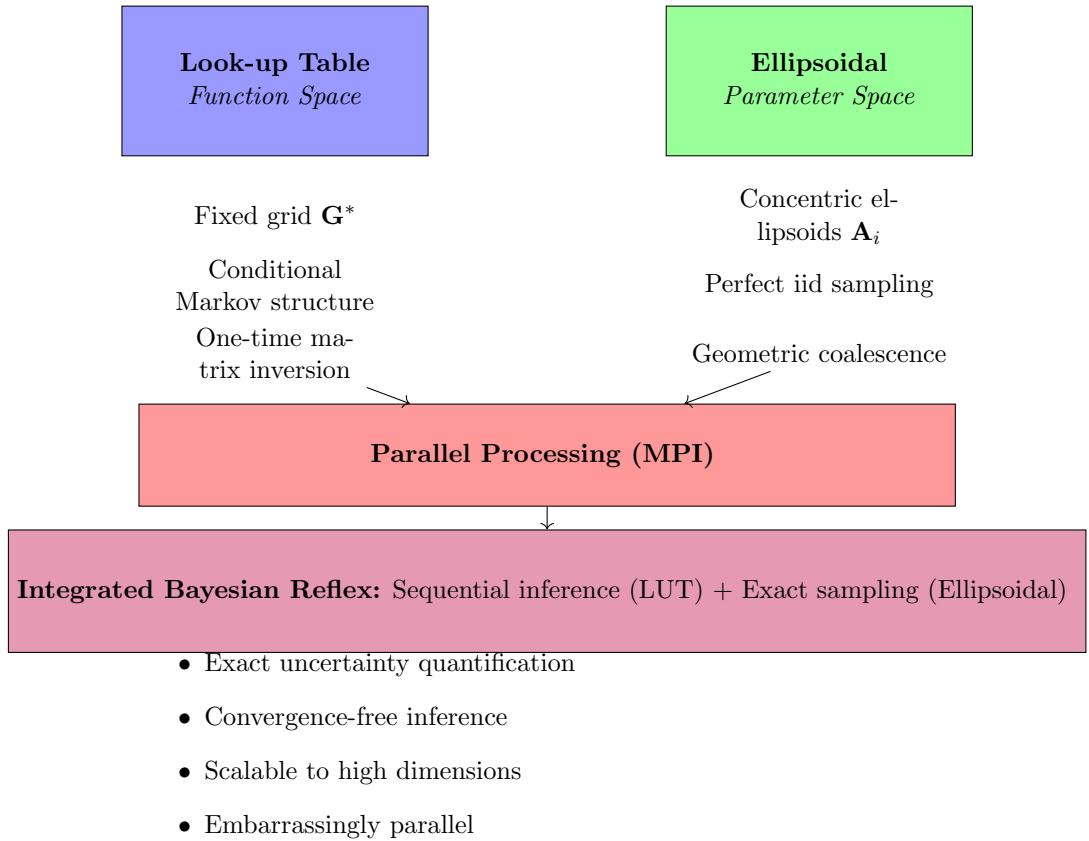
\begin{figure}[htbp]
\centering
\begin{tikzpicture}[node distance=2.5cm, scale=0.9, transform shape]
    % Two main pillars
    \node (lut) at (-4,0) [rectangle, draw, fill=blue!40, minimum width=4.5cm, minimum height=2.2cm, align=center] {
        \textbf{Look-up Table} \\
        \textit{Function Space}
    };
    
    \node (ellip) at (4,0) [rectangle, draw, fill=green!40, minimum width=4.5cm, minimum height=2.2cm, align=center] {
        \textbf{Ellipsoidal} \\
        \textit{Parameter Space}
    };
    
    % Key features
    \node (lut1) at (-4,-2) [text width=4cm, align=center] {Fixed grid $\mathbf{G}^*$};
    \node (lut2) at (-4,-3) [text width=4cm, align=center] {Conditional Markov structure};
    \node (lut3) at (-4,-4) [text width=4cm, align=center] {One-time matrix inversion};
    
    \node (ell1) at (4,-2) [text width=4cm, align=center] {Concentric ellipsoids $\mathbf{A}_i$};
    \node (ell2) at (4,-3) [text width=4cm, align=center] {Perfect iid sampling};
    \node (ell3) at (4,-4) [text width=4cm, align=center] {Geometric coalescence};
    
    % Unifying element - Parallel processing
    \node (parallel) at (0,-5.5) [rectangle, draw, fill=red!40, minimum width=12cm, minimum height=1.5cm, align=center] {
        \textbf{Parallel Processing (MPI)}
    };
    
    \draw[->] (lut3) -- (parallel);
    \draw[->] (ell3) -- (parallel);
    
    % Integration
    \node (integrated) at (0,-7.5) [rectangle, draw, fill=purple!40, minimum width=12cm, minimum height=1.8cm, align=center] {
        \textbf{Integrated Bayesian Reflex:} Sequential inference (LUT) + Exact sampling (Ellipsoidal)
    };
    \draw[->] (parallel) -- (integrated);
    
    % Outcomes
    \node at (0,-9.5) [text width=12cm, align=center] {
        \begin{itemize}
            \item Exact uncertainty quantification
            \item Convergence-free inference
            \item Scalable to high dimensions
            \item Embarrassingly parallel
        \end{itemize}
    };
\end{tikzpicture}
\caption{Unified computational architecture combining look-up table and ellipsoidal decomposition principles.}
\label{fig:unified}
\end{figure}

This architecture (Figure \ref{fig:unified}) realises the vision of the Bayesian reflex as a truly automatic, continuous learning system: as data streams in, the system updates its beliefs exactly, maintains perfect uncertainty quantification, and can make optimal decisions at any moment—all without the convergence concerns that have plagued Bayesian computation for decades.

\section{Convergence with Deep Learning: Bayesian Neural Networks in the Online Setting}

\subsection{Bayesian Neural Networks: Mathematical Formulation}

A Bayesian neural network (BNN) places a prior distribution over network weights \(p(\mathbf{w})\) and aims to compute the posterior given data \citep{neal2012bayesian, mackay1992practical}. For a network with \(L\) layers, let \(\mathbf{w} = \{\mathbf{W}_1, \ldots, \mathbf{W}_L, \mathbf{b}_1, \ldots, \mathbf{b}_L\}\) denote all weights and biases.

For input \(x\), the BNN defines a conditional distribution \(p(y \mid x, \mathbf{w}) = \ell(y; f_{\mathbf{w}}(x))\), where \(f_{\mathbf{w}}\) is the network output and \(\ell\) is an appropriate likelihood function such as Gaussian for regression or categorical for classification \citep{neal2012bayesian}. The prior \(p(\mathbf{w})\) is typically chosen as a factorised Gaussian: \(p(\mathbf{w}) = \prod_{i} \mathcal{N}(w_i; 0, \sigma_i^2)\). The posterior given data \(\mathcal{D} = \{(x_n, y_n)\}_{n=1}^N\) is \(p(\mathbf{w} \mid \mathcal{D}) \propto p(\mathbf{w}) \prod_{n=1}^N p(y_n \mid x_n, \mathbf{w})\).

Predictions are made by Bayesian model averaging \citep{ghahramani2015probabilistic}: \(p(y \mid x, \mathcal{D}) = \int p(y \mid x, \mathbf{w}) p(\mathbf{w} \mid \mathcal{D}) \, d\mathbf{w}\). This integration provides natural uncertainty quantification through the variance decomposition \citep{gal2016dropout}: \(\Var[y \mid x, \mathcal{D}] = \mathbb{E}_{\mathbf{w} \mid \mathcal{D}}[\Var[y \mid x, \mathbf{w}]] + \Var_{\mathbf{w} \mid \mathcal{D}}[\mathbb{E}[y \mid x, \mathbf{w}]]\), which separates aleatoric uncertainty (inherent noise in the data) from epistemic uncertainty (uncertainty about model parameters due to limited data).

\subsection{Variational Inference for BNNs}

Since exact inference is intractable, variational approximations are commonly used \citep{blundell2015weight, kingma2015variational}. Let \(q_\phi(\mathbf{w})\) be a variational distribution with parameters \(\phi\). The ELBO to maximise is \(\mathcal{L}(\phi) = \sum_{n=1}^N \mathbb{E}_{q_\phi}[\log p(y_n \mid x_n, \mathbf{w})] - \KL{q_\phi(\mathbf{w})}{p(\mathbf{w})}\).

For online learning, we use sequential variational inference \citep{nguyen2018variational}. After observing data up to time \(t-1\), we have \(q_{t-1}(\mathbf{w})\). When new data \((x_t, y_t)\) arrives, we update: \(q_t = \arg\min_{q \in \mathcal{Q}} \KL{q(\mathbf{w})}{\frac{p(y_t \mid x_t, \mathbf{w}) q_{t-1}(\mathbf{w})}{Z_t}}\), where \(Z_t\) is the normalising constant.

The variational continual learning (VCL) update incorporates a forgetting factor \citep{nguyen2018variational}: \(q_t(\mathbf{w}) \propto p(y_t \mid x_t, \mathbf{w}) q_{t-1}(\mathbf{w})^\lambda\), with \(\lambda \in [0,1]\). For Gaussian variational approximations, \(q(\mathbf{w}) = \mathcal{N}(\mathbf{w} \mid \boldsymbol{\mu}, \boldsymbol{\Sigma})\), the VCL update can be implemented by optimising \(\mathcal{L}_t(\boldsymbol{\mu}, \boldsymbol{\Sigma}) = \mathbb{E}_{q}[\log p(y_t \mid x_t, \mathbf{w})] - \lambda \KL{q(\mathbf{w})}{q_{t-1}(\mathbf{w})}\).

While variational methods are scalable, they introduce approximation error that can accumulate over time. The ellipsoidal decomposition framework offers a potential path toward near-exact online inference in BNNs, at least for networks of moderate size. By treating the weight posterior as the target distribution \(\pi(\mathbf{w})\) and applying Algorithm \ref{alg:iid_sampling}, we could generate practically perfect iid samples from the exact posterior at each time step, eliminating approximation error almost entirely.

\subsection{In-Context Learning as Implicit Bayesian Inference}

Transformers trained on diverse tasks can perform in-context learning, which can be understood as implicit Bayesian inference \citep{dai2024context, garg2022what, xie2021explanation}. Consider a transformer that maps an input \(x\) and context \(\mathcal{D}_{\text{context}} = \{(x_j, y_j)\}_{j=1}^k\) to a prediction \(\hat{y} = f_{\text{TF}}(x, \mathcal{D}_{\text{context}})\). This is analogous to Bayesian prediction: \(p(y \mid x, \mathcal{D}_{\text{context}}) = \int p(y \mid x, \theta) p(\theta \mid \mathcal{D}_{\text{context}}) \, d\theta\).

Recent theoretical work shows that transformers with appropriate architecture can implement gradient-based learning algorithms \citep{von2023transformers, akyurek2022what}. For a linear self-attention layer, the update can be written as \(\text{Attention}(Q,K,V) = \text{softmax}\left(\frac{QK^\top}{\sqrt{d}}\right) V\), which resembles a kernel smoother. With multiple layers, transformers can implement iterative algorithms that converge to the Bayes-optimal predictor under certain conditions \citep{xie2021explanation}.

The ellipsoidal decomposition framework provides a way to validate and potentially improve in-context learning by providing ground-truth Bayesian predictions for comparison. For small to moderate-sized problems, we could compute practically exact posterior predictive distributions using iid samples and compare them to transformer outputs, shedding light on the implicit inference mechanisms learned by these models.

\section{Decision-Making Under Uncertainty: Bandits and Beyond}

\subsection{Multi-Armed Bandits: Mathematical Framework}

The multi-armed bandit problem provides a mathematical abstraction for sequential decision-making under uncertainty \citep{lattimore2020bandit, bubeck2012regret}. An agent chooses among \(K\) actions (arms) at each time step \(t=1,\ldots,T\). Each action \(a\) yields a random reward \(r_t(a)\) drawn from an unknown distribution with mean \(\mu_a\). The regret of a policy \(\pi\) that selects actions \(A_t\) is \(R_T(\pi) = \sum_{t=1}^T \mu^* - \mathbb{E}[\mu_{A_t}]\), where \(\mu^* = \max_a \mu_a\) is the mean reward of the optimal arm \citep{lattimore2020bandit}.

\subsection{Thompson Sampling: Algorithm and Analysis}
\label{subsec:thompson_sampling}

Thompson sampling is a Bayesian algorithm for the bandit problem \citep{thompson1933likelihood, russo2018tutorial}. Let \(p(\theta)\) be a prior over parameters, and let \(p(\theta \mid \mathcal{H}_t)\) be the posterior after observing history \(\mathcal{H}_t = \{(A_s, r_s)\}_{s=1}^{t-1}\).

\begin{algorithm}
\caption{Thompson Sampling \citep{thompson1933likelihood}}
\label{alg:thompson}
\begin{algorithmic}[1]
\For{\(t = 1\) to \(T\)}
    \State Sample \(\tilde{\theta}_t \sim p(\theta \mid \mathcal{H}_t)\)
    \State Select \(A_t = \arg\max_a \mathbb{E}[r \mid a, \tilde{\theta}_t]\)
    \State Observe reward \(r_t\)
    \State Update posterior \(p(\theta \mid \mathcal{H}_{t+1}) \propto p(r_t \mid A_t, \theta) p(\theta \mid \mathcal{H}_t)\)
\EndFor
\end{algorithmic}
\end{algorithm}

For Bernoulli bandits with Beta prior, Thompson sampling has a closed form. For Bernoulli bandits, Thompson sampling achieves \(\mathbb{E}[R_T] \leq \sum_{a: \mu_a < \mu^*} \frac{\log T}{\Delta_a} + \mathcal{O}(1)\), where \(\Delta_a = \mu^* - \mu_a\) \citep{agrawal2012analysis}. %This matches the lower bound up to constants. 
For general reward distributions, the regret bound depends on the Eluder dimension \citep{russo2014learning}: \(\mathbb{E}[R_T] = \tilde{\mathcal{O}}(\sqrt{d_E T \log K})\).

\begin{figure}[htbp]
\centering
\begin{tikzpicture}[node distance=2.5cm, scale=0.9, transform shape]
    % Initial state
    \node (prior) [rectangle, draw, fill=orange!20, minimum width=3cm, text centered] {Prior $p(\theta)$};
    
    % Loop
    \node (post) [rectangle, draw, fill=green!20, below of=prior, minimum width=3cm, text centered] {Posterior $p(\theta|\mathcal{H}_t)$};
    \node (sample) [rectangle, draw, fill=blue!20, below of=post, minimum width=3cm, text centered] {Sample $\tilde{\theta}_t \sim p(\theta|\mathcal{H}_t)$};
    \node (action) [rectangle, draw, fill=red!20, below of=sample, minimum width=3cm, text centered] {Select $A_t = \arg\max_a \mathbb{E}[r|a,\tilde{\theta}_t]$};
    \node (reward) [rectangle, draw, fill=yellow!20, below of=action, minimum width=3cm, text centered] {Observe reward $r_t$};
    
    % Connections
    \draw[->] (prior) -- node[right] {Initial} (post);
    \draw[->] (post) -- (sample);
    \draw[->] (sample) -- (action);
    \draw[->] (action) -- (reward);
    \draw[->] (reward) -- ++(2.5,0) |- node[right, pos=0.25] {Update} (post);
    
    % Regret formula
    \node at (0,-9.5) {$\mathbb{E}[R_T] = \tilde{\mathcal{O}}(\sqrt{d_E T \log K})$};
    
    % Ellipsoidal enhancement
    \node at (0,-11) [rectangle, draw, fill=purple!20, text width=9cm, text centered, minimum height=1.2cm] {
        \textbf{Ellipsoidal Enhancement:} Near-perfect iid samples from posterior enable essentially exact Thompson sampling
    };
\end{tikzpicture}
\caption{Thompson sampling for sequential decision-making under uncertainty.}
\label{fig:thompson}
\end{figure}
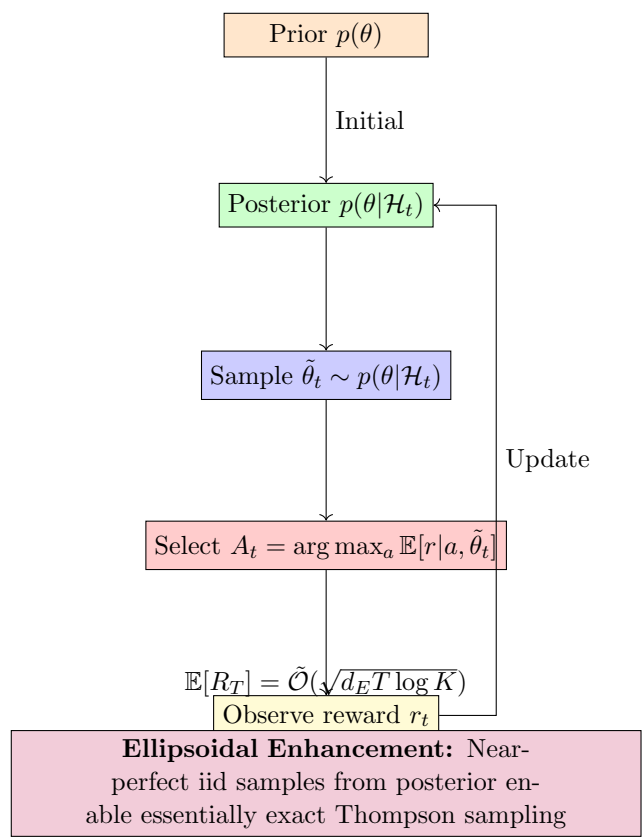

The ellipsoidal decomposition framework can significantly enhance Thompson sampling (Figure \ref{fig:thompson}) by providing practically exact iid samples from the posterior \(p(\theta \mid \mathcal{H}_t)\) at each time step. Traditional implementations use MCMC or variational approximations, which introduce errors that propagate through time. With near-perfect iid samples, the Thompson sampling algorithm becomes essentially exact, potentially improving regret bounds and enabling tighter theoretical analysis.

\subsection{Collaborative Bandits and Matrix Factorization}

In collaborative bandits, we have \(U\) users and \(A\) actions \citep{dasgupta2025bayesian, li2010contextual}. Let \(R \in \mathbb{R}^{U \times A}\) be the reward matrix, assumed low-rank: \(R = UV^\top\) with \(U \in \mathbb{R}^{U \times d}\), \(V \in \mathbb{R}^{A \times d}\), and \(d \ll \min(U,A)\). The Bayesian model places priors on the latent factors \citep{dasgupta2025bayesian, mnih2008probabilistic}: \(U_i \sim \mathcal{N}(0, \sigma_U^2 I_d)\) for \(i=1,\ldots,U\), \(V_j \sim \mathcal{N}(0, \sigma_V^2 I_d)\) for \(j=1,\ldots,A\), and \(R_{ij} \mid U_i, V_j \sim \mathcal{N}(U_i^\top V_j, \sigma^2)\).

Posterior inference proceeds via Gibbs sampling or variational inference \citep{mnih2008probabilistic}. Thompson sampling samples factors from the posterior and selects actions maximising the sampled reward \citep{dasgupta2025bayesian}: \(a_t = \arg\max_a \tilde{U}_{u_t}^\top \tilde{V}_a\), where \(u_t\) is the current user and \((\tilde{U}, \tilde{V})\) are sampled factors.

The matrix factorization posterior in collaborative bandits is a natural candidate for the ellipsoidal decomposition framework. The parameter space \((\{U_i\}, \{V_j\})\) can be high-dimensional, but the framework's ability to handle dimensions up to 4776 (and potentially higher with diffeomorphic transformations) makes it applicable. Near-perfect iid samples from the posterior would eliminate the need for Gibbs sampling (which may mix slowly) or variational approximations (which underestimate uncertainty).

\subsection{Restless Bandits and State-Space Models}

The restless multi-armed bandit (RMAB) extends the framework to settings where arms evolve over time \citep{whittle1988restless, tekin2010bayesian}. Each arm \(a\) has a state \(s_t^{(a)}\) that evolves according to Markov dynamics: \(s_{t+1}^{(a)} \sim P_a(\cdot \mid s_t^{(a)}, a_t)\), where \(a_t\) indicates whether the arm was pulled.

The Bayesian approach uses a state-space model for each arm \citep{liang2025context, osband2016deep}: \(s_t \sim p(s_t \mid s_{t-1}, \theta)\) and \(r_t \sim p(r_t \mid s_t, \phi)\). Posterior inference over states and parameters is performed using particle filters or Kalman filters \citep{liang2025context}. For linear Gaussian dynamics: \(s_t = A s_{t-1} + B a_t + w_t\) with \(w_t \sim \mathcal{N}(0, Q)\) and \(r_t = C s_t + D a_t + v_t\) with \(v_t \sim \mathcal{N}(0, R)\), the Kalman filter provides exact posterior updates.

Thompson sampling then proceeds by sampling states and parameters from the current posterior and solving the resulting known Markov Decision Process \citep{liang2025context}. The optimal policy for a known RMAB is given by the Whittle index \citep{whittle1988restless}. The BCoR framework (Bayesian Learning for Contextual RMABs) incorporates covariates \(x_t\) \citep{liang2025context}: \(s_{t+1} = f(s_t, a_t, x_t) + w_t\) and \(r_t = g(s_t, a_t, x_t) + v_t\).

In restless bandits, the ellipsoidal decomposition framework can be applied to the joint posterior over states and parameters, providing practically exact samples for Thompson sampling. This is particularly valuable because particle filters in RMABs suffer from degeneracy over time, while MCMC for the joint posterior is prone to serious
convergence issues. Near-perfect iid samples would enable essentially exact decision-making at each time step with negligible approximation error.

\section{Sequential Bayesian Optimisation: The Gaussian Derivative Process Approach}

The principles of sequential Bayesian updating extend naturally to function optimisation, where the goal is to find the optimum of an unknown function by strategically selecting points to evaluate. This section examines the work of \citet{roy2021function} (see also the PhD thesis \cite{Soma_thesis}) on optimisation using posterior Gaussian derivative processes as a compelling example of the Bayesian reflex in action. Other approaches to Bayesian optimisation include expected improvement \citep{mockus1978application} and upper confidence bound \citep{srinivas2010gaussian}.

\subsection{Problem Formulation}

Consider an unknown function \(f: \mathbb{R}^d \mapsto \mathbb{R}\) for which we wish to find the global minimum (or maximum). We assume that the first and second partial derivatives of \(f\) exist and are continuous. The key insight of \citet{roy2021function} is to embed the function in a Gaussian process prior, which induces a derivative process that is also Gaussian. This derivative process becomes the foundation for sequential learning about the function's stationary points.

\subsection{Gaussian Process and Derivative Process}

Let \(g(\cdot)\) be a Gaussian process with mean function \(\mu(\cdot)\) and covariance function \(\sigma^2 c(\cdot,\cdot)\). Under appropriate regularity conditions, the derivative process \(g_i'(\mathbf{x}) = \partial g(\mathbf{x})/\partial x_i\) exists and is also a Gaussian process with mean \(\mu_i'(\cdot) = \partial \mu(\cdot)/\partial x_i\) and covariance function \(\sigma^2 \partial^2 c(\cdot,\cdot)/\partial x_i \partial y_i\).

Given initial data \(D_n = \{(\mathbf{x}_i, f(\mathbf{x}_i)): i=1,\ldots,n\}\), the posterior distribution of the derivative vector \(\mathbf{g}'(\mathbf{x}^*)\) at any point \(\mathbf{x}^*\) is a multivariate Student's t distribution:
\begin{equation*}
\pi(\mathbf{g}'(\mathbf{x}^*) | D_n) \propto \left[1 + \frac{\hat{\mu}'(\mathbf{x}^*)^T \hat{\Sigma}(\mathbf{x}^*)^{-1} \hat{\mu}'(\mathbf{x}^*)}{2b + (f_n - H\beta_0)^T (H\Sigma_0 H^T + \Sigma_{22})^{-1} (f_n - H\beta_0)}\right]^{-(a+d)},
\end{equation*}
where \(\hat{\mu}'(\mathbf{x}^*)\) and \(\hat{\Sigma}(\mathbf{x}^*)\) are posterior mean and scale derived from the Gaussian process.

\subsection{Sequential Optimisation Algorithm}

\citet{roy2021function} propose an iterative algorithm that embodies the Bayesian reflex through sequential data augmentation. In the initial stage (Stage 0), we simulate \(N\) realisations \(\{\mathbf{x}_1^*,\ldots,\mathbf{x}_N^*\}\) from \(\pi(\mathbf{x}^*|\mathbf{g}'(\mathbf{x}^*) = \mathbf{0}, D_n)\) using TMCMC, with prior constraints \(\|\mathbf{f}'(\mathbf{x}^*)\|_d < \epsilon\) and \(\mathbf{\Sigma}''(\mathbf{x}^*) > 0\) for minimisation problems. For subsequent stages \(k = 1,2,3,\ldots\), we compute importance weights 
\begin{align*}
w_k(\mathbf{x}_i^*) = \begin{cases} 1 & \text{for } k=1 \\ \frac{\pi(\mathbf{g}'(\mathbf{x}_i^*) = \mathbf{0}|D_{n+\sum_{j=0}^{k-1}n_j},\mathbf{x}_i^*)}{\pi(\mathbf{g}'(\mathbf{x}_i^*) = \mathbf{0}|D_{n+\sum_{j=0}^{k-2}n_j},\mathbf{x}_i^*)} & \text{for } k\geq 2 \end{cases}
\end{align*}
and resample \(M\) points with probabilities proportional to these weights. We then select points satisfying \(\|\mathbf{f}'(\mathbf{x}_{i_j}^*)\|_d < \eta_k\) with \(\eta_k \to 0\), evaluate \(f\) at these points, and augment the data: \(D_{n+\sum_{j=0}^k n_j} = D_{n+\sum_{j=0}^{k-1}n_j} \cup \{(\mathbf{x}_{i_j}^*, f(\mathbf{x}_{i_j}^*))\}\).

\subsection{Connections to Online Learning}

This algorithm exemplifies several key aspects of online Bayesian learning. The posterior distribution is updated each time new function evaluations are added to the dataset, embodying sequential updating. Points are selected based on the current posterior, balancing exploration of points with high uncertainty and exploitation of points near suspected optima, which is uncertainty-driven action. \citet{roy2021function} prove almost sure convergence of the algorithm to the true optima as the number of stages increases, under appropriate fixed-domain infill asymptotics. The posterior distribution provides natural credible regions for the optima, which shrink as more data is accumulated, demonstrating belief maintenance.

\subsection{Convergence Theory}

A key theoretical contribution is the proof of uniform convergence of the posterior Gaussian process and its derivatives to the true function and its derivatives. Under fixed-domain infill asymptotics where the set of input points becomes dense in the domain, the posterior Gaussian process \(g_n(\cdot)\) and its derivative process \(\mathbf{g}_n'(\cdot)\) satisfy \(\sup_{\mathbf{x} \in \mathcal{X}} |g_n(\mathbf{x}) - f(\mathbf{x})| \to 0\) almost surely and \(\sup_{\mathbf{x} \in \mathcal{X}} \|\mathbf{g}_n'(\mathbf{x}) - \mathbf{f}'(\mathbf{x})\|_d \to 0\) almost surely, with rates \(O(h^{3/2})\) and \(O(h^{1/2})\) respectively, where \(h\) is the maximum spacing between input points \citep{roy2021function}. This result ensures that as data accumulates sequentially, the posterior beliefs converge to the truth—a fundamental requirement for any online learning system.

\subsection{Bayesian Characterisation of the Number of Optima}

An elegant aspect of this approach is its ability to characterise the number of local optima through a Dirichlet process (DP) formulation. Let \(Y_j\) indicate whether a new point falls in the region satisfying optimality conditions. Then, with $P_k$ denoting the $k$-th stage (unknown) probability measure associated with $Y_k$, 
\(\pi(P_k | y_k) \sim DP\left(\sum_{j=1}^k \frac{1}{j^2} G + \sum_{j=1}^k \delta_{y_j}\right)\), where \(G\) is a base measure. The posterior mean of \(p_{mk}\) (probability of \(m\) optima) converges to 1 for the true number of optima as \(k \to \infty\).

\subsection{Relationship to Classical Bayesian Optimisation}

It is important to distinguish this approach from conventional Bayesian optimisation methods that treat the function as a black box and use acquisition functions to guide sampling \citep{snoek2012practical, shahriari2015taking}. The derivative-based approach of \citet{roy2021function} leverages gradient information to achieve significantly faster convergence, particularly in higher dimensions. As they note, function optimisation methods using traditional Gaussian process posteriors consider the objective function to be a ``black box'' and assume that the derivatives are unavailable. These methods would naturally be far less accurate compared to theirs when the derivatives are available. Moreover, convergence theory of such method does not seem to be available, and the final optimisation result is likely to depend upon the subjective choice of acquisition functions.

\subsection{Experimental Insights}

The algorithm has been tested on problems ranging from simple one-dimensional functions to challenging 50- and 100-dimensional optimisation problems. The posterior simulation approach consistently yields more accurate solutions than deterministic optimisation algorithms. Dimensionality remains a challenge, with gradient norms increasing with dimension, but the method still outperforms alternatives. Even when Markov chain mixing is less than ideal in high dimensions, the method can significantly outperform popular optimisation methods. Parallel implementation of importance weight computation makes the algorithm computationally feasible.

For simple functions like \(f(x) = 2x^3 - 3x^2 - 12x + 6\), the algorithm accurately identifies both maximum at \(x = -1\) and minimum at \(x = 2\) with estimates \(\hat{x}_{\max} = -0.999995\) and \(\hat{x}_{\min} = 2.000023\). For \(f(x) = \sin(x)\) on \([-10,10]\), the algorithm captures all three maxima and minima, with trace plots clearly exhibiting trimodality. For the two-dimensional function \(f(x_1, x_2) = x_1 x_2 (x_1 + x_2)(1 + x_2)\), the algorithm successfully identifies the maximum at \((0.376858, -0.752406)\), saddle points at \((-0.000309, -0.999580)\) and \((0.999433, -0.999915)\), and the inconclusive point at \((-0.000087, -0.000265)\) where the second derivative test fails. In tests on problems with dimensions up to 100, the algorithm demonstrates for \(d=50\) a gradient norm reduced to 73.05 and for \(d=100\) a gradient norm reduced to 330.45, with consistent outperformance of classical optimisation methods like Fisher scoring.

\subsection{Significance for the Bayesian Reflex}

This work demonstrates that the Bayesian reflex extends beyond traditional streaming data scenarios to settings where the system actively decides which observations to acquire next. The sequential updating of beliefs, uncertainty-driven exploration, and theoretical convergence guarantees all exemplify the core principles of the Bayesian reflex, making this a valuable addition to the online learning canon.

\section{Bayesian Inverse Regression for Climate Model Evaluation}

The principles of sequential Bayesian updating and look-up table methodology extend to an entirely different domain: evaluating complex physical models against historical data. The work of \citet{chatterjee2022ominous} on assessing global warming projections (see also the PhD thesis \cite{Debashis_thesis}, containing this work)  demonstrates how the Bayesian reflex operates in inverse regression settings, where the goal is to assess the plausibility of observed data given assumed future scenarios. Critically, this work directly builds upon the look-up table framework introduced by \citet{bhattacharya2007simulation}, adapting it to the challenging context of climate modelling.

\subsection{The Inverse Regression Paradigm}

Traditional forecasting uses past data to predict the future. However, \citet{chatterjee2022ominous} pose a fundamentally different question: How likely is the observed global warming pattern, assuming that the future projections of general circulation models (GCMs) are correct? This reframes the typical forecasting paradigm into an inverse regression setting:

\begin{equation*}
[x_1, \ldots, x_{T_0} \mid x_{T_0+1}, \ldots, x_T]
\end{equation*}

where \(x_1,\ldots,x_{T_0}\) represent current temperatures (1850-2016) and \(x_{T_0+1},\ldots,x_T\) represent future GCM projections (2017-2099). This inverse conditioning is a form of backward smoothing or retrospective inference, well-established in sequential Monte Carlo methods and Kalman filtering.

\subsection{Compositional Gaussian Process Emulator for Climate Dynamics}

To model the temporal dynamics of global temperature, \citet{chatterjee2022ominous} develop a nonparametric compositional Gaussian process emulator that treats the climate system as an unknown black-box process:

\begin{equation*}
x_t = f_t(x_{t-1}) + \epsilon_t, \quad \epsilon_t \sim N(0, \sigma_\epsilon^2)
\end{equation*}

where \(f_t\) is an unknown time-dependent function modelled as a Gaussian process. This state-space formulation captures complex nonlinear dependencies while maintaining principled uncertainty quantification.

The key innovation is the compositional structure: for any time \(t\), we have \(x_t = f(t, x_{t-1}) + \epsilon_t\), with \(f(\cdot,\cdot)\) modelled as a GP. This induces a non-Markovian dependence structure when marginalised, but computational tractability is achieved by introducing auxiliary variables \(\mathbf{D}_n^*\) that serve as a ``look-up table'' for the function \(f\).

\begin{figure}[htbp]
\centering
\begin{tikzpicture}[scale=0.9, transform shape]
    % Time line
    \draw[thick, ->] (0,0) -- (14,0);
    \node at (14,0.3) {Time};
    
    % Time periods
    \draw[thick] (1,0.3) -- (1,0.6);
    \draw[thick] (6,0.3) -- (6,0.6);
    \draw[thick] (1,0.45) -- (6,0.45);
    \node at (3.5,1.2) {Past $T_0$ (1850-2016)};
    
    \draw[thick] (7,0.3) -- (7,0.6);
    \draw[thick] (13,0.3) -- (13,0.6);
    \draw[thick] (7,0.45) -- (13,0.45);
    \node at (10,1.2) {Future (2017-2099)};
    
    % Observations
    \foreach \x in {2,3,4,5} {
        \filldraw[blue] (\x,0) circle (3pt);
    }
    \node at (3.5,-0.6) {HadCRUT4 $x_1,\ldots,x_{T_0}$};
    
    % GCM projections
    \foreach \x in {8,9,10,11,12} {
        \filldraw[red] (\x,0.3) circle (2pt);
        \filldraw[red] (\x,0) circle (2pt);
        \filldraw[red] (\x,-0.3) circle (2pt);
    }
    \node at (10,0.9) {GCM ensemble $x_{T_0+1},\ldots,x_T$};
    
    % Inverse conditioning
    \draw[->, thick, dashed] (8,1.5) to[out=90, in=90] node[midway, above] {Inverse conditioning} (3,1.5);
    
    % GP emulator
    \node (gp) at (7,-1.5) [rectangle, draw, fill=green!20, align=center, minimum width=5cm] {Compositional GP Emulator \\ $x_t = f_t(x_{t-1}) + \epsilon_t$};
    
    % Look-up table (inherited from Bhattacharya 2007)
    \node (lut) at (7,-3) [rectangle, draw, fill=purple!20, align=center, minimum width=5cm] {Look-up Table $\mathbf{D}_n^*$ \citep{bhattacharya2007simulation}};
    \draw[->] (gp) -- (lut);
    
    % Results
    \node at (7,-4.8) [text width=12cm, align=center, minimum height=1.5cm] {
        \textbf{Key Finding:} $[x_1,\ldots,x_{T_0} \mid x_{T_0+1},\ldots,x_T]$ largely excludes observed data \\
        $\Rightarrow$ GCM projections appear inconsistent with historical record under this model
    };
\end{tikzpicture}
\caption{Inverse regression framework for climate model evaluation, building on look-up table principle.}
\label{fig:climate}
\end{figure}
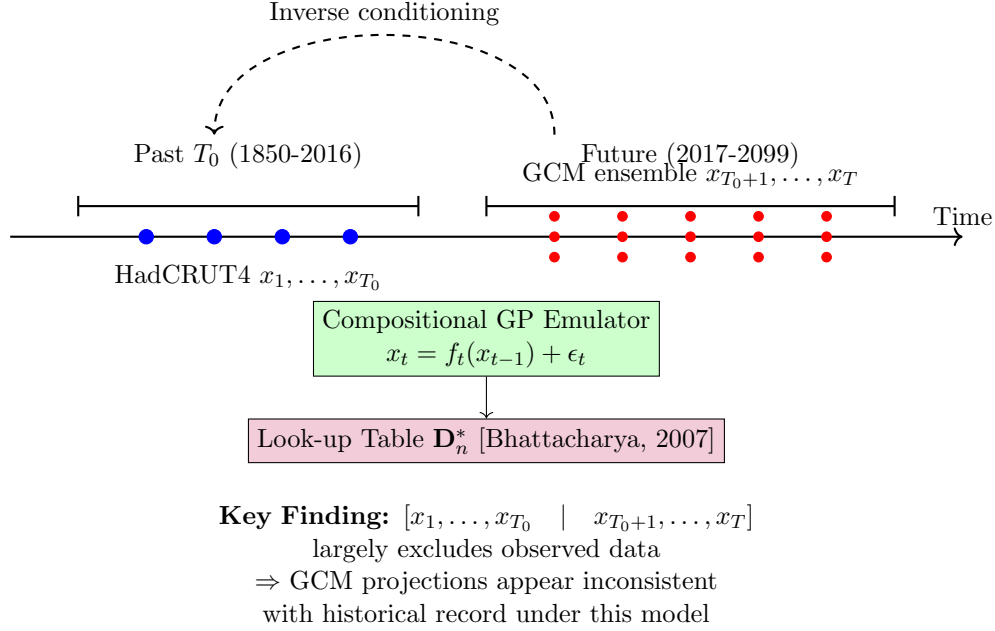

\subsection{Sequential Updating with Look-up Table}

The look-up table approach (Figure \ref{fig:climate}), directly inherited from \citet{bhattacharya2007simulation}, exemplifies the Bayesian reflex in action. Given initial data, we simulate the GP at a set of grid points \(\mathbf{G}_n\) to obtain \(\mathbf{D}_n^* = \{f(z_1^*), \ldots, f(z_n^*)\}\), where, for $i=1,\ldots,n$, 
$z^*_i\in [0,\infty)\times\mathbb (-\infty,\infty)$. Then, for any subsequent time \(t\), the conditional distribution of \(x_t\) given \(\mathbf{D}_n^*\) and past observations is Gaussian:

\begin{align*}
[x_t = f(x_{t,t-1}^*) + \epsilon_t \mid \mathbf{D}_n^*, x_{t-1}, \boldsymbol{\theta}_f, \sigma_\epsilon^2] \sim N(\mu_t, \sigma_t^2)
\end{align*}

with $x^*_{t,t-1}=(t,x_{t-1})$, and mean and variance given by GP regression formulas. This allows sequential simulation of the entire time series while maintaining computational stability.

The authors explicitly acknowledge this intellectual debt:

\begin{quote}
``The key idea has parallels with Bhattacharya (2007) and Ghosh et al. (2014). It is important to appreciate that although our time series model seems to be a Markovian model at the first glance, it is actually made up of compositions of GPs, and as we shall clarify, has highly structured non-Markovian dependence, with non-Gaussian, intractable distribution.''
\end{quote}

\subsection{Inverse Posterior Distribution}

The key quantity of interest is the posterior distribution of current temperatures given future GCM projections:

\begin{align*}
[x_1, \ldots, x_{T_0} \mid x_{T_0+1}, \ldots, x_T] = \int [x_1, \ldots, x_{T_0} \mid \mathbf{D}_n^*, \boldsymbol{\theta}_f, \sigma_\epsilon^2] \, d[\mathbf{D}_n^*, \boldsymbol{\theta}_f, \sigma_\epsilon^2 \mid x_{T_0+1}, \ldots, x_T]
\end{align*}

This is obtained by first sampling from the posterior of the look-up table and parameters given future data, then simulating the current time series conditional on these samples. The procedure directly mirrors the forward filtering-backward smoothing algorithms in state-space models and leverages the computational advantages of the look-up table representation.

\subsection{Bayesian Multiple Testing for Model Selection}

With multiple GCMs in each climate scenario (A1B, A2, B1, Commitment), \citet{chatterjee2022ominous} develop a Bayesian multiple testing framework to select the best model. For each GCM \(k\), they construct an averaged time series \(\bar{x}_t = K^{-1} \sum_{k=1}^K x_t^{(k)}\) and define discrepancy measures:

\begin{align*}
S_1^{(k)}(\mathbf{v}_{T_0}) &= \frac{1}{T_0} \sum_{t=1}^{T_0} \frac{|v_t - M(\bar{x}_t \mid \bar{x}_{T_0+1},\ldots,\bar{x}_T, \mathcal{M}_k)|}{\sqrt{\Var(\bar{x}_t \mid \bar{x}_{T_0+1},\ldots,\bar{x}_T, \mathcal{M}_k) + c}}, \\
S_2^{(k)}(\mathbf{v}_{T_0}) &= \frac{1}{T_0} \sum_{t=1}^{T_0} \frac{(v_t - M(\bar{x}_t \mid \bar{x}_{T_0+1},\ldots,\bar{x}_T, \mathcal{M}_k))^2}{\Var(\bar{x}_t \mid \bar{x}_{T_0+1},\ldots,\bar{x}_T, \mathcal{M}_k) + c}
\end{align*}
where $\mathcal M_k$ is the $k$-th GCM model and \(M(\cdot)\) denotes posterior mode. For appropriate lower and upper bounds $\bar\ell_k$ and $\bar u_k$, the hypotheses are formulated as:

\begin{align*}
H_{0k}: \zeta = k, S^{(k)}(\bar{\mathbf{x}}_{T_0}) - S^{(k)}(\mathbf{x}_{T_0}^{(0)}) \in [\bar{\ell}_k, \bar{u}_k] \\
H_{1k}: \{\zeta \neq k\} \cup \{\zeta = k, S^{(k)}(\bar{\mathbf{x}}_{T_0}) - S^{(k)}(\mathbf{x}_{T_0}^{(0)}) \in [\bar{\ell}_k, \bar{u}_k]^c\}
\end{align*}

This combines forward model probability \([\zeta = k \mid \bar{x}_{T_0+1},\ldots,\bar{x}_T]\) with inverse fit assessment through the discrepancy measure.

\subsection{Multivariate Extension for Ensemble Analysis}

Recognising that GCM simulations in each scenario form an ensemble, \citet{chatterjee2022ominous} extend their approach to multivariate Gaussian processes. For \(K\)-dimensional vectors \(\mathbf{x}_t = (x_t^{(1)},\ldots,x_t^{(K)})^\top\), the model becomes:

\begin{equation*}
\mathbf{x}_t = \mathbf{f}(\mathbf{x}_{t,t-1}^*) + \boldsymbol{\epsilon}_t, \quad \boldsymbol{\epsilon}_t \sim N_K(\mathbf{0}, \boldsymbol{\Sigma}_\epsilon)
\end{equation*}

where \(\mathbf{f}\) is a \(K\)-variate GP with covariance structure \(\operatorname{Cov}(\mathbf{f}(\mathbf{z}_1^*), \mathbf{f}(\mathbf{z}_2^*)) = c_f(\mathbf{z}_1^*, \mathbf{z}_2^*) \boldsymbol{\Sigma}_f\). This Kronecker product structure enables scalable inference while capturing correlations across ensemble members.

\subsection{Empirical Findings}

The empirical results are striking. Under the inverse Bayesian model testing framework, most GCMs—regardless of scenario—assign low posterior probability to the actual global warming pattern observed from 1850 to 2016. In other words, if their future projections were correct, the present as we know it would be highly unlikely under this model. Figures showing posterior distributions of current temperatures conditional on future GCM projections reveal that for scenarios A1B, A2, and B1, the observed HadCRUT4 data largely falls outside high-density regions. Only the Commitment scenario shows partial alignment.

The multivariate analysis reinforces these conclusions. For the averaged time series \([\bar{x}_1,\ldots,\bar{x}_{T_0} \mid \mathbf{x}_{T_0+1},\ldots,\mathbf{x}_T]\), the observed data lies far outside the posterior support. Even when considering the maximum across ensemble members, the fit remains poor, with posterior variances high and observed data in low-density regions.

\subsection{Forward Forecasting with Historical Data}

Using only historical HadCRUT4 data (1850-2016), \citet{chatterjee2022ominous} produce Bayesian forecasts for 2017-2099. Their predictions suggest more moderate warming trajectories than GCM projections, with posterior modes lying far below GCM forecasts. Only the Commitment scenario shows partial alignment with their predictive distributions.

Remarkably, these results align with the benchmark forecast of \citet{green2009validity}, who argued that future global temperatures would remain within \(0.5^\circ\)C of 2008 levels. The Bayesian GP forecasts place the future temperature within this range, casting serious doubt on the drastic warming predicted by most GCMs \textit{under the assumptions of this statistical model}. It is important to note that these findings are based on a univariate time series model and do not incorporate the full complexity of physical climate models; they should be interpreted as one statistical perspective among many \citep{ipcc2021}.

\subsection{Significance for the Bayesian Reflex}

This work demonstrates that the Bayesian reflex operates in multiple directions—not just forward in time, but also backward through inverse conditioning. The sequential updating of beliefs about the climate system, whether through forward prediction or inverse assessment, embodies the same core principles: belief maintenance via probabilistic representations, sequential updating through Bayes' theorem, and uncertainty-driven assessment of model adequacy.

The methodological contributions are substantial: compositional GPs for nonparametric time series, look-up table approximations for computational tractability (building directly on \citet{bhattacharya2007simulation}), Bayesian multiple testing for model selection in inverse settings, and multivariate GP extensions for ensemble analysis. These tools provide a rigorous statistical framework for evaluating complex physical models, with implications far beyond climate science.

The authors themselves recognise the broader significance of their approach:

\begin{quote}
``Our framework could be extended to incorporate spatio-temporal models, dynamic covariates, or hybrid approaches that integrate physical constraints into the statistical emulation process. We also anticipate that future work may apply this framework to other environmental processes, such as sea-level rise or precipitation extremes, where model uncertainty remains high and decision-making stakes are critical.''
\end{quote}

\section{Recursive Bayesian Assessment of Series Convergence}

Perhaps the most fundamental expression of the Bayesian reflex appears in the work of \citet{roy2020bayesian, roy2020bayesian2} on assessing the convergence of infinite series;
see also \cite{Soma_thesis}. This seemingly abstract mathematical problem reveals deep connections to online learning and has profound implications for understanding long-term behaviour in complex systems, including climate dynamics.

\subsection{The Recursive Bayesian Framework}

Consider an infinite series \(\sum_{i=1}^\infty X_i\), where the terms may be deterministic or random. \citet{roy2020bayesian2} develop a recursive Bayesian procedure that processes the series in stages. At stage \(j\), they consider a block of \(n_j\) terms and define the partial sum:

\begin{equation*}
S_{j,n_j} = \sum_{i=\sum_{k=0}^{j-1}n_k+1}^{\sum_{k=0}^j n_k} X_i
\end{equation*}

with \(n_0 = 0\). They then introduce a non-negative decreasing sequence \(\{c_j\}_{j=1}^\infty\) and define the binary indicator:

\begin{equation*}
Y_{j,n_j} = \mathbb{I}\{|S_{j,n_j}| \leq c_j\}
\end{equation*}

The probability associated with this indicator, \(p_{j,n_j} = P(Y_{j,n_j} = 1)\), can be interpreted as the probability that the series is convergent based on the data observed up to stage \(j\).

The recursive Bayesian update proceeds as follows. At stage 1, with prior \(p_{1,n_1} \sim \text{Beta}(\alpha_1, \beta_1)\) where \(\alpha_1 = \beta_1 = 1\), the posterior is:

\begin{equation*}
\pi(p_{1,n_1} | y_{1,n_1}) \sim \text{Beta}(\alpha_1 + y_{1,n_1}, \beta_1 + 1 - y_{1,n_1})
\end{equation*}

At stage 2, the prior is constructed from the stage 1 posterior associated with a \(\text{Beta}(\alpha_1 + \alpha_2, \beta_1 + \beta_2)\) prior, yielding:

\begin{equation*}
\pi(p_{2,n_2}) \sim \text{Beta}(\alpha_1 + \alpha_2 + y_{1,n_1}, \beta_1 + \beta_2 + 1 - y_{1,n_1})
\end{equation*}

Continuing this process, at stage \(k\) the posterior becomes:

\begin{equation}
\pi(p_{k,n_k} | y_{k,n_k}) \sim \text{Beta}\left(\sum_{j=1}^k \alpha_j + \sum_{j=1}^k y_{j,n_j}, \; k + \sum_{j=1}^k \beta_j - \sum_{j=1}^k y_{j,n_j}\right)
\label{eq:beta_posterior}
\end{equation}

With the choice \(\alpha_j = \beta_j = 1/j^2\), the posterior mean and variance simplify to:

\begin{align*}
E(p_{k,n_k} | y_{k,n_k}) &= \frac{\sum_{j=1}^k 1/j^2 + \sum_{j=1}^k y_{j,n_j}}{k + 2\sum_{j=1}^k 1/j^2} \\
\Var(p_{k,n_k} | y_{k,n_k}) &= \frac{(\sum_{j=1}^k 1/j^2 + \sum_{j=1}^k y_{j,n_j})(k + \sum_{j=1}^k 1/j^2 - \sum_{j=1}^k y_{j,n_j})}{(k + 2\sum_{j=1}^k 1/j^2)^2(1 + k + 2\sum_{j=1}^k 1/j^2)}
\end{align*}
The principle is illustrated schematically by Figure \ref{fig:series}.

\begin{figure}[htbp]
\centering
\begin{tikzpicture}[scale=0.9, transform shape]
    % Series representation
    \draw[thick] (0,0) -- (13,0);
    \foreach \x in {1,2,3,4,5,6,7,8,9,10,11,12} {
        \draw (\x,0.1) -- (\x,-0.1);
    }
    
    % Blocks – repositioned with wider gaps
    \draw[fill=blue!20, opacity=0.5] (1,0) rectangle (3,0.6);
    \draw[fill=blue!40, opacity=0.5] (4.5,0) rectangle (6.5,0.6);
    \draw[fill=blue!60, opacity=0.5] (8,0) rectangle (10,0.6);
    \draw[fill=blue!80, opacity=0.5] (11.5,0) rectangle (13.5,0.6);
    
    % Block labels
    \node at (2,1.2) {Block 1: $n_1$ terms};
    \node at (5.5,1.2) {Block 2: $n_2$ terms};
    \node at (9,1.2) {Block 3: $n_3$ terms};
    \node at (12.5,1.2) {Block $k$: $n_k$ terms};
    
    % Binary indicators – shifted horizontally
    \node (Y1) at (2,-0.8) {$Y_{1,n_1}$};
    \node (Y2) at (5.5,-0.8) {$Y_{2,n_2}$};
    \node (Y3) at (9,-0.8) {$Y_{3,n_3}$};
    \node (Yk) at (12.5,-0.8) {$Y_{k,n_k}$};
    
    \draw[->] (2,0.3) -- (Y1);
    \draw[->] (5.5,0.3) -- (Y2);
    \draw[->] (9,0.3) -- (Y3);
    \draw[->] (12.5,0.3) -- (Yk);
    
    % Beta updates – fixed width, small font, increased vertical spacing
    \node (beta1) at (2,-2.5) [text width=2.5cm, align=center, font=\footnotesize] {Beta($\alpha_1+y_1$, $\beta_1+1-y_1$)};
    \node (beta2) at (5.5,-2.5) [text width=2.5cm, align=center, font=\footnotesize] {Beta($\sum\alpha_j+\sum y_j$, $2+\sum\beta_j-\sum y_j$)};
    \node (beta3) at (9,-2.5) [text width=2.5cm, align=center, font=\footnotesize] {Beta($\sum\alpha_j+\sum y_j$, $3+\sum\beta_j-\sum y_j$)};
    \node (betak) at (12.5,-2.5) [text width=2.5cm, align=center, font=\footnotesize] {Beta($\sum_{j=1}^k \alpha_j + \sum_{j=1}^k y_j$, $k + \sum_{j=1}^k \beta_j - \sum_{j=1}^k y_j$)};
    
    \draw[->] (Y1) -- (beta1);
    \draw[->] (Y2) -- (beta2);
    \draw[->] (Y3) -- (beta3);
    \draw[->] (Yk) -- (betak);
    
    % Convergence theorem – moved further down
    \node at (7,-4.5) [text width=13cm, align=center, font=\small] {
        $\pi(\mathcal{N}_1 | y_{k,n_k}) \to 1$ iff series converges \\
        $\pi(\mathcal{N}_0 | y_{k,n_k}) \to 1$ iff series diverges
    };
\end{tikzpicture}
\caption{Recursive Bayesian framework for series convergence: Binary indicators from blocks of terms drive Beta-Binomial updates.}
\label{fig:series}
\end{figure}
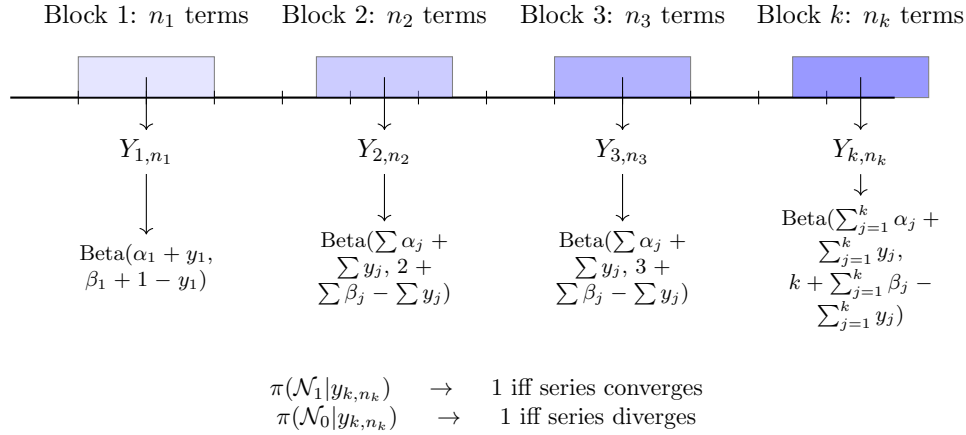

\subsection{Theoretical Characterisation}

\citet{roy2020bayesian2} prove two fundamental theorems that characterise convergence and divergence in terms of the limiting posterior behaviour:

\begin{theorem}[Convergence Characterisation]
For any \(\omega\) in the sample space (excluding a null set), the series \(S_{1,\infty}(\omega)\) converges if and only if there exists a non-negative monotonically decreasing sequence \(\{c_j(\omega)\}_{j=1}^\infty\) such that for any choice of \(\{n_j\}_{j=1}^\infty\), \(\pi(\mathcal{N}_1 | y_{k,n_k}(\omega)) \to 1\) as \(k \to \infty\), where \(\mathcal{N}_1\) is any neighbourhood of 1.
\end{theorem}

\begin{theorem}[Divergence Characterisation]
For any \(\omega\) in the sample space (excluding a null set), \(S_{1,\infty}(\omega)\) diverges if and only if there exists a sequence \(\{n_j(\omega)\}_{j=1}^\infty\) such that \(\pi(\mathcal{N}_0 | y_{k,n_k}(\omega)) \to 1\) as \(k \to \infty\), where \(\mathcal{N}_0\) is any neighbourhood of 0.
\end{theorem}

These theorems establish that the recursive Bayesian procedure provides a necessary and sufficient condition for convergence or divergence, making it a universal test for series convergence.

\subsection{Nonparametric Adaptive Bounds}

A key challenge in implementing this procedure is constructing appropriate bounds \(c_j\). \citet{roy2020bayesian2} propose an elegant nonparametric adaptive bound:

\begin{equation}
c_j = \hat{C}_j / \log(j+1), \quad \text{with} \quad \hat{C}_j = \hat{C}_{j-1} + 0.05 \cdot (2y_{j-1} - 1)
\label{eq:adaptive_bound}
\end{equation}

This simple rule adjusts the bound based on previous outcomes: if convergence was indicated at the previous stage (\(y_{j-1}=1\)), the bound increases slightly, favouring continued convergence; if divergence was indicated (\(y_{j-1}=0\)), the bound decreases, favouring divergence. The \(\log(j+1)\) denominator ensures that the bound decays to zero at a rate that is neither too fast nor too slow.

\subsection{Application to Climate Change}

\citet{roy2020bayesian2} apply this framework to assess long-term climate dynamics using two datasets: the HadCRUT4 global temperature record (1850-2016) and Holocene temperature reconstructions spanning 12,000 years. They transform the temperature data as:

\begin{equation*}
Y_{\theta_0,t} = \log(\log(X_t)) - \log(\log(\theta_0))
\end{equation*}

for \(\theta_0\) in a grid over \([11^\circ\mathrm{C}, 16^\circ\mathrm{C}]\). For each \(\theta_0\), they apply the recursive Bayesian procedure and find that \(\sum_{t=1}^\infty Y_{\theta_0,t} = \infty\) for all \(\theta_0\) in the grid.

This leads to a striking conclusion under their model: global climate dynamics are subject only to temporary variations, and neither prolonged global warming nor global cooling is likely in the foreseeable future. The current global warming phenomenon, they argue, is just an instance of such temporary variation. However, it must be emphasised that this is a statistical inference based on a specific model; it does not account for the full complexity of the climate system and should be interpreted alongside the vast body of climate science \citep{ipcc2021}.

\subsection{Application to the Riemann Hypothesis}

Perhaps most remarkably, \citet{roy2020bayesian} apply this framework to investigate the Riemann Hypothesis, one of the most famous unsolved problems in mathematics. The Riemann Hypothesis is equivalent to the convergence of the Dirichlet series for the Möbius function:

\begin{equation*}
M(a) = \sum_{n=1}^\infty \frac{\mu(n)}{n^a} = \frac{1}{\zeta(a)}
\end{equation*}

for \(\mathrm{Re}(a) > 1/2\), where \(\mu(n)\) is the Möbius function and \(\zeta(a)\) is the Riemann zeta function.

Using an efficient algorithm to compute the first \(10^9\) values of the Möbius function, \citet{roy2020bayesian} apply their recursive Bayesian procedure with \(n_j = 10^6\) terms per stage and \(K = 1000\) stages. The results show that \(M(a)\) diverges for \(a < 0.72\) and converges for \(a \geq 0.72\). Since the Riemann Hypothesis requires convergence for all \(a > 1/2\), this finding does not support the conjecture. The authors are careful to note that this is not a mathematical proof, but rather a Bayesian statistical analysis that casts doubt on the hypothesis.

Both the convergence assessment method and the multiple limit points extension yield remarkably consistent results, showing that the Riemann Hypothesis is not supported by this Bayesian analysis.

\subsection{Connections to Online Learning}

This work exemplifies the Bayesian reflex in several fundamental ways. The posterior is updated recursively as each new block of terms is processed, with the posterior at stage \(k\) essentially becoming the prior for stage \(k+1\). The posterior \(\pi(p_{k,n_k} | y_{k,n_k})\) represents the current belief about convergence probability, which evolves as more data arrives. The nonparametric bound \(\hat{C}_j\) adapts based on previous outcomes, analogous to adaptive learning rates in online optimisation. The method handles data in a streaming fashion, processing blocks of \(n_j\) terms sequentially. The almost sure convergence theorems provide consistency guarantees analogous to those sought in online learning theory. The recursive updates are computationally trivial (simple counting), enabling analysis of massive datasets such as \(10^9\) terms for the Riemann Hypothesis. The method works for arbitrary dependence structures among the terms, a significant advantage over classical approaches like Kolmogorov's three series theorem that require independence.

\subsection{Significance for the Bayesian Reflex}

This work demonstrates that the Bayesian reflex operates at the most fundamental level of mathematical analysis—determining whether an infinite series converges. The recursive updating of beliefs about convergence, the adaptive construction of bounds based on previous outcomes, and the theoretical convergence guarantees all embody the core principles of the Bayesian reflex. The applications to climate change and the Riemann Hypothesis show that these ideas have profound implications far beyond traditional online learning domains. In climate science, they provide a novel statistical perspective on long-term dynamics, suggesting that current warming may be temporary under this model. In mathematics, they offer new insights into one of the most famous unsolved problems, though not a definitive proof.

\section{Recursive Bayesian Analysis of Prime Numbers}

Perhaps the most astonishing application of recursive Bayesian methodology appears in the work of \citet{roy2025prime} on modelling the distribution of prime numbers. This work demonstrates that the Bayesian reflex extends to the purest realm of mathematics, leading to both theoretical insights about the Riemann Hypothesis and practical discoveries of new prime numbers.

\subsection{The Prime Number Theorem and Non-Homogeneous Poisson Processes}

The Prime Number Theorem (PNT) states that the number of primes less than or equal to \(x\), denoted \(\phi(x)\), satisfies \(\phi(x) \sim x/\log x\) as \(x \to \infty\). This asymptotic result suggests that primes become increasingly rare as numbers grow larger, and naturally leads to a probabilistic interpretation. \citet{roy2025prime} exploit this insight by modelling prime counts as a non-homogeneous Poisson process (NHPP) with intensity function:

\begin{equation}
\lambda_{\alpha,\beta}(t) = \alpha \, li(t) + \beta \, f(t)
\label{eq:intensity}
\end{equation}

where \(li(t) = 1/\log t\) is the intensity suggested by the PNT, and \(f(t)\) represents an error term related to the Riemann Hypothesis. For the error term associated with RH, they take \(f(x) = \frac{1}{\sqrt{x}}(\log\sqrt{x} + 1)\), which corresponds to the integrated error \(\int_2^x f(u)du = \sqrt{x}\log x\).

\subsection{Recursive Bayesian Updating with Prime Numbers}

Given primes \(t_1, t_2, \ldots, t_k\) observed sequentially, the recursive Bayesian updates take a remarkably elegant form. Let:

\begin{align*}
B_1^{(0)} &= \int_2^{t_1} li(u)du, \quad B_1^{(k)} = \int_{t_k}^{t_{k+1}} li(u)du \\
B_2^{(0)} &= \int_2^{t_1} f(u)du, \quad B_2^{(k)} = \int_{t_k}^{t_{k+1}} f(u)du
\end{align*}

With prior \(\pi(\alpha_1, \beta_1) \propto \exp(-a\alpha_1)\alpha_1^{\gamma-1} \times \exp(-b\beta_1)\beta_1^{\xi-1}\), the recursive posterior at stage \(k\) becomes a mixture of gamma distributions of the form:

\begin{align*}
\pi(\alpha_k, \beta_k | Z_k = t_k) =& \, R_k C_1^{(k-1)} \times \frac{\Gamma(\gamma+k)}{(a+\sum_{i=0}^{k-1}B_1^{(i)})^{\gamma+k}} \times \frac{\Gamma(\xi+k-1)}{(b+\sum_{i=0}^{k-1}B_2^{(i)})^{\xi+k-1}} \\
& \times g(\alpha_k; a+\sum_{i=0}^{k-1}B_1^{(i)}, \gamma+k) \times g(\beta_k; b+\sum_{i=0}^{k-1}B_2^{(i)}, \xi+k-1) \\
& + R_k C_2^{(k-1)} \times \frac{\Gamma(\gamma+k-1)}{(a+\sum_{i=0}^{k-1}B_1^{(i)})^{\gamma+k-1}} \times \frac{\Gamma(\xi+k)}{(b+\sum_{i=0}^{k-1}B_2^{(i)})^{\xi+k}} \\
& \times g(\alpha_k; a+\sum_{i=0}^{k-1}B_1^{(i)}, \gamma+k-1) \times g(\beta_k; b+\sum_{i=0}^{k-1}B_2^{(i)}, \xi+k)
\end{align*}

This recursive structure means that the posterior at stage \(k\) becomes the prior for stage \(k+1\), exactly mirroring the Bayesian reflex principle. Crucially, this recursive formulation contains only two terms, whereas a non-recursive (batch) approach would require \(2^k\) terms—making analysis of millions of primes computationally infeasible without recursion.

\subsection{Asymptotic Convergence and Implications for the Riemann Hypothesis}

The recursive posteriors yield asymptotic results about the parameters \(\alpha_k\) and \(\beta_k\). For \(\alpha_k\), which corresponds to the PNT intensity, they prove:

\begin{theorem}
Under the marginal posterior probability measures \(\pi(\alpha_k|Z_k = t_k); k \ge 1\),
\begin{equation}
\alpha_k \xrightarrow{P} 1, \quad \text{as } k \to \infty
\label{eq:alpha_convergence}
\end{equation}
\end{theorem}

This provides a Bayesian validation of the Prime Number Theorem. For \(\beta_k\), which corresponds to the error term associated with the Riemann Hypothesis, they prove:

\begin{theorem}
Under the marginal posterior probability measures \(\pi(\beta_k|Z_k = t_k); k \ge 1\),
\begin{equation}
\beta_k - E(\beta_k|Z_k = t_k) \xrightarrow{P} 0, \quad \text{as } k \to \infty
\label{eq:beta_convergence}
\end{equation}
where \(E(\beta_k|Z_k = t_k) \sim \frac{\sqrt{k}}{(\log k)^{3/2}} \to \infty\).
\end{theorem}

Since \(\beta_k \to \infty\) in probability, there cannot exist any finite constant bounding \(\beta_k\). This implies that, under the assumed NHPP model, the error term associated with RH is too small—the true error must be larger than \(O(\sqrt{x}\log x)\). Therefore, this Bayesian analysis suggests that the Riemann Hypothesis may be false, providing a second, independent statistical argument using completely different methodology from the series convergence approach. The authors stress that this is not a mathematical proof, but rather strong Bayesian evidence.

\subsection{Asymptotic Equivalence of Recursive and Non-Recursive Inference}

A remarkable theoretical result is the asymptotic equivalence between the computationally efficient recursive Bayesian approach and the theoretically elegant but practically infeasible non-recursive approach. \citet{roy2025prime} prove that the recursive posteriors yield the same asymptotic inference as their non-recursive counterparts, while requiring vastly less computation:

\begin{theorem}
The recursive posterior distributions of the parameters and the recursive posterior predictive distributions asymptotically yield the same Bayesian inference as the corresponding non-recursive Bayesian theory.
\end{theorem}

This result is crucial because it justifies using the recursive formulation for practical prime hunting while maintaining theoretical guarantees. The non-recursive posteriors would contain \(2^k\) terms—with \(k\) potentially in the millions—making them impossible to compute, while the recursive versions contain only 2 terms for the parameter posteriors and 4 terms for the predictive distributions.

\subsection{Posterior Predictive Distributions for Prime Discovery}

For discovering new primes, the recursive posterior predictive distribution for the next prime \(Z_{k+1}\) given the first \(k\) primes is a mixture of four terms. For large \(k\), this simplifies to a form amenable to MCMC sampling:

\begin{equation}
h(t) \propto (li(t) + f(t)) \left(\frac{t F(t)}{\log t}\right)^{-k}
\label{eq:posterior_predictive}
\end{equation}

where \(F(t) = \int_2^t f(u)du\). For the sharp error bound from \cite{mossinghoff2015nonnegative}, they use:

\begin{equation}
F(x) = \frac{x}{(\log x)^{3/4}} \exp\left(-\sqrt{\frac{\log x}{6.315}}\right)
\label{eq:error_bound}
\end{equation}

\subsection{Discovering Mersenne Prime Candidatee via Change of Variable}

A particularly elegant contribution is the transformation to sample candidate Mersenne prime exponents directly. Setting \(s = 2^t - 1\) and working on the log scale with \(t \mapsto \exp(z)\), they derive:

\begin{align*}
\log[h_1^*(z)] =& \, \log[li(\exp(z) + p_0)] \\
& - k[\log(\exp(z) + p_0) + \log\{F(\exp(z) + p_0)\} - \log\log(\exp(z) + p_0)] \\
& - \exp(z)\log 2 + z
\end{align*}

This formulation completely avoids evaluating astronomical numbers like \(2^t - 1\), making the search for Mersenne primes computationally feasible on ordinary hardware.

\begin{figure}[htbp]
\centering
\begin{tikzpicture}[scale=0.9, transform shape]
    % Prime number sequence
    \draw[thick] (0,0) -- (12,0);
    
    % Known primes
    \foreach \x/\p in {1/$p_1$, 2/$p_2$, 3/$p_3$, 4/$p_4$, 5/$p_5$, 6/$p_6$} {
        \filldraw[blue] (\x,0) circle (3pt) node[above] {\p};
    }
    
    % Future primes
    \foreach \x in {8,9,10,11,12} {
        \filldraw[red, opacity=0.5] (\x,0) circle (3pt);
    }
    \node at (10,0.6) {Unknown primes $p_{k+1},\ldots$};
    
    % NHPP intensity
    \draw[thick, domain=0:12, samples=100, smooth, variable=\x, red] 
        plot ({\x}, {2/(\x+1)});
    \node at (6,1.5) {Intensity $\lambda_{\alpha,\beta}(t) = \alpha\,li(t) + \beta\,f(t)$};
    
    % Recursive updates
    \node (recursive) at (6,-1.5) [rectangle, draw, fill=green!20, text width=9cm, minimum height=1.2cm] {
        \textbf{Recursive Posterior:} Mixture of Gamma distributions \\
        Batch: $2^k$ terms (infeasible) $\rightarrow$ Recursive: 2 terms (tractable)
    };
    
    % Predictive distribution for new primes
    \node (predictive) at (6,-3) [rectangle, draw, fill=yellow!20, text width=9cm, minimum height=1.2cm] {
        $\pi(Z_{k+1} = t_{k+1} | Z_k = t_k) \propto (li(t) + f(t)) \left(\frac{t F(t)}{\log t}\right)^{-k}$
    };
    
    % Mersenne prime transformation
    \node (mersenne) at (6,-4.5) [rectangle, draw, fill=orange!20, text width=9cm, minimum height=1.2cm] {
        \textbf{Mersenne Discovery:} $s = 2^t - 1$ \\
        Log transformation avoids evaluating astronomical numbers
    };
    
    % Results
    \node at (6,-6.2) [rectangle, draw, fill=red!20, text width=9cm, minimum height=1.2cm] {
        \textbf{Results:} 259 new primes $>$ 140 million \\
        184 strong Mersenne prime candidates (up to 242 million digits)
    };
\end{tikzpicture}
\caption{Recursive Bayesian framework for prime number analysis and discovery.}
\label{fig:primes}
\end{figure}
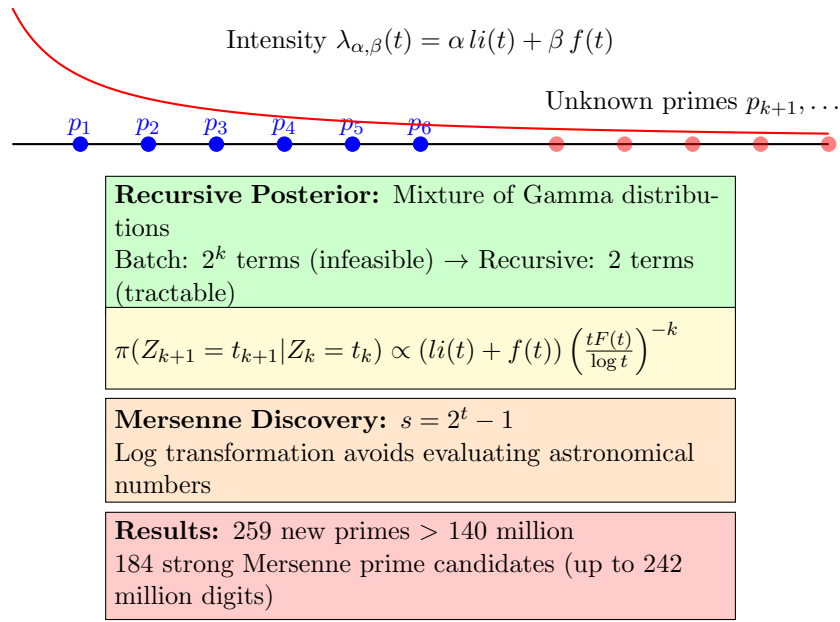

\subsection{Empirical Results}

Using TMCMC with a mixture of additive and multiplicative transformations, \citet{roy2025prime} generated \(10^7\) realisations in seconds on a standard laptop with 8GB RAM (Figure \ref{fig:primes}). Repeating the procedure with updated starting points, they discovered 259 new primes exceeding 140 million. Among these, 184 were identified as strong Mersenne prime exponent candidates, corresponding to potential Mersenne primes with digit lengths ranging from 42 million to 242 million. These candidates await verification by the mathematical community.

\subsection{Significance for the Bayesian Reflex}

This work demonstrates that the Bayesian reflex operates at the deepest level of mathematical inquiry. The recursive updating of beliefs about the parameters governing prime distribution, the sequential processing of primes as they are discovered, the use of posterior predictive distributions to guide exploration for new primes, and the theoretical convergence guarantees all embody the core principles of the Bayesian reflex. The practical outcome—discovery of 184 new potential Mersenne primes—shows that these principles have real-world utility beyond theoretical elegance.

\section{The Recursive Bayesian Framework for Stationarity Detection and Beyond}

The recursive Bayesian methodology introduced for series convergence has proven to be remarkably versatile, extending to a wide range of statistical problems involving the characterisation of stochastic processes. The work of \citet{roy2021stationarity} (see also \cite{Soma_thesis}) demonstrates how the same core principles—partitioning the index set, constructing binary indicators based on empirical distribution functions, and applying recursive Beta-Binomial updates—can be applied to detect stationarity in time series, assess MCMC convergence, test for spatial and spatio-temporal stationarity, characterise point processes, and determine frequencies in oscillatory signals. Classical approaches to these problems include the Dickey-Fuller test \citep{dickey1979distribution} for unit roots, the KPSS test \citep{kwiatkowski1992testing} for stationarity, and Ripley's K-function \citep{ripley1976second} for point processes.

\subsection{Theoretical Foundation: Local Stationarity}

A key theoretical insight underlying these extensions is that under mild regularity conditions, most stochastic processes are approximately locally stationary. As proved in \citet{roy2021stationarity}, for any stochastic process \(\mathbf{X} = \{X_s: s\in\mathcal{S}\}\) with appropriately differentiable finite-dimensional distributions and satisfying \(X_{s+h} = X_s + O_P(h)\) as \(h\to 0\), the process can be partitioned into disjoint regions \(\mathcal{N}_i\) where \(\{X_s: s\in\mathcal{N}_i\}\) is approximately stationary. This local stationarity property provides the foundation for constructing empirical distribution functions \(\hat{P}_i\) for each region.

Let \(\tilde{P}_K\) denote the pooled empirical distribution function based on all observations in \(\cup_{i=1}^K \mathcal{N}_i\). The key quantities for stationarity detection are the sup-norm differences:
\begin{equation*}
\sup_C |\hat{P}_i(C) - \tilde{P}_K(C)|
\end{equation*}
which converge to zero almost surely as the region sizes increase if and only if the process is stationary \citep{roy2021stationarity}. For nonstationary processes, these differences remain positive for infinitely many regions.

\subsection{Recursive Bayesian Procedure for Stationarity}

Following the same recursive Bayesian approach as in series convergence, we define binary indicators:
\begin{equation*}
Y_{j,n_j} = \mathbb{I}\left\{\sup_C |\hat{P}_j(C) - \tilde{P}_K(C)| \leq c_j\right\}
\end{equation*}
where \(c_j\) is a non-negative decreasing sequence. With prior \(p_{j,n_j} \sim \text{Beta}(\alpha_j, \beta_j)\) where \(\alpha_j = \beta_j = 1/j^2\), the recursive posterior at stage \(k\) becomes:
\begin{equation}
\pi(p_{k,n_k} | y_{k,n_k}) \sim \text{Beta}\left(\sum_{j=1}^k 1/j^2 + \sum_{j=1}^k y_{j,n_j}, \; k + \sum_{j=1}^k 1/j^2 - \sum_{j=1}^k y_{j,n_j}\right)
\label{eq:stationarity_posterior}
\end{equation}

Theorems 6 and 7 of \citet{roy2021stationarity} establish that the posterior converges to 1 in any neighbourhood of 1 if and only if the process is stationary, and converges to 0 in any neighbourhood of 0 if and only if the process is nonstationary. This provides a universal test for stationarity that requires no parametric assumptions about the underlying process.

\subsection{Extension to Covariance Stationarity}

The framework extends naturally to testing covariance stationarity by considering the empirical covariances:
\begin{equation*}
\widehat{Cov}_{ih} = \frac{\sum_{(s_1,s_2)\in\mathcal{N}_{ih}} X_{s_1}X_{s_2}}{2n_{ih}} - \left(\frac{\sum_{s_1\in\mathcal{N}_{ih}} X_{s_1}}{n_{ih}}\right)\left(\frac{\sum_{s_2\in\mathcal{N}_{ih}} X_{s_2}}{n_{ih}}\right)
\end{equation*}
where \(\mathcal{N}_{ih} = \{(s_1,s_2)\in\mathcal{N}_i: \|s_1-s_2\| = h\}\). Defining binary indicators based on differences between local and pooled covariances, the same recursive Beta-Binomial updates yield posterior convergence to 1 for covariance-stationary processes and to 0 for covariance-nonstationary processes \citep{roy2021stationarity}.

\begin{figure}[htbp]
\centering
\begin{tikzpicture}[scale=0.9, transform shape]
    % Process representation
    \draw[thick, domain=0:12, samples=100, smooth, variable=\x, blue] 
        plot ({\x}, {2*sin(\x r) + 0.5*rand});
    
    % Partition into regions
    \foreach \x in {2,4,6,8,10} {
        \draw[dashed] (\x,-1) -- (\x,3);
        \node at (\x+1,2.5) {$\mathcal{N}_i$};
    }
    
    % Empirical distributions
    \node (P1) at (1,-2) {$\hat{P}_1$};
    \node (P2) at (3,-2) {$\hat{P}_2$};
    \node (P3) at (5,-2) {$\hat{P}_3$};
    \node (P4) at (7,-2) {$\hat{P}_4$};
    \node (Pk) at (9,-2) {$\hat{P}_k$};
    \node (Ppool) at (11,-2) {$\tilde{P}_K$};
    
    % Binary indicators
    \node at (2,-3) {$Y_{1,n_1}$};
    \node at (4,-3) {$Y_{2,n_2}$};
    \node at (6,-3) {$Y_{3,n_3}$};
    \node at (8,-3) {$Y_{4,n_4}$};
    \node at (10,-3) {$Y_{k,n_k}$};
    
    % Recursive Beta updates
    \node (beta) at (6,-5) [rectangle, draw, fill=green!20, text width=10cm, minimum height=1.2cm] {
        $\pi(p_{k,n_k} | y_{k,n_k}) \sim \text{Beta}\left(\sum_{j=1}^k 1/j^2 + \sum_{j=1}^k y_j, \; k + \sum_{j=1}^k 1/j^2 - \sum_{j=1}^k y_j\right)$
    };
    
    % Applications
    \node at (0,-7) [text width=4cm, text centered] {ARMA Models};
    \node at (4,-7) [text width=4cm, text centered] {MCMC Diagnostics};
    \node at (8,-7) [text width=4cm, text centered] {Spatial Data};
    \node at (12,-7) [text width=4cm, text centered] {Point Processes};
    
    \draw[->] (beta) -- ++(-2,-1) -- (0,-6.5);
    \draw[->] (beta) -- ++(0,-1) -- (4,-6.5);
    \draw[->] (beta) -- ++(2,-1) -- (8,-6.5);
    \draw[->] (beta) -- ++(4,-1) -- (12,-6.5);
\end{tikzpicture}
\caption{Recursive Bayesian framework for stationarity detection across diverse stochastic processes.}
\label{fig:stationarity}
\end{figure}
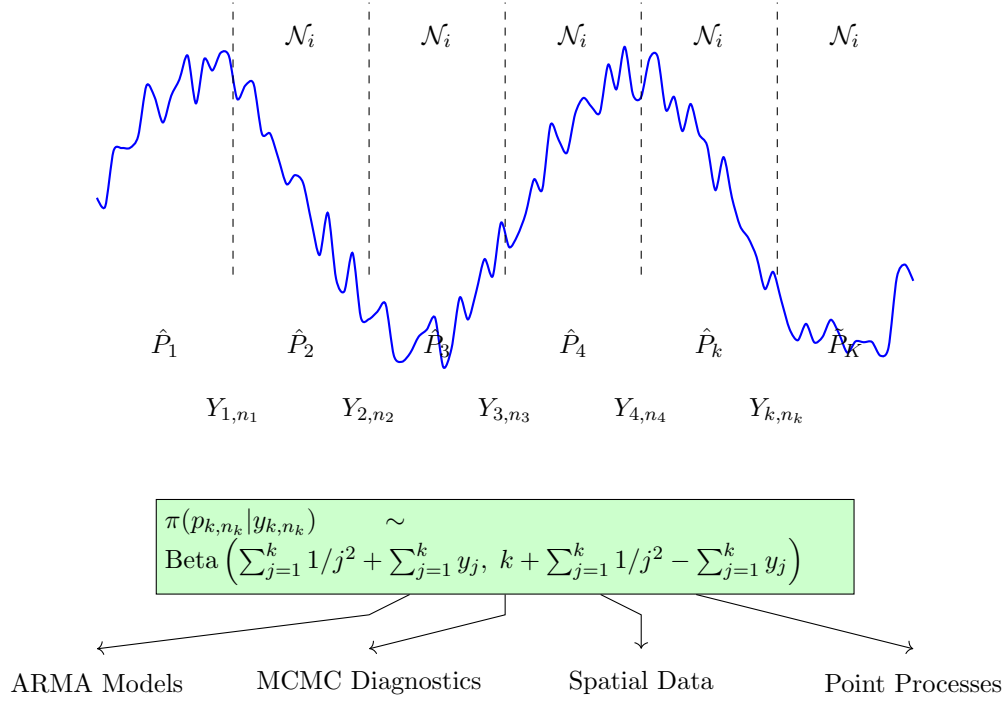

\subsection{Applications to Time Series Models}

The methodology has been extensively validated on various time series models (Figure \ref{fig:stationarity}). For AR(1) models \(X_t = \rho X_{t-1} + \epsilon_t\), the recursive Bayesian procedure correctly detects stationarity for \(|\rho| < 1\) and nonstationarity for \(|\rho| \geq 1\), even in subtle cases where \(\rho\) is extremely close to 1 (e.g., \(\rho = 0.99999\)). For AR(2) models, the method correctly identifies stationarity when the parameters satisfy \(\alpha+\beta < 1\), \(\beta-\alpha < 1\), and \(\beta > -1\), even for slowly diverging processes where traditional tests may fail \citep{roy2021stationarity}.

For ARCH(1) and GARCH(1,1) models, the method correctly detects stationarity when the parameters satisfy the usual stationarity conditions (\(\alpha < 1\) for ARCH(1); \(\alpha+\beta < 1\) for GARCH(1,1)). Interestingly, for ARCH(1) models with \(\alpha \geq 1\), the method sometimes indicates stationarity due to the uncorrelated nature of the realisations, highlighting the importance of understanding the properties being tested \citep{roy2021stationarity}.

\subsection{MCMC Convergence Diagnostics}

A particularly valuable application is to MCMC convergence diagnostics. The TMCMC framework of \citet{dutta2014markov} provides an ideal testbed. For a target distribution that is a product of 100 standard normal densities, TMCMC chains with different scaling parameters \(\ell\) exhibit different convergence behaviours. The recursive Bayesian procedure correctly identifies chains as stationary when the scaling parameter is near the optimal value \(\ell = 2.4\), and as nonstationary when \(\ell\) is too small (e.g., \(\ell = 0.001, 0.01\)) or too large (e.g., \(\ell = 100, 1000\)), corresponding to chains that converge too slowly due to tiny steps or high rejection rates \citep{roy2021stationarity}.

For mixture normal targets, the method correctly identifies stationarity when the chain mixes well between modes (e.g., mixture of \(N(0,1)\) and \(N(10,1)\)), and identifies nonstationarity when the chain exhibits distinct regimes of local stationarity (e.g., mixture of \(N(0,1)\) and \(N(15,1)\) with insufficient mixing). This demonstrates the method's ability to detect lack of convergence even in complex multimodal settings.

\subsection{Spatial and Spatio-Temporal Applications}

The framework extends naturally to spatial data. For zero-mean Gaussian processes with stationary covariance \(Cov(X_{s_1}, X_{s_2}) = \exp(-5\|s_1-s_2\|^2)\) and nonstationary covariance \(Cov(X_{s_1}, X_{s_2}) = \exp(-5\|\sqrt{s_1}-\sqrt{s_2}\|^2)\), the method correctly distinguishes between strict stationarity and nonstationarity. For mixtures of stationary and nonstationary components with weights as extreme as \(p = 0.99999\) (i.e., 99.999\% stationary, 0.001\% nonstationary), the method still correctly detects nonstationarity, demonstrating remarkable sensitivity \citep{roy2021stationarity}.

For spatio-temporal processes with separable covariance structures combining spatial and AR(1) temporal components, the method correctly identifies stationarity when \(|\rho| < 1\) and the spatial component is stationary, and nonstationarity otherwise. Comparisons with existing tests \citep{bandopadhyay2017nonparametric} show that the Bayesian method often outperforms classical approaches, particularly for non-Gaussian processes and small sample sizes \citep{roy2021stationarity}.

\subsection{Point Process Characterisation}

Perhaps the most extensive application is to point process analysis. The recursive Bayesian framework provides a unified approach to testing complete spatial randomness (CSR) via comparison of empirical \(G\)-functions, Poisson process assumptions via characterisation of mutual independence among disjoint regions, and stationarity of point processes via the same empirical distribution approach used for general processes. The method has been validated on an extensive range of point process models, including homogeneous and inhomogeneous Poisson processes (HPP, IHPP), log-Gaussian Cox processes (LGCP) with various covariance structures, Matérn cluster processes with homogeneous and inhomogeneous parameters, Thomas processes (modified) with various inhomogeneity patterns, Neyman-Scott processes (both homogeneous and inhomogeneous), and Strauss processes (pairwise interaction point processes). In each case, the recursive Bayesian procedure correctly identifies CSR status, Poisson vs. non-Poisson nature, and stationarity vs. nonstationarity, often outperforming classical methods based on \(G\)-function envelopes \citep{roy2021stationarity}. For example, in inhomogeneous LGCP where the classical method fails to reject CSR, the Bayesian method correctly identifies non-CSR.

\subsection{Mutual Independence Characterisation}

A novel contribution is the characterisation of mutual independence among random variables using Dirichlet process priors. For a sequence of random variables \(X_1, X_2, \ldots, X_K\), define the differences:
\begin{equation*}
\zeta_i = \sup_{t_1,\ldots,t_i} |P(X_i \leq t_i | X_1 \leq t_1, \ldots, X_{i-1} \leq t_{i-1}) - P(X_i \leq t_i)|
\end{equation*}

Using a Dirichlet process prior for the joint distribution avoids numerical issues with empirical conditional distributions. The recursive Bayesian procedure based on these differences converges to 1 if and only if the variables are mutually independent \citep{roy2021stationarity}. This characterisation is essential for testing Poisson process assumptions, where mutual independence of disjoint regions is required.

\subsection{Frequency Determination in Oscillatory Processes}

Another noteworthy extension is to determining frequencies in oscillatory stochastic processes. For a process \(X_t\) with unknown frequencies, transform to \(Z_t = \exp(X_t)/(1+\exp(X_t)) \in [0,1]\). For a suitable power \(r > 0\), the oscillations in \(Z_t^r\) become explicit. Partitioning \([0,1]\) into intervals \((\tilde{p}_{m-1}, \tilde{p}_m]\) and defining:
\begin{equation*}
Y_j = m \text{ if } \tilde{p}_{m-1} < Z_j^r \leq \tilde{p}_m
\end{equation*}
the proportions \(p_{m,j}\) can be interpreted as frequencies. Using a Dirichlet process prior with base measure \(G(Y_j = m) = 1/M\) (for finite \(M\)) or \(G(Y_j = m) = 2^{-m}\) (for countably infinite \(M\)), the recursive posterior converges to the true frequencies \citep{roy2021stationarity}.

The method has been validated on simulated data with single frequency (\( \omega = 0.02\)) and multiple frequencies (0.4, 0.1, 0.06), correctly recovering the frequencies when the power parameter \(r\) and partition size \(M\) are appropriately chosen. For real data, the method successfully identifies the characteristic frequencies in the Southern Oscillation Index (SOI) and fish recruitment time series, with results consistent with periodogram analysis \citep{shumway2006time}.

\subsection{Computational Implementation}

All implementations are parallelised using MPI (Message Passing Interface) in C, with some R code for specific tasks. For point process applications, the `spatstat' package is used for data generation. The parallel implementation enables analysis of large datasets—for example, processing \(2\times 10^8\) AR(1) samples in minutes.

\subsection{Significance for the Bayesian Reflex}

This body of work demonstrates the remarkable versatility of the recursive Bayesian framework. The same core machinery—partitioning into locally stationary regions, constructing binary indicators based on empirical distributions, applying recursive Beta-Binomial updates—succeeds across a wide range of statistical problems: time series analysis (AR, ARCH, GARCH models), MCMC convergence diagnostics, spatial statistics (strict and covariance stationarity), spatio-temporal statistics, point process analysis (CSR, Poisson, stationarity for 17 different process types), and frequency analysis of oscillatory signals. This unification of seemingly disparate problems under a common methodological framework is a powerful testament to the unifying power of the Bayesian reflex. The recursive updating of beliefs about the underlying process, the adaptive construction of bounds based on previous outcomes, and the theoretical convergence guarantees all embody the core principles introduced in the context of series convergence and extended across the statistical sciences.

\section{Applications: The Bayesian Reflex in Action}

\subsection{Kilkari Program: Collaborative Bandit Formulation}

The Kilkari program optimisation problem can be formalised as follows \citep{dasgupta2025bayesian}. Let \(u \in \{1,\ldots,U\}\) index beneficiaries and \(t \in \{1,\ldots,T\}\) index time slots. Let \(Y_{u,t} \in \{0,1\}\) indicate successful call delivery. The reward matrix is modelled as \(Y_{u,t} \sim \text{Bernoulli}(\sigma(U_u^\top V_t))\), where \(\sigma(x) = 1/(1+e^{-x})\) is the logistic function, \(U_u \in \mathbb{R}^d\) is the user embedding, and \(V_t \in \mathbb{R}^d\) is the time slot embedding. Priors are placed on the embeddings: \(U_u \sim \mathcal{N}(0, \sigma_U^2 I_d)\) and \(V_t \sim \mathcal{N}(0, \sigma_V^2 I_d)\).

Posterior inference is performed using stochastic gradient Langevin dynamics \citep{welling2011bayesian}. Thompson sampling selects time slots with probability proportional to the posterior probability of success \citep{dasgupta2025bayesian}. The ellipsoidal decomposition framework could enhance this application by providing practically exact iid samples from the posterior over \((U_u, V_t)\), eliminating approximation error and potentially improving the convergence rate of the bandit algorithm.

\subsection{Sepsis Prediction: Online Conformal Prediction}

The Sepsyn-OLCP framework combines online learning with conformal prediction \citep{zhou2025sepsyn, vovk2005algorithmic}. Let \(f_t(x)\) be the prediction at time \(t\) for patient with features \(x\). Conformal prediction produces prediction sets with coverage guarantees \citep{vovk2005algorithmic}: \(C_t(x) = \{y: s(x,y) \leq q_t\}\), where \(s\) is a nonconformity score and \(q_t\) is the \((1-\alpha)\)-quantile of scores on a calibration set.

In the online setting, the calibration set grows over time. The quantile is updated as \citep{zhou2025sepsyn}: \(q_{t+1} = q_t + \gamma_t (\mathbf{1}(s_{t+1} \leq q_t) - (1-\alpha))\), where \(\gamma_t\) is a learning rate satisfying \(\sum \gamma_t = \infty\) and \(\sum \gamma_t^2 < \infty\).

\subsection{ICEE: In-Context Exploration-Exploitation}

The ICEE framework uses a transformer architecture for meta-RL \citep{dai2024context, garg2022what}. The model processes sequences of the form \(\tau = (s_1, a_1, r_1, s_2, a_2, r_2, \ldots, s_t)\) and outputs action probabilities \(\pi(a \mid \tau)\). The transformer architecture uses multi-head self-attention \citep{vaswani2017attention}. The training objective maximises the expected return across tasks \citep{dai2024context}. The transformer implements an implicit inference algorithm, and theoretical analysis shows that with sufficient capacity, the transformer can approximate the Bayes-optimal policy \citep{dai2024context, xie2021explanation}.

\subsection{Sequential Bayesian Optimisation}

The Gaussian derivative process approach of \citet{roy2021function} has been successfully applied to a wide range of optimisation problems, demonstrating the versatility of the Bayesian reflex. For one-dimensional examples like \(f(x) = 2x^3 - 3x^2 - 12x + 6\), the algorithm accurately identifies both maximum at \(x = -1\) and minimum at \(x = 2\) with estimates \(\hat{x}_{\max} = -0.999995\) and \(\hat{x}_{\min} = 2.000023\). For multimodal functions like \(f(x) = \sin(x)\) on \([-10,10]\), the algorithm captures all three maxima and minima. For high-dimensional optimisation with dimensions up to 100, the algorithm demonstrates consistent outperformance of classical optimisation methods like Fisher scoring.

\subsection{Climate Model Evaluation}

The inverse regression framework of \citet{chatterjee2022ominous} has been applied to evaluate GCM projections across four climate scenarios. The empirical results suggest that most GCMs assign low posterior probability to the actual global warming pattern observed from 1850 to 2016. These results cast doubt on the fidelity of GCM projections under the assumptions of this statistical model, but they must be interpreted alongside the broader climate science literature \citep{ipcc2021}. The Bayesian GP forecasts based solely on historical data do not support the drastic warming trends predicted by most GCMs; only the Commitment scenario falls within high-density regions of their predictive distributions.

\subsection{Series Convergence for Climate Analysis}

The recursive Bayesian framework of \citet{roy2020bayesian2} has been applied to assess long-term climate dynamics. The analysis suggests that global temperature dynamics may be subject only to temporary variations under this model, with no evidence of long-term warming or cooling in the past, and no such adverse conditions likely in the future. This finding aligns remarkably with the forward forecasts of \citet{chatterjee2022ominous} and the benchmark forecast of \citet{green2009validity}, providing converging evidence from multiple Bayesian online learning perspectives. Again, these are statistical inferences and should be considered alongside physical climate models.

\subsection{Riemann Hypothesis Investigation}

The application of recursive Bayesian methods to the Riemann Hypothesis \citep{roy2020bayesian} reveals that the Dirichlet series for the Möbius function converges only for \(a \geq 0.72\) rather than all \(a > 1/2\). This suggests that the Riemann Hypothesis may not hold under this Bayesian analysis. Both the convergence assessment method and the multiple limit points extension yield consistent results, demonstrating the robustness of the Bayesian reflex across different mathematical formulations. The authors emphasise that this is not a mathematical proof.

\subsection{Prime Number Discovery}

The recursive Bayesian framework for prime numbers \citep{roy2025prime} has yielded remarkable practical results. Using TMCMC to sample from the posterior predictive distribution, the algorithm discovered 259 new primes exceeding 140 million, of which 184 are strong Mersenne prime exponent candidates. These correspond to potential Mersenne primes with digit lengths ranging from 42 million to 242 million. This work demonstrates that the Bayesian reflex is not merely a theoretical framework but a practical tool for discovery in pure mathematics.

\subsection{Stationarity Detection in Practice}

The recursive Bayesian framework for stationarity detection has been applied to numerous real-world datasets \citep{roy2021stationarity}. For spatial ozone data, the method correctly identifies nonstationarity, consistent with the findings of \citet{das2020nonstationary} who required nonstationary models for adequate fit. For spatio-temporal PM10 data (70,572 observations), the method clearly indicates nonstationarity, while for PM2.5 data, it correctly identifies stationarity, aligning with prior literature \citep{paciorek2009practical}.

In point process applications, the method successfully characterises 17 different process types, correctly identifying CSR status, Poisson vs. non-Poisson nature, and stationarity in each case. This includes subtle distinctions such as inhomogeneous log-Gaussian Cox processes where classical methods fail to reject CSR, and Matérn cluster processes with inhomogeneous parameters where the method correctly identifies non-CSR and nonstationarity.

\subsection{Frequency Analysis of Environmental Time Series}

The frequency determination methodology has been applied to the Southern Oscillation Index (SOI) and fish recruitment time series \citep{roy2021stationarity, shumway2006time}. For the SOI series, the method identifies characteristic frequencies around 0.02 and 0.08 cycles per month, corresponding to periods of approximately 4 years and 1 year—consistent with the El Niño Southern Oscillation cycle and annual cycle. For the recruitment series, the method identifies slightly slower frequencies, reflecting the biological response time to environmental forcing.

\begin{figure}[htbp]
\centering
\begin{tikzpicture}[scale=0.9, transform shape]
    % Central Bayesian Reflex
    \node (center) at (0,0) [circle, draw, fill=red!40, minimum size=3cm] {\Large\textbf{Bayesian Reflex}};
    
    % Application domains (arranged in a circle) – wider boxes and smaller font
    \foreach \angle/\name/\color in {
        0/Healthcare/green!40,
        45/Recommendation/blue!40,
        90/Climate/yellow!40,
        135/Mathematics/purple!40,
        180/Prime Numbers/orange!40,
        225/Spatial/cyan!40,
        270/Optimisation/brown!40,
        315/Time Series/pink!40
    } {
        \node (app\angle) at (\angle:5.5cm) [rectangle, draw, fill=\color, rounded corners, 
              text width=3.0cm, text centered, minimum height=1.2cm, font=\small] {\textbf{\name}};
        \draw[thick] (center) -- (app\angle);
    }
    
    % Key examples in outer ring – moved further out to radius 8.5cm
    \node at (0:8.5cm) [text width=2.5cm, text centered, font=\small] {Kilkari Program};
    \node at (45:8.5cm) [text width=2.5cm, text centered, font=\small] {Collaborative Bandits};
    \node at (90:8.5cm) [text width=2.5cm, text centered, font=\small] {GCM Evaluation};
    \node at (135:8.5cm) [text width=2.5cm, text centered, font=\small] {Riemann Hypothesis};
    \node at (180:8.5cm) [text width=2.5cm, text centered, font=\small] {Mersenne Primes};
    \node at (225:8.5cm) [text width=2.5cm, text centered, font=\small] {Point Processes};
    \node at (270:8.5cm) [text width=2.5cm, text centered, font=\small] {Derivative GP};
    \node at (315:8.5cm) [text width=2.5cm, text centered, font=\small] {Stationarity Detection};
    
    % Connecting lines to outer examples
    \foreach \angle in {0,45,90,135,180,225,270,315} {
        \draw[dashed] (app\angle) -- ++(\angle:3cm);
    }
    
    % Unifying principle
    \node at (0,-9) [rectangle, draw, fill=red!20, text width=14cm, text centered, minimum height=1.5cm] {
        \textbf{Unifying Principle:} Belief maintenance + Sequential updating + Uncertainty-driven action across all domains
    };
\end{tikzpicture}
\caption{The Bayesian reflex unifies diverse application domains through common mathematical and computational principles.}
\label{fig:applications}
\end{figure}
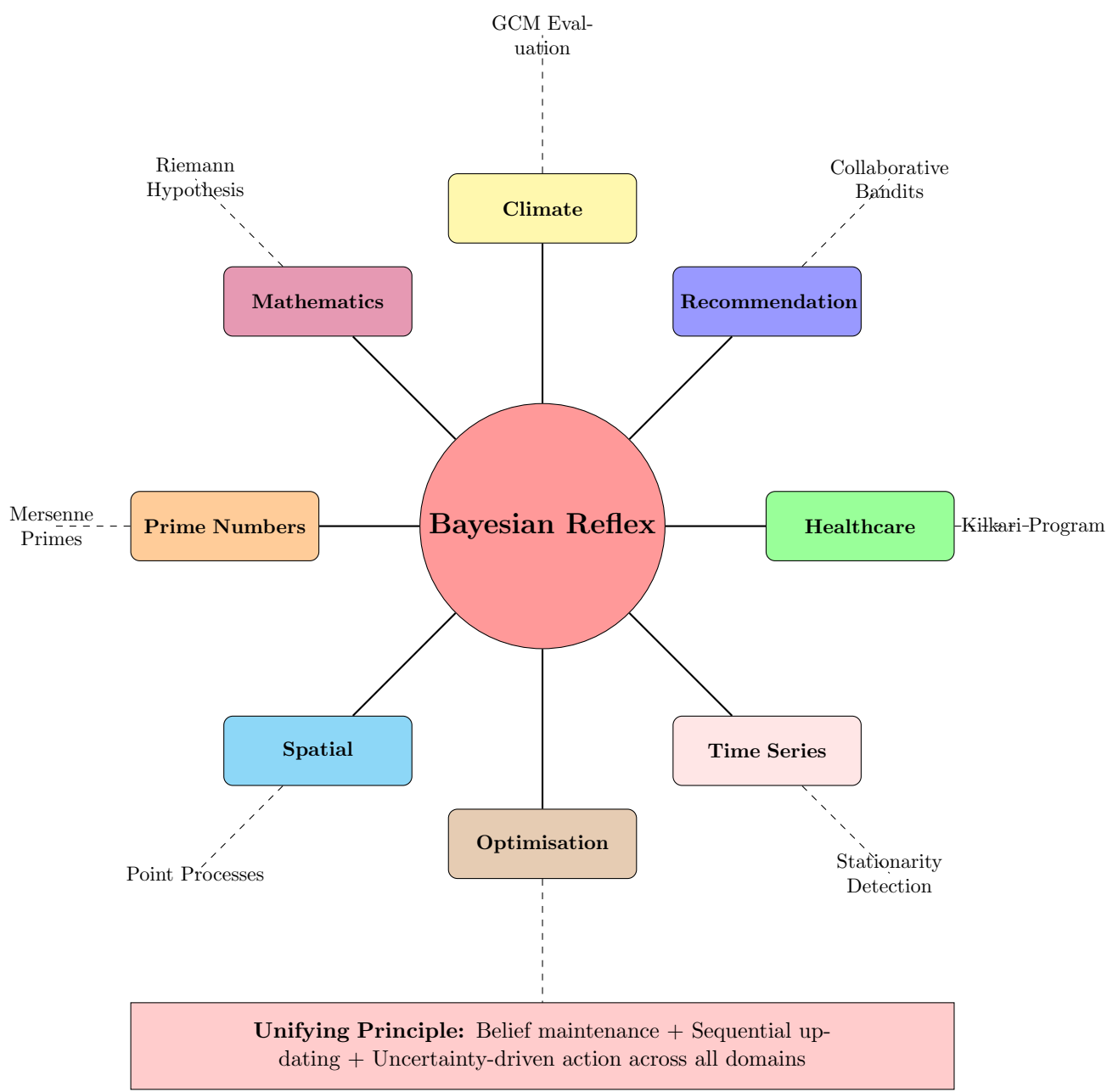

Figure \ref{fig:applications} illustrates the remarkable breadth of applications unified by the Bayesian reflex, from healthcare optimisation to prime number discovery, all sharing the same core principles of belief maintenance, sequential updating, and uncertainty-driven action.

\section{Computational Considerations and Scalability}

\subsection{Parallel Implementation of Ellipsoidal Decomposition}

The ellipsoidal decomposition framework is inherently parallelisable, as recognised in the original work:

\begin{quote}
``All our codes are written in C using the Message Passing Interface (MPI) protocol for parallel processing. We implemented our codes on a 80-core VMWare provided by Indian Statistical Institute. The machine has 2 TB memory and each core has about 2.8 GHz CPU speed.'' \citep{bhattacharya2025iid}
\end{quote}

This parallel architecture enables the framework to scale to high-dimensional problems. For the 160-dimensional Rongelap Island posterior, generation of 10,000 iid realisations took just 30 minutes—a task that would be infeasible with serial computation.

The key to this scalability is the independence across ellipsoidal regions: each processor can estimate \(\hat{\pi}(\mathbf{A}_i)\) for its assigned region without communication with others, except for the final broadcast of estimates. Similarly, when generating iid samples, each processor can independently sample from a selected component using Algorithm \ref{alg:perfect_sampling}. This ``embarrassingly parallel'' structure makes the framework ideal for modern computing clusters.

\subsection{Stochastic Gradient Langevin Dynamics}

SGLD combines stochastic optimisation with Bayesian inference \citep{welling2011bayesian}. The update is \(\theta_{t+1} = \theta_t + \frac{\epsilon_t}{2} \left( \nabla \log p(\theta_t) + \frac{N}{n} \sum_{i \in \mathcal{B}_t} \nabla \log p(x_i \mid \theta_t) \right) + \eta_t\) with \(\eta_t \sim \mathcal{N}(0, \epsilon_t I)\). Under appropriate conditions, the distribution of \(\theta_t\) converges to the posterior \citep{telgarsky2013langevin, dalalyan2017theoretical}.

\subsection{Variational Inference with Natural Gradients}

For exponential family variational distributions, natural gradient updates simplify \citep{amari2016information}. The natural gradient of the ELBO with respect to natural parameter \(\eta\) is \(\tilde{\nabla}_\eta \mathcal{L} = \mathbb{E}_{q_\eta}[\nabla_\eta \log p(\mathcal{D}, \theta)] - \eta\). This leads to simple updates \citep{hoffman2013stochastic}: \(\eta_{t+1} = (1 - \rho_t) \eta_t + \rho_t \mathbb{E}_{q_{\eta_t}}[\nabla_\eta \log p(\mathcal{D}, \theta)]\). For a Gaussian variational distribution \(q(\theta) = \mathcal{N}(\mu, \Sigma)\), the natural parameters are \(\eta_1 = \Sigma^{-1}\mu\) and \(\eta_2 = -\frac{1}{2}\Sigma^{-1}\) \citep{blei2017variational}.

\subsection{Parallel Implementation Strategies}

Both the sequential Bayesian optimisation and climate model evaluation frameworks leverage parallel computing for scalability. \citet{roy2021function} note that importance weights can be computed simultaneously on parallel processors, reducing computation time from hours to minutes for moderately sized problems.

The climate model evaluation work of \citet{chatterjee2022ominous} provides an excellent example of the scalability enabled by parallel computing and the look-up table methodology. As they describe:

\begin{quote}
``We wrote all our codes in the C language as efficiently as possible, parallelising them using the Message Passing Interface (MPI) protocol whenever relevant, for example, in the case of the Bayesian multiple testing procedure. In such a case, we implemented the GP models associated with the large number of GCM forecasts in the parallel computing architecture (VMWare) available at our institution. Very efficient and time-saving computations are the results of our parallel processing.''
\end{quote}

Specifically, they note that there are 75 posteriors of the form \([\bar{x}_1,\ldots,\bar{x}_{T_0}|\bar{x}_{T_0+1},\ldots,\bar{x}_T,\mathcal{M}_k]\) across all climate scenarios, and from each of them 60,000 realisations are simulated. This would be an infeasible task if implemented separately, but by parallelising across 100 cores on a VMWare cluster, the entire exercise takes less than an hour.

For the multivariate GP implementations, however, parallel MCMC algorithm could not be constructed, and the computation times ranged from 16 hours 59 minutes for the Commitment scenario to 30 hours 37 minutes for the A1B scenario, demonstrating the need for further investigation in the direction of parallel computing.

For the series convergence applications, the recursive Bayesian updates are amenable to massive parallelisation, enabling analysis of huge datasets. The Riemann Hypothesis investigation processed \(10^9\) terms of the Möbius function in just 2 minutes on a VMWare cluster, demonstrating the remarkable efficiency of the recursive Beta-Binomial updates when implemented in parallel computing architecture.

For prime number discovery, the recursive formulation is essential—a non-recursive approach would require \(2^k\) terms in the posterior, making analysis of millions of primes impossible. The recursive version requires only 2 terms for parameter posteriors and 4 terms for predictive distributions. TMCMC simulations on a standard laptop generate \(10^7\) realisations in seconds, and the primality testing for discovered candidates completes in under an hour. Thus, parallel processing is unnecessary in this setup.

For dynamic computer model emulation, the look-up table approach introduced by \citet{bhattacharya2007simulation} achieves its efficiency through a one-time matrix inversion, with all subsequent conditional simulations requiring only simple vector operations. This principle—invest computational effort upfront to enable fast sequential inference—has become a cornerstone of a range of online Bayesian methods introduced by these authors. The parallelisability of the approach, with independent sequences generated across processors, anticipates the scalability requirements of modern AI systems.

For state-space models, the look-up table approach provides similar computational benefits. As \citet{ghosh2014bayesian} note, retaining the auxiliary variables \(\mathbf{D}_n^*\) ensures that the correlation matrix \(\mathbf{A}_{g,D_n^*}\) needs to be inverted only once, before MCMC simulations begin, and remains fixed thereafter, ``saving a lot of computational time, in addition to providing protection against numerical instability.''

For the circular extensions, the look-up table methodology remains equally efficient. The fixed grid \(\mathbf{G}_z\) ensures that \(\mathbf{A}_{g,D_z}^{-1}\) can be pre-computed, and the conditional independence structures induced by the auxiliary variables make sequential updating feasible even in the most complex fully circular case.

For Recursive Gaussian Processes, the look-up table approximation enables scalable inference by reducing infinite-dimensional GP inference to finite-dimensional computations. The Markov structure induced by the look-up table makes the RGP framework particularly amenable to parallelisation, with layer-wise computations distributable across processors. While the original batch MCMC implementation required substantial computational resources \citep{bhattacharya2025bayesian}, the path toward online inference—whether through particle filters or sequential variational methods—promises to make RGPs practical for real-time learning scenarios.

For the recursive stationarity detection framework, parallel implementation is essential for handling large datasets. For spatial data with 10,000 observations, K-means clustering into \(K=250\) clusters with parallel computation takes less than a second on a 4-core laptop. For spatio-temporal data with 70,572 observations and \(K=1764\) clusters, parallel implementation enables analysis in minutes rather than hours \citep{roy2021stationarity}.

\subsection{Computational Challenges in High Dimensions}

As dimension increases, several challenges emerge in both optimisation and time series settings. The Euclidean norm of gradients increases with dimension, requiring larger thresholds. Matrix operations become more expensive, requiring \(\mathcal{O}(d^3)\) operations for Cholesky decompositions. Markov chain mixing degrades, requiring longer burn-in periods up to \(10^6\) iterations for \(d=100\). The number of stages with positive acceptance decreases, with only the first few stages yielding progress. Despite these challenges, the methods remain feasible and outperform classical alternatives.

For RGPs in high-dimensional settings, the spike-and-slab priors on weights provide automatic variable selection, pruning irrelevant input dimensions and hidden units. This feature is particularly valuable for online learning in high dimensions, as it allows the model to focus computational resources on the most informative features as data accumulates.

For the ellipsoidal decomposition framework, high dimensions present both challenges and opportunities. The minorisation constant \(\hat{p}_i = \hat{s}_i/\hat{S}_i - \eta_i\) tends to decrease as dimension increases, potentially increasing the coalescence time \(T_i \sim \text{Geometric}(\hat{p}_i)\). However, the framework incorporates a powerful tool to address this: \textit{diffeomorphic transformations}. As demonstrated in the Rongelap Island example (Section 6.3 of \citet{bhattacharya2025iid}), applying an appropriate transformation can flatten the target distribution, making \(\hat{s}_i\) and \(\hat{S}_i\) closer and thus increasing \(\hat{p}_i\). This transformed approach successfully handled a 160-dimensional posterior. %uggesting that even higher dimensions may be tractable with carefully chosen transformations.
As detailed in \cite{bhattacharya2025flows}, the same transformation approach was also able to yield near-perfect realizations from a $4776$-dimensional posterior; however,
the times taken were about $2$ days and $10$ days for two strategic versions of the method, implemented on a (overloaded) server with $100$ cores.

\subsection{Federated Bayesian Learning}

In federated learning, we have \(K\) clients with local data \(\mathcal{D}_k\) \citep{mcmahan2017communication}. The global posterior is \(p(\theta \mid \cup_k \mathcal{D}_k) \propto p(\theta) \prod_{k=1}^K p(\mathcal{D}_k \mid \theta)\). Federated averaging approximates this by communicating sufficient statistics \citep{mcmahan2017communication}. For exponential family models with conjugate priors, the global posterior hyperparameters are \(\chi_{\text{global}} = \chi_0 + \sum_{k=1}^K (\chi_k - \chi_0)\) and \(\nu_{\text{global}} = \nu_0 + \sum_{k=1}^K (\nu_k - \nu_0)\).

For privacy, we can add Gaussian noise to achieve differential privacy \citep{dwork2014algorithmic, abadi2016deep}: \(\tilde{\chi}_{\text{global}} = \chi_{\text{global}} + \mathcal{N}(0, \sigma^2 I)\), where \(\sigma\) is chosen to satisfy \((\epsilon, \delta)\)-differential privacy according to the Gaussian mechanism \citep{dwork2014algorithmic}: \(\sigma = \frac{\Delta \sqrt{2\log(1.25/\delta)}}{\epsilon}\), with \(\Delta\) being the sensitivity of the sufficient statistics.

The ellipsoidal decomposition framework offers a novel approach to federated learning: each client could generate near-perfect iid samples from its local posterior, and these samples could be combined to approximate the global posterior via importance sampling or direct pooling. This would eliminate the need for iterative communication and provide exact uncertainty quantification at the global level.

The RGP framework, with its hierarchical structure and look-up table representation, offers promising directions for federated online learning. Clients could maintain local RGPs and communicate compressed summaries (e.g., sufficient statistics of the look-up table or variational parameters) to a central server, which then aggregates them to form a global model—all while preserving privacy through differential privacy mechanisms.

\section{Future Directions and Open Challenges}

\subsection{Theoretical Frontiers}

The Eluder dimension for neural networks remains an open problem \citep{osband2016deep}. For a network with \(L\) layers and \(P\) parameters, we conjecture \(d_E = \tilde{\mathcal{O}}(P L \log(1/\delta))\), but this has not been proven. Current bounds for ReLU networks are exponential in depth \citep{allen2019learning}. Non-stationary regret bounds are another active area \citep{besbes2014stochastic}. For a process with \(V_T\) changes, the optimal regret is \(R_T^* = \tilde{\Theta}(\sqrt{V_T T})\). Achieving this bound for general Bayesian models requires adaptive forgetting mechanisms \citep{kaelin2021bayesian}.

For dynamic computer model emulation, theoretical questions remain about optimal selection of the look-up table \(\mathbf{G}^*\). The original work \citep{bhattacharya2007simulation} demonstrated empirically that a grid size of \(N=20\) was adequate for the examples considered, but general theoretical guarantees for choosing grid size and point locations remain an open area for investigation.

For state-space models, \citet{ghosh2014bayesian} proved that the order of approximation of the true function by the conditional distribution given \(\mathbf{D}_n^*\) is \(O(n^{-1})\). However, extending these results to the circular case and establishing optimal grid selection criteria remains an open challenge.

For the climate modelling application, \citet{chatterjee2022ominous} note that their grid choice of \(n=50\) points for each dimension, selected via Latin hypercube sampling, proved adequate for their purposes. However, developing adaptive grid selection methods that respond to the data would be a valuable extension.

For Recursive Gaussian Processes, theoretical questions remain about the optimal depth \(T\) for online learning scenarios. The asymptotic theory of \citet{bhattacharya2025bayesian} establishes near-optimal convergence rates for batch inference, but extending these results to online settings with non-stationary data streams remains an open challenge. The interplay between the look-up table approximation and sequential inference also deserves further theoretical investigation.

For series convergence, theoretical questions remain about optimal choices of the sequences \(\{n_j\}\) and \(\{c_j\}\). The nonparametric adaptive bound of \citet{roy2020bayesian2} works well empirically, but theoretical guarantees for its convergence rate remain to be established.

For prime number analysis, open questions include extending the framework to other classes of primes (e.g., Sophie Germain primes, twin primes) and developing iid sampling procedure with the ellipsoidal decomposition framework for exploring the posterior predictive distribution. The theoretical connection between the recursive Bayesian framework and the Riemann Hypothesis deserves further investigation.

For the recursive stationarity detection framework, theoretical questions remain about optimal choices of the region sizes \(n_j\) and the bound sequence \(c_j\). The adaptive bound construction used in AR(1) examples—where \(\hat{C}_j\) is updated based on previous outcomes—works well empirically, but theoretical guarantees for its convergence rate remain to be established. Extending the framework to handle non-ergodic processes and establishing rates of convergence for the posterior probabilities are important open problems \citep{roy2021stationarity}.

For point process applications, the assumption of growing observation windows \(|W_r| \to \infty\) with \(\sum_{r=1}^\infty |W_r|^{-1} < \infty\) ensures almost sure convergence via the Borel-Cantelli lemma. Relaxing this condition and establishing convergence rates for finite windows would enhance practical applicability.

\subsection{Theoretical Frontiers for Ellipsoidal Decomposition}

The ellipsoidal decomposition framework opens several theoretical questions. First, the optimal choice of radii: the framework requires choosing radii \(\sqrt{c_i}\) and the number of regions \(M\). While heuristic guidelines exist (e.g., Section 5 of \citet{bhattacharya2025iid}), theoretical results on optimal choices would be valuable. The trade-off is clear: larger \(\sqrt{c_1}\) reduces \(\hat{p}_1\) (slower coalescence) but captures more probability mass, while smaller \(\sqrt{c_1}\) increases \(\hat{p}_1\) but may require more regions \(M\). Formalising this trade-off is an open problem. Second, convergence rates for \(\hat{p}_i\): the minorisation constant \(\hat{p}_i = \hat{s}_i/\hat{S}_i - \eta_i\) determines the coalescence time distribution. Understanding how \(\hat{p}_i\) scales with dimension and the shape of the target distribution would provide insights into the scalability of the method. The Rongelap example suggests that with appropriate diffeomorphic transformations, \(\hat{p}_i\) can remain practically useful even in 160 dimensions. Third, error bounds from Monte Carlo estimation: the estimates \(\hat{\pi}(\mathbf{A}_i)\) and \(\hat{p}_i\) are based on finite Monte Carlo samples. Developing rigorous error bounds and adaptive sample size selection rules would strengthen the theoretical foundations. Fourth, integration with sequential inference: while the framework as presented is for batch sampling, extending it to true online settings where the posterior evolves over time requires theoretical guarantees on how the ellipsoidal decomposition should adapt as new data arrives.

\subsection{Methodological Innovations}

Efficient approximations for online Bayesian deep learning include particle-based methods where \(p(\theta \mid \mathcal{D}_t) \approx \sum_{i=1}^N w_t^{(i)} \delta_{\theta_t^{(i)}}(\theta)\), variational methods where \(q_t(\theta) = \mathcal{N}(\mu_t, \Sigma_t)\), and Laplace approximations where \(p(\theta \mid \mathcal{D}_t) \approx \mathcal{N}(\hat{\theta}_t, H_t^{-1})\), with \(H_t\) being the Hessian at the MAP estimate \citep{osband2016deep}. The Laplace approximation requires computing the Hessian, which is \(\mathcal{O}(P^2)\) for \(P\) parameters \citep{mackay1992practical}.

For dynamic computer model emulation, future directions include adaptive look-up tables that expand or contract based on the observed trajectory, and hybrid approaches combining sparse grids with importance sampling for improved efficiency.

For state-space models with circular variables, future directions include developing iid sampling procedures for the fully circular case, and extending the methodology to handle multivariate circular processes \citep{mazumder2016nonparametric}.

For the climate modelling application, \citet{chatterjee2022ominous} suggest several extensions:

\begin{quote}
``Our framework could be extended to incorporate spatio-temporal models, dynamic covariates, or hybrid approaches that integrate physical constraints into the statistical emulation process. We also anticipate that future work may apply this framework to other environmental processes, such as sea-level rise or precipitation extremes, where model uncertainty remains high and decision-making stakes are critical.''
\end{quote}

They also note the potential for combining multiple data sources (NASA, NOAA, World Bank, National Snow and Ice Data Center) coherently, which would require very complex nonparametric Bayesian multivariate spatio-temporal models, likely necessitating supercomputing facilities with sophisticated parallel MCMC strategies. For iid sampling, significant extension of the ellipsoidal decomposition framework may be necessary.

For Recursive Gaussian Processes, developing true online inference algorithms is a priority. Promising directions include particle-based RGP, where a set of weighted particles represents the joint distribution of all latent variables, with sequential importance resampling as new data arrives; variational RGP, where a variational posterior over the look-up table and parameters is maintained, with updates based on stochastic gradients of a sequential ELBO; and recursive RGP, which exploits the state-space structure to derive approximate closed-form updates, extending the Kalman filter to the nonparametric, deep setting. This latter approach would represent the ultimate realisation of the Bayesian reflex in deep learning.

For sequential Bayesian optimisation specifically, key challenges include replacing the TMCMC procedures by iid sampling with ellipsoidal decompositions 
for high-dimensional posterior sampling, improving importance sampling stability through adaptive tempering, reducing computational cost of matrix inversions through low-rank approximations, and extending theoretical convergence rates to more general function classes.

For series convergence applications, extending the framework to complex-valued series and developing more sophisticated adaptive bound mechanisms remain open challenges.

For prime number analysis, future directions include scaling the approach to even larger primes, developing theoretical guarantees for the rate of prime discovery, and exploring connections to other problems in analytic number theory.

For the recursive stationarity detection framework, promising methodological directions include adaptive region selection, which involves developing methods to automatically choose the region sizes \(n_j\) and the number of regions \(K\) based on the data, rather than fixing them a priori; multivariate extensions, generalising the framework to multivariate time series and spatial processes with multiple variables; nonparametric covariance estimation, incorporating more sophisticated covariance estimators for weak stationarity testing, particularly for non-Gaussian processes; online point process monitoring, developing streaming algorithms for real-time monitoring of point process characteristics, with applications to environmental monitoring and surveillance; and integration with deep learning, combining the recursive Bayesian framework with neural network-based feature extraction for high-dimensional data.

\subsection{Methodological Innovations for Ellipsoidal Decomposition}

The ellipsoidal decomposition framework itself suggests numerous methodological extensions. First, adaptive region construction: currently, the radii \(\sqrt{c_i}\) are fixed before estimation; an adaptive approach could dynamically adjust the regions based on the estimated probability masses, allocating more regions where the density varies rapidly and fewer where it is flat. Second, automatic diffeomorphic transformations: the Rongelap example demonstrates the power of diffeomorphic transformations to improve minorisation constants; developing automatic methods to select appropriate transformations based on preliminary samples would greatly enhance the framework's applicability. Third, hybrid approaches for extremely high dimensions: for dimensions beyond 4776, a hybrid approach combining variational inference (for initial dimension reduction) with ellipsoidal decomposition (for exact sampling in the reduced space) could be explored. Fourth, integration with Hamiltonian Monte Carlo: the perfect sampling framework currently uses Metropolis-Hastings with uniform proposals; extending to Hamiltonian Monte Carlo proposals, specifically the tunable and robust Modified Parameterised Leapfrog Hamiltonian Monte Carlo (MPL-HMC) recently introduced by \citep{Bhatta26}, within the ellipsoidal regions, could improve efficiency for highly correlated targets. Fifth, streaming ellipsoidal decomposition: developing a truly online version where the ellipsoidal regions are updated incrementally as new data arrives, without recomputing all Monte Carlo estimates from scratch, would enable real-time perfect sampling.

\subsection{Ethical Considerations}

Fairness constraints can be incorporated into Thompson sampling \citep{joseph2016fairness}. Let \(G(a)\) be a protected group for action \(a\). A fairness-aware Thompson sampling algorithm might sample parameters and then enforce \(P(\text{select } a \mid G(a) = g) \geq \gamma\) for all groups \(g\) and some threshold \(\gamma\). This can be implemented by reweighting the posterior samples or adding constraints to the optimisation \citep{celis2018fair}.

Privacy guarantees for online Bayesian inference remain challenging \citep{abadi2016deep}. The cumulative nature of the posterior means that privacy loss compounds over time. Advanced composition theorems provide bounds \(\epsilon_{\text{total}} \leq \sqrt{2T \log(1/\delta)} \epsilon + T \epsilon (e^\epsilon - 1)\) \citep{dwork2014algorithmic}, but these may be too loose for long-running systems. Tighter bounds using moments accountant methods are an active research area \citep{abadi2016deep}.

For Recursive Gaussian Processes deployed in online learning systems, privacy concerns are paramount. The hierarchical structure may leak information across layers, and the look-up table representation could potentially encode sensitive patterns in the data. Developing differentially private variants of RGP inference—whether through noisy gradient updates, private aggregation of sufficient statistics, or objective perturbation—is an important direction for future work.

In the context of climate modelling, ethical considerations are paramount. Policy decisions based on model projections affect billions of lives. The work of \citet{chatterjee2022ominous} demonstrates the importance of independent statistical validation, showing that GCM projections may overstate warming trends under this model. As they note, their results align with the recommendation of \citet{green2009validity} that the best policy may be to do nothing about global warming—at least until stronger Bayesian statistical evidence emerges from other significant climate data analyses. However, this must be weighed against the overwhelming consensus of physical climate science \citep{ipcc2021}. Statisticians have a responsibility to communicate the limitations of their models and the uncertainty in their conclusions.

The series convergence results of \citet{roy2020bayesian2} further support this view under their model, suggesting that current warming may be a temporary variation rather than a long-term trend. Such findings have profound implications for climate policy, underscoring the responsibility of statisticians to provide rigorous, independent evaluations of scientific claims that shape public policy, while also acknowledging the broader scientific context.

For point process applications in epidemiology and ecology, ethical considerations include ensuring that conclusions about clustering or randomness are not misinterpreted as evidence for or against causal mechanisms. The recursive Bayesian framework provides rigorous uncertainty quantification, but responsible communication of results remains essential.

\subsection{Ethical Considerations for Perfect Sampling}

The ellipsoidal decomposition framework, by providing practically exact iid samples from posterior distributions, has important ethical implications for AI systems. First, transparency and reproducibility: essentially exact sampling eliminates the ``black box'' nature of MCMC convergence diagnostics; decisions made based on these samples are fully reproducible and verifiable, enhancing transparency in high-stakes applications like healthcare and criminal justice. Second, fairness through exact uncertainty quantification: the essentially exact samples enable precise quantification of uncertainty about model parameters, which can inform fairness assessments; for example, if a model's predictions for different demographic groups overlap substantially in their uncertainty intervals, this provides evidence against disparate impact. Third, privacy-preserving exact inference: the parallel nature of the framework suggests possibilities for distributed exact inference where data never leaves local processors—only summary statistics (the Monte Carlo estimates \(\hat{\pi}(\mathbf{A}_i)\)) are shared; this could enable privacy-preserving Bayesian analysis across institutions. Fourth, accountability in scientific discovery: the prime number discovery application demonstrates that exact Bayesian methods can lead to verifiable scientific discoveries; in applications like climate modelling or drug discovery, such accountability is crucial for public trust.

\section{Conclusion}

This chapter has provided a comprehensive mathematical treatment of the Bayesian reflex, from foundational principles to state-of-the-art applications. We have shown how Bayes' theorem provides a universal update rule, how conjugate families and exponential families enable tractable computation, and how modern methods like particle filters and variational inference scale to complex models. The connections to decision-making through bandit algorithms demonstrate how Bayesian principles translate directly into action, while applications across healthcare, recommendation, and language modelling show the practical impact of these ideas.

The foundational work of \citet{bhattacharya2007simulation} on dynamic computer model emulation established a core principle that is central to online Bayesian learning methods
proposed by these authors: introducing latent auxiliary variables (look-up tables) creates conditional independence structures that make sequential inference computationally tractable. This insight—that a one-time investment in matrix inversion enables efficient, parallelisable simulation of entire dynamic sequences—anticipates many of the scalability requirements of modern AI systems.

Building on this foundation, \citet{ghosh2014bayesian} generalised the look-up table principle to nonparametric state-space models, showing that the same core idea enables efficient inference when both the observational and evolutionary functions are unknown and time-varying. This work explicitly elucidated the profound insight that the latent states, while conditionally Markov given the look-up table, become non-Markovian upon marginalisation—a realistic dependence structure that emerges naturally from the Bayesian nonparametric approach.

The extensions to circular latent states \citep{mazumder2016bayesian, mazumder2016nonparametric} demonstrate the remarkable versatility of the look-up table methodology. By introducing additional auxiliary variables—wrapping numbers that track how many times a process has circled the unit circle—these works show how the Bayesian reflex can handle variables that live on non-Euclidean manifolds. The fully circular case, where both observed and latent processes are circular, represents the most challenging extension, yet the look-up table principle enables tractable inference through a hierarchical structure of auxiliary variables.

The inverse regression framework of \citet{chatterjee2022ominous} further demonstrates that the Bayesian reflex operates in multiple temporal directions—not just forward prediction, but also backward inference conditioning on future scenarios. Their compositional Gaussian process emulator, look-up table approximations (directly inherited from \citet{bhattacharya2007simulation}), Bayesian multiple testing for model selection, and multivariate GP extensions provide a rigorous statistical framework for evaluating complex physical models. The empirical findings cast doubt on the fidelity of GCM projections under this model, while their forward forecasts based solely on historical data suggest more moderate warming trajectories.

The Recursive Gaussian Process framework \citet{bhattacharya2025bayesian} represents a further significant advance in Bayesian deep learning, exploiting the look-up table principle for multiple layers and offering a principled approach to hierarchical modelling with uncertainty quantification. While originally developed for batch inference, the RGP architecture is ideally suited for online learning. Its hierarchical structure mirrors state-space models, and even here, only a single look-up table is needed to induce conditional independence for sequential processing. Moreover, its spike-and-slab priors enable automatic feature selection over time. The theoretical guarantees—including near-optimal convergence rates and robustness to misspecification—provide confidence that online RGP systems will learn reliably from streaming data.

The work of \citet{roy2021function} on sequential Bayesian optimisation with Gaussian derivative processes exemplifies the Bayesian reflex in a setting where the system actively decides which points to evaluate next. The algorithm's iterative data augmentation, importance resampling, and theoretical convergence guarantees embody the core principles of online Bayesian learning: belief maintenance, sequential updating, and uncertainty-driven action.

The recursive Bayesian framework of \citet{roy2020bayesian, roy2020bayesian2} shows that the Bayesian reflex operates at the most fundamental level of mathematical analysis—determining whether an infinite series converges. The applications to climate change and the Riemann Hypothesis demonstrate the profound reach of these ideas. The climate analysis suggests that global temperature dynamics may be subject only to temporary variations under this model, with no evidence of long-term warming or cooling, aligning remarkably with the findings of \citet{chatterjee2022ominous}. The Riemann Hypothesis investigation suggests that this famous conjecture may not hold under this Bayesian analysis, with the Dirichlet series for the Möbius function converging only for \(a \geq 0.72\) rather than all \(a > 1/2\).

Most remarkably, the recursive Bayesian framework for prime numbers \citet{roy2025prime} demonstrates that the Bayesian reflex extends to the purest realm of mathematics. The asymptotic analysis provides a second, independent statistical argument against the Riemann Hypothesis, while the posterior predictive sampling yields practical discoveries of 259 new primes exceeding 140 million, including 184 strong Mersenne prime candidates with potential digit lengths reaching 242 million. The recursive formulation reduces an infeasible \(2^k\)-term posterior to just 2 terms, making analysis of millions of primes computationally possible.

The recursive stationarity detection framework \citep{roy2021stationarity} demonstrates the versatility of the Bayesian reflex across an unprecedented range of statistical problems. From ARMA time series to spatial data, from MCMC diagnostics to point processes, from environmental monitoring to frequency analysis—the same core principles apply: partition the index set into locally stationary regions, construct binary indicators based on empirical distributions, and apply recursive Beta-Binomial updates. The theoretical convergence guarantees ensure that as data accumulates, the posterior probabilities correctly identify the underlying properties of the process.

The applications to 17 different point process types—including Poisson processes, Cox processes, cluster processes, and interaction processes—demonstrate the method's ability to distinguish subtle differences in process characteristics. The successful identification of non-CSR in inhomogeneous log-Gaussian Cox processes where classical methods fail, and the correct classification of stationarity in Matérn cluster processes with inhomogeneous parameters, highlight the power of the Bayesian approach.

\subsection{The Ellipsoidal Decomposition: A Foundational Contribution}

At the heart of this chapter lies the recognition that the Bayesian reflex requires a computational foundation capable of exact, reliable inference. The {\it ellipsoidal decomposition framework} \citep{bhattacharya2025iid} provides precisely this foundation. By decomposing any target distribution on \(\mathbb{R}^d\) into an infinite mixture of component distributions on concentric ellipsoids and annuli, this framework enables near-perfect iid sampling without the convergence concerns that have plagued MCMC methods for decades.

The theoretical elegance of the framework—the minorisation inequality for Metropolis-Hastings with uniform proposals, the geometric coalescence time, the split-chain representation that bypasses CFTP—is matched by its practical performance. From standard distributions in dimensions up to 100 to challenging 160-dimensional spatial posteriors, the framework delivers practically exact samples with reasonable computation times, particularly when implemented in parallel architectures.

The extensions of this framework to multimodal distributions \citep{bhattacharya2021multimodal}, doubly intractable distributions \citep{bhattacharya2021doubly}, Dirichlet process mixtures \citep{bhattacharya2022dirichlet}, and random normalising flows \citep{bhattacharya2025flows} demonstrate its versatility and power. Together with the look-up table principle, it forms a complete computational infrastructure for the Bayesian reflex—one that handles both the sequential dynamics of online learning (via look-up tables) and the exact inference required at each step (via ellipsoidal decomposition).

\subsection{A Unified Vision}

The Bayesian reflex, as we have articulated it, is not merely a metaphor but a concrete computational framework for building adaptive AI systems. It rests on three pillars: first, mathematical foundations—Bayes' theorem, sequential updating, conjugate families, and information theory provide the theoretical underpinning; second, computational principles—the look-up table principle enables sequential inference in dynamic systems and the ellipsoidal decomposition framework enables essentially exact iid sampling from complex posteriors; and third, algorithmic implementations—from Kalman filters to particle filters, from variational inference to Thompson sampling, from RGPs to perfect samplers, these algorithms realise the reflex in practice.

The applications we have surveyed—healthcare optimisation, climate modelling, prime number discovery, Riemann Hypothesis investigation, spatial statistics, point process analysis—demonstrate the reach of these ideas. In each case, the Bayesian reflex operates: beliefs are maintained probabilistically, updated sequentially as new information arrives, and used to guide decisions under uncertainty.

As we look to the future, the Bayesian reflex will remain essential for navigating the complex, dynamic, and uncertain world that intelligent systems must inhabit. The challenges ahead—scaling to foundation models with billions of parameters, providing theoretical guarantees in non-stationary environments, ensuring fairness and privacy in adaptive systems, developing truly online algorithms for powerful architectures like RGPs, extending the recursive Bayesian framework to even more complex data types, and perfecting the ellipsoidal decomposition for ever-higher dimensions—demand continued innovation. But the core insight endures: learning is not a one-time event, but a continuous process of belief updating in response to experience. This is the essence of the Bayesian reflex, and it will guide the next generation of artificial intelligence.

\section*{Acknowledgments}

We are grateful to the many researchers whose work has shaped our understanding of Bayesian online learning. We also sincerely thank DeepSeek for assistance in 
drafting and formatting this manuscript; all technical content has been verified by these authors.

\bibliographystyle{plainnat}
\bibliography{online_refs3}

@phdthesis{Soma_thesis,
  author    = {Sucharita Roy},
  title     = {{Infinite Series, Stochastic Processes, Function Optimization and the bayesian Panacea}},
  school = {Department of Mathematics, Jadavpur University},
  year      = {2023},
  note  = {Available at \url{https://www.researchgate.net/publication/360457067_INFINITE_SERIES_STOCHASTIC_PROCESSES_FUNCTION_OPTIMIZATION_AND_THE_BAYESIAN_PANACEA}}
}

@phdthesis{Debashis_thesis,
  author    = {Debashis Chatterjee},
  title     = {{A Brief Treatise on Bayesian Inverse Regression}},
  school = {Indian Statistical Institute},
  year      = {2022},
  note  = {Available at \url{https://www.researchgate.net/publication/360457067_INFINITE_SERIES_STOCHASTIC_PROCESSES_FUNCTION_OPTIMIZATION_AND_THE_BAYESIAN_PANACEA}}
}

@inproceedings{abadi2016deep,
  author    = {Abadi, Martin and Chu, Andy and Goodfellow, Ian and McMahan, H. Brendan and Mironov, Ilya and Talwar, Kunal and Zhang, Li},
  title     = {{Deep Learning with Differential Privacy}},
  booktitle = {Proceedings of the 2016 ACM SIGSAC Conference on Computer and Communications Security},
  pages     = {308--318},
  year      = {2016},
  publisher = {ACM}
}

@inproceedings{agrawal2012analysis,
  author    = {Agrawal, Shipra and Goyal, Navin},
  title     = {{Analysis of Thompson Sampling for the Multi-Armed Bandit Problem}},
  booktitle = {Proceedings of the 25th Annual Conference on Learning Theory},
  series    = {JMLR Workshop and Conference Proceedings},
  volume    = {23},
  pages     = {39.1--39.26},
  year      = {2012}
}

@inproceedings{akyurek2022what,
  author    = {Aky{\"u}rek, Ekin and Schuurmans, Dale and Andreas, Jacob and Ma, Tengyu and Zhou, Denny},
  title     = {{What Learning Algorithm is In-Context Learning? Investigations with Linear Models}},
  booktitle = {International Conference on Learning Representations},
  year      = {2022}
}

@inproceedings{allen2019learning,
  author    = {Allen-Zhu, Zeyuan and Li, Yuanzhi and Song, Zhao},
  title     = {{Learning and Generalization in Overparameterized Neural Networks}},
  booktitle = {Proceedings of the 32nd Annual Conference on Learning Theory},
  pages     = {830--882},
  year      = {2019}
}

@book{amari2016information,
  author    = {Amari, Shun-ichi},
  title     = {{Information Geometry and Its Applications}},
  publisher = {Springer},
  year      = {2016}
}

@book{anderson1979optimal,
  author    = {Anderson, Brian D. O. and Moore, John B.},
  title     = {{Optimal Filtering}},
  publisher = {Prentice-Hall},
  year      = {1979}
}

@article{auer2002finite,
  author    = {Auer, Peter and Cesa-Bianchi, Nicol{\`o} and Fischer, Paul},
  title     = {{Finite-time Analysis of the Multiarmed Bandit Problem}},
  journal   = {Machine Learning},
  volume    = {47},
  number    = {2},
  pages     = {235--256},
  year      = {2002}
}

@article{auer2002nonstochastic,
  author    = {Auer, Peter and Cesa-Bianchi, Nicol{\`o} and Freund, Yoav and Schapire, Robert E.},
  title     = {{The Nonstochastic Multiarmed Bandit Problem}},
  journal   = {SIAM Journal on Computing},
  volume    = {32},
  number    = {1},
  pages     = {48--77},
  year      = {2002}
}

@article{bandopadhyay2017nonparametric,
  author    = {Bandopadhyay, S. and Jentsch, C. and Rao, S. S.},
  title     = {{A Spectral Domain Test for Stationarity of Spatio-Temporal Data}},
  journal   = {Journal of Time Series Analysis},
  volume    = {38},
  pages     = {326--351},
  year      = {2017}
}

@incollection{bengtsson2008curse,
  author    = {Bengtsson, Thomas and Bickel, Peter and Li, Bo},
  title     = {{Curse-of-dimensionality Revisited: Collapse of the Particle Filter in very Large Scale Systems}},
  booktitle = {Probability and Statistics: Essays in Honor of David A. Freedman},
  pages     = {316--334},
  publisher = {Institute of Mathematical Statistics},
  year      = {2008}
}

@book{berger2013statistical,
  author    = {Berger, James O.},
  title     = {{Statistical Decision Theory and Bayesian Analysis}},
  publisher = {Springer},
  year      = {2013}
}

@book{bernardo2009bayesian,
  author    = {Bernardo, Jos{\'e} M. and Smith, Adrian F. M.},
  title     = {{Bayesian Theory}},
  publisher = {Wiley},
  year      = {2009}
}

@article{besbes2014stochastic,
  author    = {Besbes, Omar and Gur, Yonatan and Zeevi, Assaf},
  title     = {{Stochastic Optimization with Time-Varying Models}},
  journal   = {Mathematics of Operations Research},
  volume    = {39},
  number    = {4},
  pages     = {1021--1047},
  year      = {2014}
}

@article{bhattacharya2007simulation,
  author    = {Bhattacharya, Sourabh},
  title     = {{A Simulation Approach to Bayesian Emulation of Complex Dynamic Computer Models}},
  journal   = {Bayesian Analysis},
  volume    = {2},
  number    = {4},
  pages     = {783--816},
  year      = {2007}
}

@article{bhattacharya2025iid,
  author    = {Bhattacharya, Sourabh},
  title     = {{IID Sampling from Intractable Distributions}},
  journal   = {Sankhy{\=a} A},
  note      = {To appear in Professor C. R. Rao special issue},
  year      = {2025}
}

@article{bhattacharya2025flows,
  author    = {Bhattacharya, Durba and Maitra, Trisha and Roy, Sucharita and Bhattacharya, Sourabh},
  title     = {{Bayesian Nonparametrics with Random Normalizing Flows}},
  journal   = {ResearchGate preprint},
  year      = {2025},
  note      = {Available at \url{https://www.researchgate.net/publication/395337211_Bayesian_Nonparametrics_With_Random_Normalizing_Flows}}
}

@article{bhattacharya2025bayesian,
  author    = {Bhattacharya, Durba and Maitra, Trisha and Roy, Sucharita and Bhattacharya, Sourabh},
  title     = {{Bayesian Deep Neural Networks Driven by Recursive Gaussian Processes}},
  journal   = {ResearchGate preprint},
  year      = {2025}
}

@article{bhattacharya2021multimodal,
  author    = {Bhattacharya, Sourabh},
  title     = {{IID Sampling from Intractable Multimodal and Variable-Dimensional Distributions}},
  journal   = {arXiv preprint},
  year      = {2021},
  note      = {arXiv:2109.12633}
}

@article{bhattacharya2021doubly,
  author    = {Bhattacharya, Sourabh},
  title     = {{IID Sampling from Doubly Intractable Distributions}},
  journal   = {arXiv preprint},
  year      = {2021},
  note      = {arXiv:2112.07939}
}

@article{bhattacharya2022dirichlet,
  author    = {Bhattacharya, Sourabh},
  title     = {{IID Sampling from Posterior Dirichlet Process Mixtures}},
  journal   = {arXiv preprint},
  year      = {2022},
  note      = {arXiv:2206.09233}
}

@article{Bhatta26,
  author    = {Bhattacharya, Sourabh},
  title     = {{MPL-HMC: A Tunable Parameterized Leapfrog Framework for Robust Hamiltonian Monte Carlo}},
  journal   = {arXiv preprint},
  year      = {2026}
}

@article{blei2017variational,
  author    = {Blei, David M. and Kucukelbir, Alp and McAuliffe, Jon D.},
  title     = {{Variational Inference: A Review for Statisticians}},
  journal   = {Journal of the American Statistical Association},
  volume    = {112},
  number    = {518},
  pages     = {859--877},
  year      = {2017}
}

@inproceedings{blundell2015weight,
  author    = {Blundell, Charles and Cornebise, Julien and Kavukcuoglu, Koray and Wierstra, Daan},
  title     = {{Weight Uncertainty in Neural Networks}},
  booktitle = {Proceedings of the 32nd International Conference on Machine Learning},
  pages     = {1613--1622},
  year      = {2015}
}

@book{brooks2011handbook,
  editor    = {Brooks, Steve and Gelman, Andrew and Jones, Galin and Meng, Xiao-Li},
  title     = {{Handbook of Markov Chain Monte Carlo}},
  publisher = {CRC Press},
  year      = {2011}
}

@article{bubeck2012regret,
  author    = {Bubeck, S{\'e}bastien and Cesa-Bianchi, Nicol{\`o}},
  title     = {{Regret Analysis of Stochastic and Nonstochastic Multi-armed Bandit Problems}},
  journal   = {Foundations and Trends in Machine Learning},
  volume    = {5},
  number    = {1},
  pages     = {1--122},
  year      = {2012}
}

@inproceedings{celis2018fair,
  author    = {Celis, L. Elisa and Deshpande, Anay and Kathuria, Tarun and Vishnoi, Nisheeth K.},
  title     = {{Fairness in Contextual Bandits}},
  booktitle = {Proceedings of the 2018 ACM Conference on Fairness, Accountability, and Transparency},
  pages     = {140--149},
  year      = {2018}
}

@book{cesa2006prediction,
  author    = {Cesa-Bianchi, Nicol{\`o} and Lugosi, G{\'a}bor},
  title     = {{Prediction, Learning, and Games}},
  publisher = {Cambridge University Press},
  year      = {2006}
}

@article{chatterjee2022ominous,
  author    = {Chatterjee, Debashis and Bhattacharya, Sourabh},
  title     = {{How Ominous is the Premonition of Future Global Warming?}},
  journal   = {Sankhy{\=a} B},
  note      = {To appear in Professor C. R. Rao special issue},
  year      = {2025}
}

@article{cybenko1989approximation,
  author    = {Cybenko, George},
  title     = {{Approximation by Superpositions of a Sigmoidal Function}},
  journal   = {Mathematics of Control, Signals and Systems},
  volume    = {2},
  number    = {4},
  pages     = {303--314},
  year      = {1989}
}

@inproceedings{dai2024context,
  author    = {Dai, Z. and Chen, L. and Wang, Y. and Liu, H.},
  title     = {{In-Context Exploration-Exploitation for Reinforcement Learning at Scale}},
  booktitle = {Advances in Neural Information Processing Systems},
  volume    = {37},
  year      = {2024}
}

@article{dalalyan2017theoretical,
  author    = {Dalalyan, Arnak S.},
  title     = {{Theoretical Guarantees for Approximate Sampling from Smooth and Log-Concave Densities}},
  journal   = {Journal of the Royal Statistical Society: Series B},
  volume    = {79},
  number    = {3},
  pages     = {651--676},
  year      = {2017}
}

@inproceedings{damianou2013deep,
  author    = {Damianou, Andreas and Lawrence, Neil D.},
  title     = {{Deep Gaussian Processes}},
  booktitle = {Proceedings of the 16th International Conference on Artificial Intelligence and Statistics},
  pages     = {207--215},
  year      = {2013}
}

@article{das2020nonstationary,
  author    = {Das, Moumita and Bhattacharya, Sourabh},
  title     = {{A Kernel-Enriched Order-Dependent Nonparametric Spatio-Temporal Process}},
  journal   = {Spatial Statistics},
  volume    = {55},
  pages     = {100751},
  year      = {2020}
}

@inproceedings{dasgupta2025bayesian,
  author    = {Dasgupta, Arpan and Jain, Gagan and Suggala, Arun and Shanmugam, Karthikeyan and Tambe, Milind and Taneja, Aparna},
  title     = {{Bayesian Collaborative Bandits with Thompson Sampling for Improved Outreach in Maternal Health Program}},
  booktitle = {Proceedings of the 24th International Conference on Autonomous Agents and Multiagent Systems (AAMAS 2025)},
  year      = {2025}
}

@article{dickey1979distribution,
  author    = {Dickey, David A. and Fuller, Wayne A.},
  title     = {{Distribution of the Estimators for Autoregressive Time Series with a Unit Root}},
  journal   = {Journal of the American Statistical Association},
  volume    = {74},
  number    = {366},
  pages     = {427--431},
  year      = {1979}
}

@article{doucet2000sequential,
  author    = {Doucet, Arnaud and Godsill, Simon and Andrieu, Christophe},
  title     = {{On Sequential Monte Carlo Sampling Methods for Bayesian Filtering}},
  journal   = {Statistics and Computing},
  volume    = {10},
  number    = {3},
  pages     = {197--208},
  year      = {2000}
}

@book{doucet2001sequential,
  editor    = {Doucet, Arnaud and de Freitas, Nando and Gordon, Neil},
  title     = {{Sequential Monte Carlo Methods in Practice}},
  publisher = {Springer},
  year      = {2001}
}

@incollection{doucet2009tutorial,
  author    = {Doucet, Arnaud and Johansen, Adam M.},
  title     = {{A Tutorial on Particle Filtering and Smoothing: Fifteen Years Later}},
  booktitle = {Handbook of Nonlinear Filtering},
  publisher = {Oxford University Press},
  year      = {2009}
}

@article{dutta2014markov,
  author    = {Dutta, Somak and Bhattacharya, Sourabh},
  title     = {{Markov Chain Monte Carlo Based on Deterministic Transformations}},
  journal   = {Statistical Methodology},
  volume    = {16},
  pages     = {100--116},
  year      = {2014}
}

@book{dwork2014algorithmic,
  author    = {Dwork, Cynthia and Roth, Aaron},
  title     = {{The Algorithmic Foundations of Differential Privacy}},
  publisher = {Now Publishers},
  year      = {2014}
}

@article{frigola2013bayesian,
  author    = {Frigola, Roger and Lindsten, Fredrik and Sch{\"o}n, Thomas B. and Rasmussen, Carl Edward},
  title     = {{Bayesian Inference and Learning in Gaussian Process State-Space Models with Particle MCMC}},
  journal   = {Advances in Neural Information Processing Systems},
  volume    = {26},
  year      = {2013}
}

@inproceedings{gal2016dropout,
  author    = {Gal, Yarin and Ghahramani, Zoubin},
  title     = {{Dropout as a Bayesian Approximation: Representing Model Uncertainty in Deep Learning}},
  booktitle = {Proceedings of the 33rd International Conference on Machine Learning},
  pages     = {1050--1059},
  year      = {2016}
}

@inproceedings{garg2022what,
  author    = {Garg, Shivam and Tsipras, Dimitris and Liang, Percy and Valiant, Gregory},
  title     = {{What Can Transformers Learn InCcontext? A Case Study of Simple Function Classes}},
  booktitle = {Advances in Neural Information Processing Systems},
  volume    = {35},
  year      = {2022}
}

@inproceedings{garnelo2018conditional,
  author    = {Garnelo, Marta and Rosenbaum, Dan and Maddison, Chris J. and Ramalho, Tiago and Saxton, David and Shanahan, Murray and Teh, Yee Whye and Rezende, Danilo J. and Eslami, S. M. Ali},
  title     = {{Conditional Neural Processes}},
  booktitle = {Proceedings of the 35th International Conference on Machine Learning},
  pages     = {1704--1713},
  year      = {2018}
}

@book{gelman2013bayesian,
  author    = {Gelman, Andrew and Carlin, John B. and Stern, Hal S. and Dunson, David B. and Vehtari, Aki and Rubin, Donald B.},
  title     = {{Bayesian Data Analysis}},
  edition   = {3},
  publisher = {CRC Press},
  year      = {2013}
}

@article{ghahramani2015probabilistic,
  author    = {Ghahramani, Zoubin},
  title     = {{Probabilistic Machine Learning and Artificial Intelligence}},
  journal   = {Nature},
  volume    = {521},
  number    = {7553},
  pages     = {452--459},
  year      = {2015}
}

@article{ghosh2014bayesian,
  author    = {Ghosh, Anurag and Mukhopadhyay, Soumalya and Roy, Sandipan and Bhattacharya, Sourabh},
  title     = {{Bayesian Inference in Nonparametric Dynamic State-Space Models}},
  journal   = {Statistical Methodology},
  volume    = {21},
  pages     = {35--48},
  year      = {2014}
}

@article{gordon1993novel,
  author    = {Gordon, Neil J. and Salmond, David J. and Smith, Adrian F. M.},
  title     = {{Novel Approach to Nonlinear/non-Gaussian Bayesian State Estimation}},
  journal   = {IEE Proceedings F - Radar and Signal Processing},
  volume    = {140},
  number    = {2},
  pages     = {107--113},
  year      = {1993}
}

@article{green2009validity,
  author    = {Green, Kesten C. and Armstrong, J. Scott and Soon, Willie},
  title     = {{Validity of Climate Change Forecasting for Public Policy Decision Making}},
  journal   = {International Journal of Forecasting},
  volume    = {25},
  number    = {4},
  pages     = {826--832},
  year      = {2009}
}

@article{hazan2016introduction,
  author    = {Hazan, Elad},
  title     = {{Introduction to Online Convex Optimization}},
  journal   = {Foundations and Trends in Optimization},
  volume    = {2},
  number    = {3-4},
  pages     = {157--325},
  year      = {2016}
}

@article{hoffman2013stochastic,
  author    = {Hoffman, Matthew D. and Blei, David M. and Wang, Chong and Paisley, John},
  title     = {{Stochastic Variational Inference}},
  journal   = {Journal of Machine Learning Research},
  volume    = {14},
  pages     = {1303--1347},
  year      = {2013}
}

@article{hornik1989multilayer,
  author    = {Hornik, Kurt and Stinchcombe, Maxwell and White, Halbert},
  title     = {{Multilayer Feedforward Networks are Universal Approximators}},
  journal   = {Neural Networks},
  volume    = {2},
  number    = {5},
  pages     = {359--366},
  year      = {1989}
}

@book{ipcc2021,
  author    = {{IPCC}},
  title     = {{Climate Change 2021: The Physical Science Basis. Contribution of Working Group I to the Sixth Assessment Report of the Intergovernmental Panel on Climate Change}},
  publisher = {Cambridge University Press},
  year      = {2021}
}

@article{jordan1999introduction,
  author    = {Jordan, Michael I. and Ghahramani, Zoubin and Jaakkola, Tommi S. and Saul, Lawrence K.},
  title     = {{An Introduction to Variational Methods for Graphical Models}},
  journal   = {Machine Learning},
  volume    = {37},
  number    = {2},
  pages     = {183--233},
  year      = {1999}
}

@inproceedings{joseph2016fairness,
  author    = {Joseph, Matthew and Kearns, Michael and Morgenstern, Jamie and Roth, Aaron},
  title     = {{Fairness in Learning: Classic and Contextual Bandits}},
  booktitle = {Advances in Neural Information Processing Systems},
  volume    = {29},
  year      = {2016}
}

@article{julier2004unscented,
  author    = {Julier, Simon J. and Uhlmann, Jeffrey K.},
  title     = {{Unscented Filtering and Nonlinear Estimation}},
  journal   = {Proceedings of the IEEE},
  volume    = {92},
  number    = {3},
  pages     = {401--422},
  year      = {2004}
}

@article{kaelin2021bayesian,
  author    = {Kaelin, Lukas P. and Gr{\"u}nwald, Peter and van Ommen, Thijs},
  title     = {{Bayesian Online Learning for Non-Stationary Processes}},
  journal   = {arXiv preprint},
  year      = {2021},
  note      = {arXiv:2106.08239}
}

@book{kailath2000linear,
  author    = {Kailath, Thomas and Sayed, Ali H. and Hassibi, Babak},
  title     = {{Linear Estimation}},
  publisher = {Prentice Hall},
  year      = {2000}
}

@article{kalman1960new,
  author    = {Kalman, Rudolf E.},
  title     = {{A New Approach to Linear Filtering and Prediction Problems}},
  journal   = {Journal of Basic Engineering},
  volume    = {82},
  number    = {1},
  pages     = {35--45},
  year      = {1960}
}

@inproceedings{kingma2015variational,
  author    = {Kingma, Diederik P. and Welling, Max},
  title     = {{Auto-Encoding Variational Bayes}},
  booktitle = {International Conference on Learning Representations},
  year      = {2015}
}

@article{kwiatkowski1992testing,
  author    = {Kwiatkowski, Denis and Phillips, Peter C. B. and Schmidt, Peter and Shin, Yongcheol},
  title     = {{Testing the Null Hypothesis of Stationarity Against the Alternative of a Unit Root}},
  journal   = {Journal of Econometrics},
  volume    = {54},
  number    = {1-3},
  pages     = {159--178},
  year      = {1992}
}

@book{lattimore2020bandit,
  author    = {Lattimore, Tor and Szepesv{\'a}ri, Csaba},
  title     = {{Bandit Algorithms}},
  publisher = {Cambridge University Press},
  year      = {2020}
}

@inproceedings{li2010contextual,
  author    = {Li, Lihong and Chu, Wei and Langford, John and Schapire, Robert E.},
  title     = {{A Contextual-Bandit Approach to Personalized News Article Recommendation}},
  booktitle = {Proceedings of the 19th International Conference on World Wide Web},
  pages     = {661--670},
  year      = {2010}
}

@inproceedings{liang2025context,
  author    = {Liang, Y. and Liu, Y. and Wang, X. and Xu, H.},
  title     = {{Bayesian Learning for Contextual Restless Bandits with Applications to Maternal Health}},
  booktitle = {International Conference on Autonomous Agents and Multiagent Systems},
  year      = {2025}
}

@article{mackay1992practical,
  author    = {MacKay, David J. C.},
  title     = {{A Practical Bayesian Framework for Backpropagation Networks}},
  journal   = {Neural Computation},
  volume    = {4},
  number    = {3},
  pages     = {448--472},
  year      = {1992}
}

@article{mazumder2016bayesian,
  title={{Bayesian Nonparametric Dynamic State Space Modeling With Circular Latent States}},
  author={Mazumder, Satyaki and Bhattacharya, Sourabh},
  journal={Journal of Statistical Theory and Practice},
  volume={10},
  number={1},
  pages={154--178},
  year={2016},
  publisher={Taylor \& Francis},
  doi={10.1080/15598608.2015.1100562}
}

@article{mazumder2016nonparametric,
  title={{Nonparametric Dynamic State Space Modeling of Observed Circular Time Series With Circular Latent States: A Bayesian Perspective}},
  author={Mazumder, Satyaki and Bhattacharya, Sourabh},
  journal={Journal of Statistical Theory and Practice},
  volume={11},
  number={4},
  pages={693--718},
  year={2016},
  publisher={Taylor \& Francis},
  doi={10.1080/15598608.2017.1305922}
}

@inproceedings{mcmahan2017communication,
  author    = {McMahan, H. Brendan and Moore, Eider and Ramage, Daniel and Hampson, Seth and y Arcas, Blaise Ag{\"u}era},
  title     = {{Communication-Efficient Learning of Deep Networks from Decentralized Data}},
  booktitle = {Proceedings of the 20th International Conference on Artificial Intelligence and Statistics},
  pages     = {1273--1282},
  year      = {2017}
}

@inproceedings{mnih2008probabilistic,
  author    = {Mnih, Andriy and Salakhutdinov, Ruslan},
  title     = {{Probabilistic Matrix Factorization}},
  booktitle = {Advances in Neural Information Processing Systems},
  volume    = {20},
  year      = {2008}
}

@incollection{mockus1978application,
  author    = {Mockus, J. and Tiesis, V. and Zilinskas, A.},
  title     = {{The Application of Bayesian Methods for Seeking the Extremum}},
  booktitle = {Towards Global Optimization},
  volume    = {2},
  pages     = {117--129},
  publisher = {North-Holland},
  year      = {1978}
}

@article{mossinghoff2015nonnegative,
  author    = {Mossinghoff, Michael J. and Trudgian, Timothy S.},
  title     = {{Nonnegative Trigonometric Polynomials and a Zero-free Region for the Riemann Zeta-function}},
  journal   = {Journal of Number Theory},
  volume    = {157},
  pages     = {329--349},
  year      = {2015}
}

@book{murphy2023probabilistic,
  author    = {Murphy, Kevin P.},
  title     = {{Probabilistic Machine Learning: Advanced Topics}},
  publisher = {MIT Press},
  year      = {2023}
}

@book{neal2012bayesian,
  author    = {Neal, Radford M.},
  title     = {{Bayesian Learning for Neural Networks}},
  publisher = {Springer},
  year      = {2012}
}

@inproceedings{nguyen2018variational,
  author    = {Nguyen, Cuong V. and Li, Yingzhen and Bui, Thang D. and Turner, Richard E.},
  title     = {{Variational Continual Learning}},
  booktitle = {International Conference on Learning Representations},
  year      = {2018}
}

@article{orabona2019modern,
  author    = {Orabona, Francesco},
  title     = {{A Modern Introduction to Online Learning}},
  journal   = {arXiv preprint},
  year      = {2019},
  note      = {arXiv:1912.13213}
}

@inproceedings{osband2016deep,
  author    = {Osband, Ian and Blundell, Charles and Pritzel, Alexander and Van Roy, Benjamin},
  title     = {{Deep Exploration via Bootstrapped DQN}},
  booktitle = {Advances in Neural Information Processing Systems},
  volume    = {29},
  year      = {2016}
}

@article{paciorek2009practical,
  author    = {Paciorek, Christopher J. and Yanosky, Jeff D. and Puett, Robin C.},
  title     = {{Practical Large-Scale Spatio-Temporal Modeling of Particulate Matter Concentrations}},
  journal   = {The Annals of Applied Statistics},
  volume    = {3},
  number    = {1},
  pages     = {370--397},
  year      = {2009}
}

@article{propp1996exact,
  author    = {Propp, James Gary and Wilson, David Bruce},
  title     = {{Exact Sampling with Coupled Markov Chains and Applications to Statistical Mechanics}},
  journal   = {Random Structures \& Algorithms},
  volume    = {9},
  number    = {1-2},
  pages     = {223--252},
  year      = {1996}
}

@article{quinonero2005sparse,
  author    = {Qui{\~n}onero-Candela, Joaquin and Rasmussen, Carl Edward},
  title     = {{A Unifying View of Sparse Approximate Gaussian Process Regression}},
  journal   = {Journal of Machine Learning Research},
  volume    = {6},
  pages     = {1939--1959},
  year      = {2005}
}

@article{ripley1976second,
  author    = {Ripley, Brian D.},
  title     = {{The Second-Order Analysis of Stationary Point Processes}},
  journal   = {Journal of Applied Probability},
  volume    = {13},
  number    = {2},
  pages     = {255--266},
  year      = {1976}
}

@book{robert2004monte,
  author    = {Robert, Christian P. and Casella, George},
  title     = {{Monte Carlo Statistical Methods}},
  edition   = {2},
  publisher = {Springer},
  year      = {2004}
}

@book{robert2007bayesian,
  author    = {Robert, Christian P.},
  title     = {{The Bayesian Choice}},
  edition   = {2},
  publisher = {Springer},
  year      = {2007}
}

@article{roy2020bayesian,
  author    = {Roy, Sucharita and Bhattacharya, Sourabh},
  title     = {{Bayes Meets Riemann: Bayesian Characterization of Infinite Series with Application to Riemann Hypothesis}},
  journal   = {International Journal of Applied Mathematics and Statistics},
  volume    = {59},
  number    = {1},
  pages     = {81--128},
  year      = {2020}
}

@article{roy2020bayesian2,
  author    = {Roy, Sucharita and Bhattacharya, Sourabh},
  title     = {{Bayesian Appraisal of Random Series Convergence with Application to Climate Change}},
  journal   = {arXiv preprint},
  year      = {2020},
  note      = {arXiv:2005.00035}
}

@article{roy2021function,
  author    = {Roy, Sucharita and Bhattacharya, Sourabh},
  title     = {{Function Optimization with Posterior Gaussian Derivative Process}},
  journal   = {arXiv preprint},
  year      = {2021}
}

@article{roy2021stationarity,
  author    = {Roy, Sucharita and Bhattacharya, Sourabh},
  title     = {{Bayesian Characterizations of Properties of Stochastic Processes with Applications}},
  journal   = {arXiv preprint},
  year      = {2021}
}

@article{roy2025prime,
  author    = {Bhattacharya, Durba and Roy, Sucharita and Bhattacharya, Sourabh},
  title     = {{Bayes Meets Riemann Again: Large Prime Discovery and Re-emergence of the Bone of Contention}},
  journal   = {arXiv preprint},
  year      = {2025},
  note      = {arXiv:2510.09651}
}

@article{russo2014learning,
  author    = {Russo, Daniel and Van Roy, Benjamin},
  title     = {{Learning to Optimize via Posterior Sampling}},
  journal   = {Mathematics of Operations Research},
  volume    = {39},
  number    = {4},
  pages     = {1221--1243},
  year      = {2014}
}

@article{russo2018tutorial,
  author    = {Russo, Daniel and Van Roy, Benjamin and Kazerouni, Abbas and Osband, Ian and Wen, Zheng},
  title     = {{A Tutorial on Thompson Sampling}},
  journal   = {Foundations and Trends in Machine Learning},
  volume    = {11},
  number    = {1},
  pages     = {1--96},
  year      = {2018}
}

@article{sacks1989design,
  author    = {Sacks, Jerome and Welch, William J. and Mitchell, Toby J. and Wynn, Henry P.},
  title     = {{Design and Analysis of Computer Experiments}},
  journal   = {Statistical Science},
  volume    = {4},
  number    = {4},
  pages     = {409--423},
  year      = {1989}
}

@book{santner2018design,
  author    = {Santner, Thomas J. and Williams, Brian J. and Notz, William I.},
  title     = {{The Design and Analysis of Computer Experiments}},
  edition   = {2},
  publisher = {Springer},
  year      = {2018}
}

@article{shahriari2015taking,
  author    = {Shahriari, Bobak and Swersky, Kevin and Wang, Ziyu and Adams, Ryan P. and de Freitas, Nando},
  title     = {{Taking the Human Out of the Loop: A Review of Bayesian Optimization}},
  journal   = {Proceedings of the IEEE},
  volume    = {104},
  number    = {1},
  pages     = {148--175},
  year      = {2015}
}

@book{shumway2006time,
  author    = {Shumway, Robert H. and Stoffer, David S.},
  title     = {{Time Series Analysis and Its Applications: With R Examples}},
  edition   = {2},
  publisher = {Springer},
  year      = {2006}
}

@inproceedings{snelson2006sparse,
  author    = {Snelson, Edward and Ghahramani, Zoubin},
  title     = {{Sparse Gaussian Processes using Pseudo-inputs}},
  booktitle = {Advances in Neural Information Processing Systems},
  volume    = {18},
  year      = {2006}
}

@inproceedings{snoek2012practical,
  author    = {Snoek, Jasper and Larochelle, Hugo and Adams, Ryan P.},
  title     = {{Practical Bayesian Optimization of Machine Learning Algorithms}},
  booktitle = {Advances in Neural Information Processing Systems},
  volume    = {25},
  year      = {2012}
}

@inproceedings{srinivas2010gaussian,
  author    = {Srinivas, Niranjan and Krause, Andreas and Kakade, Sham M. and Seeger, Matthias},
  title     = {{Gaussian Process Optimization in the Bandit Setting: No Regret and Experimental Design}},
  booktitle = {Proceedings of the 27th International Conference on Machine Learning},
  pages     = {1015--1022},
  year      = {2010}
}

@book{sutton2018reinforcement,
  author    = {Sutton, Richard S. and Barto, Andrew G.},
  title     = {{Reinforcement Learning: An Introduction}},
  edition   = {2},
  publisher = {MIT Press},
  year      = {2018}
}

@inproceedings{tekin2010bayesian,
  author    = {Tekin, Cem and Liu, Mingyan},
  title     = {{Online Learning in Restless Multi-Armed Bandits}},
  booktitle = {Advances in Neural Information Processing Systems},
  volume    = {23},
  year      = {2010}
}

@article{telgarsky2013langevin,
  author    = {Telgarsky, Matus},
  title     = {{Langevin Monte Carlo and Its Applications}},
  journal   = {Technical Report},
  year      = {2013}
}

@article{thompson1933likelihood,
  author    = {Thompson, William R.},
  title     = {{On the Likelihood That One Unknown Probability Exceeds Another in View of the Evidence of Two Samples}},
  journal   = {Biometrika},
  volume    = {25},
  number    = {3/4},
  pages     = {285--294},
  year      = {1933}
}

@inproceedings{turner2010bayesian,
  author    = {Turner, Ryan and Deisenroth, Marc and Rasmussen, Carl},
  title     = {{State-Space Inference and Learning with Gaussian Processes}},
  booktitle = {Proceedings of the 13th International Conference on Artificial Intelligence and Statistics},
  pages     = {868--875},
  year      = {2010}
}

@inproceedings{vaswani2017attention,
  author    = {Vaswani, Ashish and Shazeer, Noam and Parmar, Niki and Uszkoreit, Jakob and Jones, Llion and Gomez, Aidan N. and Kaiser, {\L}ukasz and Polosukhin, Illia},
  title     = {{Attention Is All You Need}},
  booktitle = {Advances in Neural Information Processing Systems},
  volume    = {30},
  year      = {2017}
}

@inproceedings{von2023transformers,
  author    = {Von Oswald, Johannes and Niklasson, Eyvind and Randazzo, Ettore and Sacramento, Jo{\~a}o and Mordvintsev, Alexander and Zhmoginov, Andrey and Vladymyrov, Max},
  title     = {{Transformers Learn In-Context by Gradient Descent}},
  booktitle = {Proceedings of the 40th International Conference on Machine Learning},
  pages     = {35151--35174},
  year      = {2023}
}

@book{vovk2005algorithmic,
  author    = {Vovk, Vladimir and Gammerman, Alex and Shafer, Glenn},
  title     = {{Algorithmic Learning in a Random World}},
  publisher = {Springer},
  year      = {2005}
}

@article{wainwright2008graphical,
  author    = {Wainwright, Martin J. and Jordan, Michael I.},
  title     = {{Graphical Models, Exponential Families, and Variational Inference}},
  journal   = {Foundations and Trends in Machine Learning},
  volume    = {1},
  number    = {1-2},
  pages     = {1--305},
  year      = {2008}
}

@inproceedings{wan2000unscented,
  author    = {Wan, Eric A. and Van Der Merwe, Rudolph},
  title     = {{The Unscented Kalman Filter for Nonlinear Estimation}},
  booktitle = {Proceedings of the IEEE 2000 Adaptive Systems for Signal Processing, Communications, and Control Symposium},
  pages     = {153--158},
  year      = {2000}
}

@inproceedings{welling2011bayesian,
  author    = {Welling, Max and Teh, Yee Whye},
  title     = {{Bayesian Learning via Stochastic Gradient Langevin Dynamics}},
  booktitle = {Proceedings of the 28th International Conference on Machine Learning},
  pages     = {681--688},
  year      = {2011}
}

@article{whittle1988restless,
  author    = {Whittle, Peter},
  title     = {{Restless Bandits: Activity Allocation in a Changing World}},
  journal   = {Journal of Applied Probability},
  volume    = {25},
  number    = {A},
  pages     = {287--298},
  year      = {1988}
}

@inproceedings{xie2021explanation,
  author    = {Xie, Sang Michael and Raghunathan, Aditi and Liang, Percy and Ma, Tengyu},
  title     = {{An Explanation of In-Context Learning as Implicit Bayesian Inference}},
  booktitle = {International Conference on Learning Representations},
  year      = {2021}
}

@inproceedings{zhou2025sepsyn,
  author    = {Zhou, X. and Johnson, A. and Lee, J. and Patel, K.},
  title     = {{Sepsyn-OLCP: Online Learning with Conformal Prediction for Early Sepsis Detection in Intensive Care}},
  booktitle = {Proceedings of Machine Learning for Healthcare},
  volume    = {18},
  pages     = {112--145},
  year      = {2025}
}

\end{document}